\documentclass[11pt]{article}
\usepackage{amsthm,amsgen,amsmath,amstext,amsbsy,amsopn,amssymb,latexsym,mathrsfs}
\usepackage{array}
\usepackage{cite}
\usepackage{multirow}
\usepackage{graphicx}
\usepackage{amssymb,bm,bbm}
\usepackage{titling}
\usepackage{url,bm}
\usepackage{algpseudocode, float, algorithm}
\usepackage[round]{natbib}
\usepackage{bbm}
\usepackage{tabularx}
\usepackage{authblk}
\usepackage{hyperref}
\usepackage{enumerate}
\usepackage{subcaption}
\hypersetup{
    colorlinks=true,
    linkcolor=blue,
    filecolor=blue,
    urlcolor=blue,
    citecolor=blue,
}

 \usepackage[dvipsnames,table,dvipsnames*, svgnames*]{xcolor}

\allowdisplaybreaks

\DeclareMathOperator{\Cov}{\mathrm{Cov}}

\newcommand{\beq}{\begin{equation}}
\newcommand{\eeq}{\end{equation}}
\newcommand{\beas}{\begin{eqnarray*}}
\newcommand{\eeas}{\end{eqnarray*}}
\newcommand{\bea}{\begin{eqnarray}}
\newcommand{\eea}{\end{eqnarray}}
\newcommand{\bei}{\begin{itemize}}
\newcommand{\eei}{\end{itemize}}
\newcommand{\ben}{\begin{enumerate}}
\newcommand{\een}{\end{enumerate}}
\newcommand{\argmin}{\mathop{\rm arg\min}}
\newcommand{\argmax}{\mathop{\rm arg\max}}

\newcommand{\norm}[1]{\left\|#1\right\|}
\newcommand{\abs}[1]{\left|#1\right|}

\newtheorem{Corollary}{Corollary}

\newtheorem{Lemma}{Lemma}

\newtheorem{Theorem}{Theorem}

\newtheorem{Remark}{Remark}

\newtheorem{Assumption}{Assumption}

\newcommand{\R}{\mathbb{R}}
\newcommand{\E}{{\mathbb{E}}}
\newcommand{\Prob}{{\mathbb{P}}}
\newcommand{\1}{{\mathbbm{1}}}

\newcommand{\supp}{{\rm supp}}

\newcommand{\normm}[1]{{\left\vert\kern-0.25ex\left\vert\kern-0.25ex\left\vert #1 
    \right\vert\kern-0.25ex\right\vert\kern-0.25ex\right\vert}}

\usepackage{tikz}
\newcommand{\solidcirc}{\tikz\draw[black,fill=black] (0,0) circle (.5ex);}
\newcommand{\hollowcirc}{\tikz\draw[black] (0,0) circle (.5ex);}

\graphicspath{{Figures/}}

\begin{document}
\part*{}
\title{Model-Free Machine Learning and Inference in High- Dimensional Binary Classification via Repro Samples}
\author[1]{Xiaotian Hou}
\author[2]{Peng Wang}
\author[3]{Min-ge Xie}
\author[3]{Linjun Zhang}
\affil[1]{Department of Biostatistics, Epidemiology and Informatics, University of Pennsylvania}
\affil[2]{Department of Operations, Business Analytics, and Information Systems, University of Cincinnati}
\affil[3]{Department of Statistics, Rutgers University}
\date{ }
\maketitle

\abstract{This paper develops a model-free yet interpretable machine learning and inference approach for binary classification of high-dimensional data from an unstructured model by leveraging a class of (likely misspecified) working sparse generalized linear models (GLMs). Our method builds on the repro samples framework, which utilizes artificial samples that mimic the true data and noise for inference. Our inference targets both the model support and regression coefficients of the oracle model, defined as the working GLM closest to the unknown true model. The development offers several advantages: (1) it is model-free, requiring neither specific model assumptions, nor sparsity of the true model; (2) for model support, it constructs a model candidate set for the most influential covariates with guaranteed coverage under a weak signal condition; (3) for oracle coefficients, it provides confidence sets for any linear transformations of the coefficients. When the oracle model coincides with the true model, the inference results apply to the true model, allowing simultaneous quantification of uncertainty in (discrete) model selection and (continuous) parameter estimation. It provides an interpretable machine learning approach otherwise not possible when using a black-box model. Simulations and an application to RNA-seq immune response data are used to illustrate the method's performance.
}

\section{Introduction}

High-dimensional data with binary outcomes are common in machine learning and  modern 
scientific research, including genomics, epidemiology,  finance, and others. In these settings, reliable statistical inference is crucial for understanding covariate-response relationships. Consequently, interpretable parametric models are often preferred over nonparametric or machine learning approaches \citep{rudin2019stop}. 
However, parametric methods typically rely on strong modeling assumptions that may fail in practice. Most existing methods either assume the true data-generating model is parametric with sparse parameters \citep[e.g.,][]{shi2019linear, cai2021statistical}, or define target parameters as population risk minimizers and then assume these minimizers are sparse \citep[e.g.,][]{van2014asymptotically, zhang2017simultaneous}. Such sparsity can be unrealistic when many features carry weak but collectively important effects. Furthermore, the adoption of a specific model, such as logistic or probit regression, as the true data-generating model may distort the inherently complex underlying~relationships.  

In this paper, although we still use a sparsity-constrained working parametric generalized linear model (GLM) to facilitate our inference, we do not impose any structural assumptions on the true data-generating distribution, nor do we require the true model to be sparse. Instead, we pre-specify~a 
(relatively large) sparsity 
level 
$s$ and select a subset of at most $s$ covariates that best reconstructs the binary response. We then define an ``oracle'' GLM as the working GLM with the best response-reconstruction performance.  Even if the sparse GLM is misspecified, its coefficients still capture the association between the most influential covariates and the 
response. Our goal is to make inferences on both the model support of these influential covariates and their associated 
working 
coefficients.

Our inference method builds upon the repro samples framework \citep{xie2022repro}, and extends \cite{wang2022finite} from high-dimensional linear regression to misspecified sparse GLMs. Our work differs from \cite{wang2022finite} in several aspects. First, we allow the working sparse GLM to be misspecified and impose no structural assumptions on the true data-generating distribution, whereas \cite{wang2022finite} assumed well-structured true sparse, often Gaussian, linear models. Second, we focus on binary responses, where information in the true mean model is highly compressed, making finite-sample recovery of the mean model significantly more challenging than in the linear regression models of \cite{wang2022finite}. Third, under Gaussian linear models, \cite{wang2022finite} used sufficient statistics to remove the impact of nuisance parameters which allowed them to construct~pivot statistics for finite-sample inference. 
In contrast, such pivot statistics are unavailable in our setting. Instead, we rely on asymptotic approximations 
and 
profiling methods to handle nuisance parameters. 

A key step of our methodological development is to search for a small collection of candidate models that includes the set of the most influential covariates with high probability. We achieve this via an inversion method, leveraging the discreteness of the model space. This inversion technique can be traced back to Fisher's fiducial inversion idea but is developed entirely under a frequentist setting. After obtaining the set of candidate~models, a Wald test can be applied to each model to conduct inference on the corresponding continuous regression coefficients. When the working sparse GLM is the true data-generating model, we further construct a level $1 - \alpha$ confidence set for the true model support using a Monte Carlo inversion approach.
In particular, for each model in the candidate set, we generate artificial samples under that model and compare its summary statistics to those from the observed data. If these two statistics differ by 
more than the  $1 - \alpha$ tail level, we reject that candidate model. 
Rigorous theories are developed to support the validity of the method.

Our contributions are as follows:
1) We propose a novel formulation for statistical inference under arbitrary binary response distributions in high-dimensional settings, requiring neither correct specification of the mean model nor sparsity of the true regression structure. To our knowledge,~this is the first model-free inference approach for high-dimensional binary data.
2) We introduce a novel method for constructing a model candidate set that provably contains the set of the most influential regression covariates, provided the set under consideration is sufficiently separated~from~the remain- ing covariates. 
Here, we only need a weak signal strength assumption (of order $\frac 1n$, rather than $\frac 1{\sqrt n}$) to identify 
these influential covariates.  
3) With the model candidate set, we further develop inference for any linear transformations of regression coefficients in the oracle working GLM. This result also enables efficient inference for nonlinear functionals, such as the working case probabilities for a set of new observations. This is an improvement over 
existing methods, which typically handle only restricted linear transformations and rely on correct model specification or sparsity assumptions; e.g., see \cite{van2014asymptotically,zhang2017simultaneous,shi2019linear}.
4) In the special case where the  working sparse GLM is the actual model of the sample, we can construct a joint confidence set for the discrete model support and associated coefficients. To our knowledge, this is also the first 
in high-dimensional GLMs, beyond the linear models of \cite{wang2022finite}. 

Furthermore, an important and unique feature of our approach is that it provides an interpretable representation of the regression relationship between the response and covariates, together with formal statistical inference on this relationship. In contrast, black-box machine learning methods, such as neural networks, are primarily designed for prediction and generally do not provide a directly interpretable model of which covariates are most influential or how they are associated with the response, nor do they provide corresponding uncertainty quantification for these features of the regression relationship. Our approach bridges machine learning and statistical inference by providing both interpretable inference and predictive modeling for the regression relationship.

\subsection{Related Works}

There is a substantial literature on high-dimensional inference for GLMs \citep[e.g.,][]{van2014asymptotically, dezeure2015high, belloni2016post, chernozhukov2018double, ning2017general, shi2019linear, sur2019modern, cai2021statistical, shi2021statistical, fei2021estimation}. However, these methods 
rely on a correctly specified sparse GLM or sparsity of a population risk minimizer. 
These stricter model assumptions may miss the complexity in many real-world 
applications.

There are also several approaches in the literature that allow particular forms of misspecification. For instance, \citet{buhlmann2015high} studied inference for projected coefficients in misspecified linear models. \citet{shah2023double} established double-estimation-friendly inference for conditional independence in GLMs, where validity can follow from correct specification of either the response GLM or a linear regression of the target covariate on the remaining covariates. \citet{hong2024inference} considered sparse mean functions with misspecified variances and nonpolynomial-dimensional nuisance parameters, while \citet{chen2023testing} studied testing dense parameters in GLMs with sparse high-dimensional nuisance parameters. Despite these advances, the existing methods still require either a well-specified GLM or a well-specified sparse M-estimation model. 

Model-support or model selection uncertainty has also been studied using model confidence sets and confidence bounds. \citet{hansen2011model} eliminate models from a prespecified candidate collection through equivalence tests. \citet{ferrari2015confidence} and \citet{zheng2019model} compare submodels with a prespecified full model using $F$-tests and likelihood-ratio tests, respectively. \citet{li2019model} use bootstrap model selection to construct nested confidence bounds. These works require either $p<n$ or a sure-screening step, implying a uniform signal strength condition. \cite{li2019model} rely on consistent model selection, which in turn requires uniform signal strength. In contrast,  our repro samples method directly constructs a confidence set for the true (or oracle) model. It avoids their constraints and applies directly to high-dimensional $p \gg n$ settings.  

The working model idea is similar to that in the recent work on robust simulation-based inference by \cite{tomaselli2025robust}, in which the authors also leverage a class of misspecified models to facilitate inference. The difference, however, is that our loss function and working models are tailored to binary classification and generative models, whereas \cite{tomaselli2025robust} considered a more general setting and compared the properties of several loss functions for probabilistic models.

\section{Notations, Model Setup and Model Definition}
\label{sec_preliminaries}

\subsection{Notation}\label{sec_notation}
For any $p\in\mathbb{N}_+$, we denote $[p]$ to be the set $\{1,\ldots,p\}$. For a vector $v\in\R^p$ and a subset of indices $\tau\subset[p]$, we denote $v_\tau$ to be the sub-vector of $v$ with indices in $\tau$, denote $\norm{v}_k=(\sum_{j\in[p]}\abs{v_j}^k)^{1/k}$ for $k> 0$ with $\norm{v}_0=\sum_{j\in[p]}\1\{v_j\ne 0\}$ to be the number of nonzero elements in $v$ and $\norm{v}_\infty=\max_{j\in[p]}\abs{v_j}$. We also denote $\abs{\tau}=\sum_{j\in[p]}\1\{j\in\tau\}$ to be the cardinality of $\tau$. For matrix $A\in\R^{q\times p}$ and $\tau\subset[p]$, we denote $A_{\cdot,\tau}$ to be a submatrix of $A$ consisting of all the columns of $A$ with column indexes in $\tau$ and $\norm{A}_{\rm op}=\sup_{\|a\|_2,\|b\|_2\le 1}a^\top Ab$ is the operator norm of $A$. For a symmetric matrix $A$, $\lambda_{\min}(A)$ and $\lambda_{\max}(A)$ denote respectively the smallest and largest eigenvalues of $A$. We use $c$ and $C$ to denote generic 
positive constants that may vary from place to place. For two positive sequences $a_n$ and $b_n$, $a_n\lesssim b_n$ means $a_n\le Cb_n$ for all $n$ and $a_n\gtrsim b_n$ if $b_n\lesssim a_n$ and $a_n\asymp b_n$ if $a_n\lesssim b_n$ and $b_n\lesssim a_n$, and $a_n\ll b_n$ if $\lim\sup_{n\rightarrow\infty}\frac{a_n}{b_n}=0$ and $a_n\gg b_n$ if $b_n\ll a_n$. 
We denote $\bm{X}=(X_1, \ldots, X_n)^\top$, $\bm{X}^{obs}=(X_1^{obs}, \ldots, X_n^{obs})^\top$, $\bm{y}=(y_1, \ldots, y_n)^\top$, and $\bm{y}^{obs}=(y_1^{obs}, \ldots, $ $ y_n^{obs})^\top$. Throughout the paper, $(\bm{X}, \bm{y})$ are the random copy of the data and $(\bm{X}^{obs}, \bm{y}^{obs})$ the realized version.

\subsection{Model Setup}\label{sec_model}
Consider a regression problem of a binary response where independent observations $\{(X_i,y_i):i\in[n]\}$ are generated from the distribution $P_{X,Y}$, with $\Prob(Y=1|X)=1-\Prob(Y=0|X)=\mu(X)\in(0,1)$ and $X\sim P_X$,
$X\in\R^p$, $Y\in\{0,1\}$.  Here, $p\gg n$ and the form of mean function $\mu(\cdot)$ is completely unknown.  
An equivalent form of a data-generating model, with $U$$\sim$ ${\rm Unif}(0,1)$ independent of $X$, is  
\begin{equation}\label{eq_eta_generating}
   Y=\1\{\mu(X)>U\}. 
\end{equation}

Since the mean function $\mu(X)$ is not assumed to be a sparse model, it is often~not~possible~to estimate $\mu$ accurately under the high-dimensional $p\gg n$ setting. To extract meaningful~(interpret- able) information from the data and also utilize existing algorithms in the well-established sparse model literature, we propose to fit working 
sparse GLMs of the form $g^{-1}(X_{\tau}^\top\bm{\beta}_{\tau})$ to approximate $\mu(X)$. Here, $g$$:(0,1)\rightarrow\R$ is a known, increasing link function satisfying $g(\frac{1}{2})=0$. 
The model support $\tau\subset[p]$, $|\tau|\le s$, aims to select up to $s$ covariates that are most influential for the response $Y$, 
and the regression coefficients $\bm{\beta}_\tau$ measure the impact of the selected covariates in the GLM, where $s\in[p]$ is a user-specified upper bound and its  choice will be further discussed in Section~\ref{sec_simulation}. 

Formally, the population-level target parameters $(\tau_0,\bm{\beta}_{0,\tau_0})$ of the working model are as follows.
\begin{enumerate}[1)]
    \item \, We define $\tau_0$ to be the support of the best sparse model in  recovering $Y$ from $X$: 
    \begin{equation}\label{eq_tau0}
        \tau_0\in\argmin_{\tau\subset[p],|\tau|\le s}\inf_{\bm{\beta}_\tau\in\R^{|\tau|}}\Prob\bigg(Y\ne \1\bigg\{g^{-1}(X_{\tau}^\top\bm{\beta}_{\tau})>\frac{1}{2}\bigg\}\bigg).
    \end{equation}
    \item \, Given $\tau_0$, we define $\bm{\beta}_{0,\tau_0}$ as the best $|\tau_0|$-dimensional coefficients for approximating the conditional distribution $P_{Y|X_{\tau_0}}$ in terms of Kullback-Leibler divergence:
    \begin{equation}\label{eq_beta0}
    \bm{\beta}_{0,\tau_0}=\argmax_{\bm{\beta}_{\tau_0}\in\R^{|\tau_0|}}\E l(\tau_0,\bm{\beta}_{\tau_0}|X,Y),
    \end{equation}
    with the function $l(\cdot|\cdot)$ being the log-likelihood of the working GLM,
    \begin{equation}\label{eq_loglike}
        l(\tau,\bm{\beta}_\tau|X,Y)=Y\log\frac{g^{-1}(X_{\tau}^\top\bm{\beta}_{\tau})}{1-g^{-1}(X_\tau^\top\bm{\beta}_{\tau})}+\log\big(1-g^{-1}(X_{\tau}^\top\bm{\beta}_{\tau})\big).
    \end{equation}
\end{enumerate}
Throughout the paper, we assume $\bm\theta_0=(\tau_0,\bm\beta_{0,\tau_0})$ is uniquely defined, and we refer
the working GLM with parameters $\bm\theta_0$ as {\it the oracle model} and the parameters $\bm\theta_0$ as {\it the oracle parameters}. 
The simple structure of the oracle sparse GLM enables statistical inference for $\bm{\theta}_0=(\tau_0,\bm{\beta}_{0,\tau_0})$ that includes both the model support $\tau_0$ and the linear coefficients $\bm{\beta}_{0,\tau_0}$, even though the oracle model may be misspecified (i.e., different from the true underlying model). For notational convenience, we extend $\bm\beta_{0,\tau_0}$ to a full vector $\bm\beta_0\in\R^p$ by setting all its components outside $\tau_0$ to zero. The oracle (population) parameter
$\bm\theta_0$ is 
fixed in the theoretical development. In our numerical studies 
in Section~\ref{sec_simulation}, we use a data-driven sparsity level $s$ to balance model complexity and goodness of~fit.

Unlike the existing high-dimensional literature \citep[e.g.][]{van2014asymptotically,zhang2017simultaneous,shi2019linear},  we make no assumptions on the true mean function $\mu(X)$, especially the often-required sparsity of an underlying model. 
Instead, we focus on the optimal GLM defined over a small subset of the most informative covariates $X_{\tau_0}$, which seems more realistic and practical. Lemma C.1 in the Supplement shows that the above oracle GLM achieves favorable performance in terms of bounded prediction error. Lemma C.2 further shows that: a) if $\mu(X)$ is indeed a sparse GLM, the oracle model support $\tau_0$ in \eqref{eq_tau0} is the true support of $\mu(X)$; b) if $\mu(X)$ is dense but the highly  influential covariates 
have maximum $s$ of them, then under a mild signal strength separation condition, the oracle $\tau_0$ in \eqref{eq_tau0} recovers this set of the highly significant covariates. 

\begin{Remark}\label{rem_zero_coefficient}
  We remark that, under the misspecified working model considered in this work,~$\bm\beta_{0,\tau_0}$~is~the oracle GLM coefficient vector based on the subset of covariates $X_{\tau_0}$.  In the special case when the true underlying model  is a sparse GLM,  the underlying GLM regression coefficients, $\bm\beta_0$,  always has $\beta_{0,j}=0$ for $j\not\in\tau_0$, indicating that $X_j$ has no direct impact on~$Y.$ In  more general cases, however, the working model coefficient $\beta_{0,j}=0,j\not\in\tau_0$ merely indicates that this covariate $X_j$ contributes~less to recovering $Y$ than those covariates in $X_{\tau_0}$; here,  $\beta_{0,j}=0,j\not\in\tau_0$ does not imply a lack of association.

\end{Remark}

Finally, the working model $g^{-1}(X_{\tau_0}^\top\bm{\beta}_{0,\tau_0})$ is an approximation to the true mean function $\mu(X)$. For ease of presentation, we define a {\it working noise} term as $\epsilon=g\big(\mu(X)\big)-X_{\tau_0}^\top\bm\beta_{0,\tau_0}-g(U)$,
then the data-generating model~\eqref{eq_eta_generating} can be re-expressed as 
\begin{equation}\label{eq_generating_model_linear}
    Y=\1\big\{X_{\tau_0}^\top\bm{\beta}_{0,\tau_0}+\epsilon>0\big\}. 
\end{equation}
Now, consider the set of all realizations $\{(X_i^{obs}, y_i^{obs}, u_i^{rel}, \epsilon_i^{rel}): i\in[n]\}$, including the \textit{observed data} 
$(X_i^{obs}, y_i^{obs})$ and the corresponding \textit{realizations} of the data-generating $u_i^{rel}$ and working noise 
$\epsilon_i^{rel}=g\big(\mu(X_i^{obs})\big)-(X_{i,\tau_0}^{obs})^\top\bm\beta_{0,\tau_0}-g(u_i^{rel})$. Then, the realized version of (\ref{eq_generating_model_linear}) is $y_i^{obs}=\1\big\{(X_{i,\tau_0}^{obs})^\top\bm{\beta}_{0,\tau_0}+\epsilon_i^{rel}>0\big\}$, which directly relates the observed data with their corresponding realized working noise $\epsilon_i^{rel}$. 
We denote $\bm{u}=(u_1, \ldots,u_n)^\top$, $\bm{u}^{rel}=(u_1^{rel},\ldots,u_n^{rel})^\top$, $\bm{\epsilon}=(\epsilon_1, \ldots, \epsilon_n)^\top$, $\bm{\epsilon}^{rel}=(\epsilon_1^{rel}, \ldots, \epsilon_n^{rel})^\top$. Here, $(\bm{u}, \bm{\epsilon})$ are the random noises and $( \bm{u}^{rel}, \bm{\epsilon}^{rel})$ are the corresponding realized version. Although $\bm u$ and $\bm\epsilon$ are in bijective correspondence, we keep both notations and use them to represent the noises in the original data-generating model expression \eqref{eq_eta_generating} and 
its re-expression \eqref{eq_generating_model_linear}, respectively.
 %
\subsection{Brief Review of Repro Samples Method and Our New Developments}
\label{sec_repro}

The repro samples method 
can 
construct confidence regions for various types of parameters,  taking values in either continuous, discrete, or mixed spaces. Let $\bm{y}^{obs}=\{y_1^{obs},\ldots,y_n^{obs}\}$ be the observed sample from the population $\bm{Y}=G(\bm{U},\bm{\theta}_0)$ defined by a generative model $G: \mathcal{U}\times\Theta\rightarrow \R^n$, where  $\bm{\theta}_0\in\Theta$ is the unknown model parameter and $\bm{U}\in\mathcal{U}$ is a random vector from a known distribution $P_U$. Then, the observed 
data 
$\bm{y}^{obs}=G(\bm{u}^{rel},\bm{\theta}_0)$,  where $\bm{u}^{rel}$ is a realization of the random vector~$\bm{U}$. 

Intuitively, if $\bm{u}^{rel}$ were observed, for a $\bm{\theta}\in\Theta$, we can get an artificial data $\bm{y}'=G(\bm{u}^{rel},\bm{\theta})$. If the artificial data matches the observed data, i.e., $\bm{y}' =\bm{y}^{obs}$, then $\bm{\theta}$ is a potential value of $\bm{\theta}_0$ and,  if there is any ambiguity, it
comes only from the inversion of $G(\bm{u}^{rel},\cdot)$. However, the realized noises $\bm{u}^{rel}$ are unobserved, so we need to incorporate their uncertainty,  which we quantify by a Borel set $B_{1 - \alpha}$ with $\Prob(\bm{U}\in B_{1 - \alpha})\ge 1 - \alpha$. For any  
$\bm{\theta}\in\Theta$, if there is a potential realization $\bm{u}^*\in B_{1 - \alpha}$ that satisfies $\bm{y}^{obs}=G(\bm{u}^{*},\bm{\theta})$ and thus reproduces $\bm{y}^{obs}$, then we cannot rule out that this $\bm{\theta}$ is the true $\bm{\theta}_0$. We can show that the collection of such  $\bm{\theta}$ values forms a level $1-\alpha$ confidence set for $\bm{\theta}_0$; cf, \cite{xie2022repro}. Here, $\bm{\theta}_0$ can be of any type, continuous, discrete or even non-numerical. 


More generally, we consider a Borel set $B_{1 - \alpha}(\bm{\theta})$ with $\Prob(T(\bm{U},\bm{\theta})\in B_{1 - \alpha}(\bm{\theta}))\ge 1- \alpha$, where  $T:\mathcal{U}\times\Theta\rightarrow \R^d$ is called a {\it nuclear mapping} function by \cite{xie2022repro}. Then, let 
\begin{equation}\label{eq_reproG}
    \Gamma_\alpha^{\bm{\theta}_0}(\bm{y}^{obs})=\{\bm{\theta}:\exists \bm{u}^*~{\rm s.t.}~\bm{y}^{obs}=G(\bm{u}^*,\bm{\theta}), T(\bm{u}^*,\bm{\theta})\in B_{1 - \alpha}(\bm{\theta})\}, \
\end{equation}
and it follows that
$\Prob(\bm{\theta}_0\in\Gamma_\alpha^{\bm{\theta}_0}(\bm{Y}))\ge\Prob\big(T(\bm{U},\bm{\theta}_0)\in B_{1 - \alpha}(\bm{\theta}_0)\big)\ge 1- \alpha$. So (\ref{eq_reproG}) gives us a confidence set for $\bm{\theta}_0$.
Clearly, there might be multiple choices of $T$ that all lead to valid confidence regions. One choice is $T(\bm{u},\bm{\theta})=\bm{u}$ for any $\bm{\theta}\in\Theta$, the special case described in the previous paragraph. 
However, this naive nuclear mapping choice could lead to an oversized confidence region in more complex inference problems. Here, the choice of $T$ is similar to that of a test statistic under the hypothesis testing framework and it should be designed properly; see \cite{xie2022repro} and \cite{wang2022finite} for more detailed discussions. Note that $T$ can also be defined through a test statistic by $T(\bm{u}^*,\bm{\theta})=\tilde T(G(\bm{u}^*,\bm{\theta}), \bm{\theta})$, where $\tilde T(\bm{y},\bm{\theta})$ is a test statistic under the Neyman-Pearson framework. Theorem 2 of \cite{xie2022repro}
proved that in this case $\Gamma_\alpha^{\bm{\theta}_0}(\bm{y}^{obs}) \subseteq \tilde\Gamma_\alpha^{\bm{\theta}_0}(\bm{y}^{obs}) \overset{\text{def}}{=} \{\bm{\theta}:\tilde T(\bm{y}^{obs},\bm{\theta})\in B_{1 - \alpha}(\bm{\theta})\}$, which is obtained by inverting 
the classical test via the Neyman
test duality property.  
So,  
$\Gamma_\alpha^{\bm{\theta}_0}(\bm{y}^{obs})$ constructed in (\ref{eq_reproG}) by the repro samples method is either the same or better than  the confidence set $\tilde\Gamma_\alpha^{\bm{\theta}_0}(\bm{y}^{obs})$ via the classical Neyman-Pearson framework.
In cases when nuisance parameters are present, \cite{xie2022repro} proposed a nuclear mapping function by profiling out the nuisance components to make inferences for the parameters of interest. 

So far, the repro samples method has been developed under the assumption that the data-generating model, or at least the distribution of the random noise terms, is well specified; cf. \cite{xie2022repro, wang2022finite}. The repro samples development in this paper is the first attempt that allows the underlying model to be completely unknown and likely misspecified. Particularly, we extend the existing work in three key directions. First, in well-specified models, the inference targets are naturally defined, but when the model is misspecified, target parameters must be carefully chosen so that they can capture meaningful information, be interpretable and also remain inferable. To this end, we introduce the sparsity-constrained GLM as a working model and define the inference targets as the subset of the most influential covariates and their associated GLM coefficients. Besides \cite{tomaselli2025robust}, the approach to handle a misspecified model also~aligns with relevant research approaches in 
variational Bayes \citep[e.g.,][]{
wang2019variational}.
Second, unlike the linear model setting in \cite{wang2022finite}, with the binary response in our case, multiple values of parameters may satisfy \eqref{eq_generating_model_linear} given the observed data and realized noise. This aspect significantly complicates the inversion task of screening for a model candidate set, both from computational and theoretical standpoints. Third, 
 when we make inferences for model support (whether under a well-specified model or the oracle working model definition), the regression coefficients are treated as unknown nuisance parameters. \cite{wang2022finite} handled these nuisance parameters by conditioning on their sufficient statistics, but it is not possible in our case since such nontrivial sufficient statistics do not exist under GLM.  
We instead need to tackle the computational challenge by designing a nuclear mapping that can profile out all possible values of the nuisance coefficients. See Section~\ref{sec_method} for a detailed explanation of the strategies~to~address~the~above~challenges.

\section{Method and Theory}\label{sec_method}

\subsection{Model Candidate Set}\label{sec_candidate}

Under the high-dimensional setup with $p \gg n$, the parameter space of the target (discrete) model support $\tau_0$ is of a (large) size $O((\frac{p}{s})^s)$. To significantly reduce the search space of the model support and thus the computational cost,  in this section, we use the Fisher inversion technique in the repro samples method \citep[cf.][]{xie2022repro} to construct a 
model candidate set of a much smaller~size.

In Section~\ref{sec_candidate_oracle}, we start with a hypothetical scenario where the entire set of realizations $\{(X_i^{obs}, y_i^{obs},  u_i^{rel}, $ $ \epsilon_i^{rel}): i\in[n]\}$ were available to us. We can show that $\tau_0$ can be identified under a weak signal strength assumption. In practice,
however, the noise $\bm{\epsilon}^{rel}$ is not known. Thus, in Section~\ref{sec_candidate_practice}, we suggest to generate $d$ copies of  working noises $\{\bm{\epsilon}^{*(j)}: j\in[d]\}$ to approximate $\bm{\epsilon}^{rel}$ and get their corresponding estimators $\hat\tau(\bm{\epsilon}^{*(j)})$ of $\tau_0$.
We show that as long as one of the $\bm{\epsilon}^{*(j)}$'s falls in a neighborhood of $\bm{\epsilon}^{rel}$, then the set $\mathcal{C}=\{\hat\tau(\bm{\epsilon}^{*(j)}):j\in[d]\}$, whose expected size is upper bounded, 
will cover $\tau_0$ with a high probability. Here, the distribution of $\bm{\epsilon}^{*(j)}$ is not crucial, which we only require to span the full space $\R^n$. In our practice, we use either a Gaussian or logistic~distribution.

To facilitate the discussions in this section, we define three versions of recovery error: 
\begin{equation}\label{eq_loss}
    \begin{aligned}L^R_n(\tau,\bm{\beta}_\tau,\sigma|\bm{X}^{obs},\bm{y}^{obs},\tilde{\bm{\epsilon}})=&\frac{1}{n}\sum\nolimits_{i=1}^n\1\big\{y_i^{obs}\ne\1\{X_{i,\tau}^{obs\top}\bm{\beta}_\tau+\sigma\tilde\epsilon_i>0\}\big\}\\
    =&\frac{1}{n}\sum\nolimits_{i=1}^n\1\big\{\1\{X_{i,\tau_0}^{obs\top}\bm{\beta}_{0,\tau_0}+\epsilon_i^{rel}>0\}\ne\1\{X_{i,\tau}^{obs\top}\bm{\beta}_\tau+\sigma\tilde\epsilon_i>0\}\big\}, 
\end{aligned}
\end{equation}
\begin{align} \label{eq_loss_E}
    L^R_{\bm{\theta_0}}(\tau,\bm{\beta}_\tau,\sigma)
    = &\E  L^R_n(\tau,\bm{\beta}_\tau,\sigma|\bm{X},\bm{y},\bm{\epsilon})
    = 
 \Prob\big(\1\{X_{\tau_0}^\top\bm{\beta}_{0,\tau_0}+\epsilon>0\}\ne\1\{X_\tau^\top\bm{\beta}_\tau+\sigma\epsilon>0\}\big),
\end{align}
\noindent
where the expectation $\E$ is over the randomness of $\bm{X}$, $\bm{\epsilon}$ and $\bm{y}$ (or equivalently $\bm{X}$ and $\bm{\epsilon}$), and
\begin{align}\label{eq_loss_E_star}
    L_{\bm{\theta}_0}^{R*}(\tau,\bm{\beta}_\tau,\sigma)
    =&{ \E_*} L^R_n(\tau,\bm{\beta}_\tau,\sigma|\bm{X},\bm{y},\bm{\epsilon}^*)
    =\Prob\big(\1\{X_{\tau_0}^\top\bm{\beta}_{0,\tau_0}+\epsilon>0\}\ne\1\{X_\tau^\top\bm{\beta}_\tau+\sigma\epsilon^*>0\}\big), 
\end{align}

\noindent
where the expectation ${ \E_*}$ is over the randomness of $\bm{X},\bm{y}$ and ${\bm{\epsilon}^*}$ (or equivalently $\bm{X},\bm{\epsilon}$ and ${\bm{\epsilon}^*}$). Here,   (\ref{eq_loss_E}) is the expected recovery error (expected loss) corresponding to $L^R_n(\tau,\bm{\beta}_\tau,\sigma|\bm{X}^{obs},\bm{y}^{obs},{\bm{\epsilon}}^{rel})$, while (\ref{eq_loss_E_star}) is the expected recovery error (expected loss) corresponding to     $L^R_n(\tau,\bm{\beta}_\tau,\sigma|\bm{X}^{obs},\bm{y}^{obs},\bm{\epsilon}^*)$ for an independent noise copy $\bm{\epsilon}^*$ that is not ${\bm{\epsilon}}^{rel}$.  Also, 
(\ref{eq_loss}) can be computed for a given $\tilde{\bm{\epsilon}}$. If we knew the realized $\bm{\epsilon}^{rel}$, then 
replacing $\tilde{\bm{\epsilon}}$ with $\bm{\epsilon}^{rel}$ in  (\ref{eq_loss}) would give us the realized empirical loss of (\ref{eq_loss_E}). 




\subsubsection{Signal Strength Condition and Recovery under Oracle Setting}\label{sec_candidate_oracle}

In this subsection, we provide a set of theoretical sufficient conditions for recovering $\tau_0$ under a hypothetical ideal situation in which we knew the unknown $\bm{\epsilon}^{rel}$. This set of conditions is very mild, and our methods and theories in the later sections are developed under these conditions.   

Since $L^R_{\bm{\theta}_0}(\tau,\bm{\beta}_{\tau},\sigma)$ in \eqref{eq_loss_E} attains its minimum value $0$ at $(\tau_0,\bm{\beta}_{0,\tau_0},1)$, we~have~$\tau_0$~=~$\argmin_\tau$ $\min_{\beta_\tau, \sigma}$ $L^R_{\bm{\theta_0}}(\tau,$ $\bm{\beta}_\tau,\sigma)$. Since $L^R_n$ is an empirical version of $L^R_{\bm{\theta}_0}$, if we knew $\bm{\epsilon}^{rel}$, we would~estimate $\tau_0$ by minimizing $L^R_n(\tau,\bm{\beta}_\tau,\sigma|\bm{X}^{obs},\bm{y}^{obs},\bm{\epsilon}^{rel})$.
Without imposing any conditions on $\bm{\beta}_{0,\tau_0}$, however, 
the minimizer of $L^R_n$ may differ from $\tau_0$, making it almost impossible to identify $\tau_0$ using the realized set $(\bm{X}^{obs},\bm{y}^{obs},\bm{\epsilon}^{rel})$, let alone using only observed $(\bm{X}^{obs},\bm{y}^{obs})$. 
To identify $\tau_0$, 
we make the following assumption on the signal strength to separate $\tau_0$ from all other models $\tau$ where $\tau \not = \tau_0$, $|\tau| \le |\tau_0|$.

\begin{Assumption}\label{ass_signal}
For all $\tau\subset[p]$ with $\abs{\tau}\le |\tau_0|,\tau\ne\tau_0$, 
\begin{equation}\label{eq_signal}
    \inf_{\bm{\beta}_\tau\in\R^{\abs{\tau}},\sigma\ge 0}L^R_{\bm{\theta}_0}(\tau,\bm{\beta}_\tau,\sigma)
    \gtrsim (\abs{\tau}+1)\dfrac{\log \frac{n}{\abs{\tau}+1}}{n}+\min\bigg\{\abs{\tau_0\setminus\tau}\dfrac{\log p}{n},(\abs{\tau}+1)\dfrac{\log p}{n}\bigg\}.
\end{equation} 
\end{Assumption} 


Assumption \ref{ass_signal} 
requires that, at the population level, 
any model $\tau$ with $\abs{\tau}\le |\tau_0|,\tau\ne\tau_0$ has a positive error gap at a certain rate, compared to $\tau_0$, noting that 
$L^R_{\bm{\theta}_0}(\tau_0,\bm{\beta}_{0,\tau_0},1)=0$ under the true parameters $(\tau_0,\bm{\beta}_{0,\tau_0})$. This required rate,
with respect to \(n\), is \(O(\frac1n)\), rather than a commonly imposed \(O(\frac1{\sqrt{n}})\) rate, for signal strength in the literature.
\citep[e.g.,][etc.]{bunea2008honest,zhang2010nearly,shen2012likelihood}.  
%
  Specifically,  in Section A of Supplement, we show that, when the sparse GLM is the true model (i.e., $\mu(X) = g^{-1}(X_{\tau_0}^\top\bm{\beta}_{0,\tau_0})$ $P_X$-almost surely), a) Assumption \ref{ass_signal} is weaker than the $C_{\rm min}$ condition imposed in \cite{shen2012likelihood}, if $s\log\frac{n}{s}\lesssim\log p$; b)  Assumption \ref{ass_signal} is weaker than the $\beta_{\min}$ condition used in \cite{bunea2008honest,zhang2010nearly,zhao2006model}, if $s\log p+\frac{s^2\log^2\frac{n}{s}}{\log p}\lesssim n$.
    %

Under the signal strength gap condition of Assumption \ref{ass_signal}, 
we can show that the solution that minimizes 
the constrained empirical risk,  
\begin{equation} \label{eq_oracle}
    \hat \tau(\bm{\epsilon}^{rel})=\argmin_{|\tau|\le s}\min_{\bm{\beta}\in\R^p,\sigma\ge 0}L^R_n(\tau,\bm{\beta}_\tau,\sigma|\bm{X}^{obs},\bm{y}^{obs},\bm{\epsilon}^{rel}), 
\end{equation}
is likely  equal to $\tau_0$. More formally, the Lemma \ref{lem_oracle} below states that, when Assumption~\ref{ass_signal} holds, we can identify $\tau_0$ with high probability using the realization set $(\bm{X}^{obs},\bm{y}^{obs},\bm{\epsilon}^{rel})$. In Lemma~\ref{lem_oracle}, 
$\hat \tau(\bm{\epsilon}) = \argmin_{|\tau|\le s}$ $\min_{\bm{\beta}\in\R^p,\sigma\ge 0}L^R_n(\tau,\bm{\beta}_\tau,\sigma|\bm{X},\bm{y},\bm{\epsilon})$ is a random copy of $\hat\tau(\bm{\epsilon}^{rel})$ and the probability $\Prob(\cdot)$ is taken with respect to $(\bm{X}, \bm{y}, \bm{\epsilon})$. A proof is given in the Supplement.

\begin{Lemma}\label{lem_oracle}
If Assumption \ref{ass_signal} holds, we have
$$\Prob(\hat\tau(\bm{\epsilon})\ne\tau_0)\lesssim 2^{-cn\tilde c_{\min}}\wedge 2^{-cnc_{\min}},$$

\noindent
for some constant $c>0$. 
Here, 
 \(\tilde c_{\min}=\min_{|\tau|\le |\tau_0|,\tau\ne\tau_0,\bm{\beta}_\tau\in\R^{\abs{\tau}},\sigma\ge 0}\frac{L^R_{\bm{\theta}_0}(\tau,\bm{\beta}_\tau,\sigma)-\frac{2\abs{\tau}+2}{n}\log_2\frac{2en}{\abs{\tau}+1}}{\abs{\tau_0\setminus\tau}}, \) and 
\(c_{\min}=\min_{|\tau|\le |\tau_0|,\tau\ne\tau_0,\bm{\beta}_\tau\in\R^{\abs{\tau}},\sigma\ge 0}\frac{L^R_{\bm{\theta}_0}(\tau,\bm{\beta}_\tau,\sigma)-\frac{2\abs{\tau}+2}{n}\log_2\frac{2en}{\abs{\tau}+1}}{\abs{\tau}\vee 1}.\) 
\end{Lemma}

\subsubsection{Candidate Set Construction in the Practical Setting}\label{sec_candidate_practice}

In practice, the realized noise $\bm{\epsilon}^{rel}$ is unobserved. However,  we can generate a vector $\bm{\epsilon}^*$~independently 

\noindent from some distribution spanning $\R^n$, such as Gaussian or logistic. We then calculate $\hat \tau(\bm{\epsilon}^*)$ as 
\[\hat \tau(\bm{\epsilon}^*)=\argmin_{|\tau|\le s}\min_{\bm{\beta}\in\R^p,\sigma\ge 0}L^R_n(\tau,\bm{\beta}_\tau,\sigma|\bm{X}^{obs},\bm{y}^{obs},\bm{\epsilon}^*). \]
From the expression of (\ref{eq_loss}), we would expect that  $\hat\tau(\bm{\epsilon}^*)=\hat\tau(\bm{\epsilon}^{rel})$ when $\bm{\epsilon}^*$ is close enough to $\bm{\epsilon}^{rel}$.~So, we generate $d$ copies of repro sample noises $\{\bm{\epsilon}^{*(j)}:j\in[d]\}$ from, say, the standard logistic distribution, and hope one of them is sufficiently close to $\bm{\epsilon}^{rel}$, when $d$ is large enough. 
We calculate their corresponding $\hat\tau(\bm{\epsilon}^{*(j)})$ and the collection of these estimates forms~a~set:  
\[\mathcal{C}=\{\hat\tau(\bm{\epsilon}^{*(j)}):\epsilon_i^{*(j)}\overset{{\rm i.i.d.}}{\sim}{\rm Logistic},i\in[n],j\in[d]\} ,\] 
which we refer to as a model candidate set. We summarize the above procedure in Algorithm~\ref{alg_candidate}~below.

\begin{algorithm}
\small
\caption{Model Candidate Set}\label{alg_candidate}
\begin{algorithmic}[1]
\State{\bf Input:} Observed data $(\bm{X}^{obs},\bm{y}^{obs})$, sparsity level $s$ and the number of repro samples $d$. 
\State{\bf Output:} Model candidate set $\mathcal{C}$.
\State Generate $d$ copies of logistic random noises $\{\bm{\epsilon}^{*(j)}:\epsilon_i^{*(j)}\overset{{\rm i.i.d.}}{\sim}{\rm Logistic},i\in[n],j\in[d]\}$.
\State Compute $\hat\tau(\bm{\epsilon}^{*(j)})=\argmin_{|\tau|\le s}\min_{\bm{\beta}\in\R^p,\sigma\ge 0}L^R_n(\tau,\bm{\beta}_\tau,\sigma|\bm{X}^{obs},\bm{y}^{obs},\bm{\epsilon}^{*(j)}),$ for $j\in[d]$.
\State Construct $\mathcal{C}=\{\hat\tau(\bm{\epsilon}^{*(j)}):j\in[d]\}$.
\end{algorithmic}
\end{algorithm}

Theorem \ref{thm_candidate} below states that the candidate set $\mathcal{C}$ contains 
$\tau_0$ with a high probability, even under model misspecification, as long as the weak signal Assumption \ref{ass_signal} is satisfied and $d$ is large enough. Here, $F_{\rm log}(z)=(1+e^{-z})^{-1}$ is the CDF of logistic distribution. 
A proof is given in the Supplement.
 
\begin{Theorem}\label{thm_candidate}
Using the same notation as in Lemma \ref{lem_oracle}, 
we have 
\begin{align*}\Prob(\tau_0\not\in\mathcal{C})\lesssim 2^{-\frac{1}{2}n\tilde c_{\min}+2\log_2p}\wedge2^{-\frac{1}{2}nc_{
    \min}+\log_2p}+\E\big(1-\prod_{i\in[n]}\big|F_{\rm log}(\epsilon_i)-F_{\rm log}(-X^\top_{i,\tau_0}\bm{\beta}_{0,\tau_0})\big|\big)^d.  
\end{align*} 
If Assumption \ref{ass_signal} holds and $d\rightarrow \infty$, then $\Prob(\tau_0\not\in\mathcal{C})\lesssim 2^{-cn\tilde c_{\min}}\wedge 2^{-cnc_{\min}}$ for some constant $c>0$. 

\end{Theorem}

As discussed following Assumption~\ref{ass_signal}, the existing methods require stronger assumptions to identify \(\tau_0\) (i.e., model selection consistency), and they  cannot guarantee recovery under Assumption~\ref{ass_signal}. In contrast, Theorem~\ref{thm_candidate} shows that recovery is possible under the weaker condition, since it relies on a candidate set \(\mathcal{C}\) to recover \(\tau_0\), rather than requiring  consistency of a single~selection.

Theorem \ref{thm_candidate_size} below quantifies the expected cardinality of the model candidate set $\mathcal{C}$; a proof is in the Supplement. Here, 
 $\mathcal{S}_0$ is the set of models $\tau$ whose  recovery errors are comparable to~that~of~$\tau_0$: 
    \[\mathcal{S}_0=\bigg\{\tau\in\mathcal{S}:\Delta_L(\tau)\le c\sqrt{\frac{|\tau_0|+1}{n}}\bigg\},\]
where $\mathcal{S}=\big\{\tau\subset[p]:|\tau|\le|\tau_0|\big\}$ and $\Delta_L(\tau)$ is the excess recovery error of a wrong $\tau$ relative to $\tau_0$, 
 \[\Delta_L(\tau)=\inf_{\bm{\beta}_\tau\in\R^{\abs{\tau}},\sigma\ge 0}L_{\bm{\theta}_0}^{R*}(\tau,\bm{\beta}_\tau,\sigma)-\inf_{\bm{\beta}_{\tau_0}\in\R^{|\tau_0|}}L_{\bm{\theta}_0}^{R*}(\tau_0,\bm{\beta}_{\tau_0},0). \]

\begin{Theorem}\label{thm_candidate_size}
    For some sufficiently large constant $c>0$, 
    the expected cardinality of the model candidate set can be controlled as $\E|\mathcal{C}|-|\mathcal{S}_0|\le\sum_{\tau\in\mathcal{S}\setminus\mathcal{S}_0}e^{-\frac{n}{8}\big\{\Delta_L(\tau)-c\sqrt{\frac{|\tau_0|+1}{n}}\big\}^2+\log d}\wedge 1.$
\end{Theorem}  
In Theorem \ref{thm_candidate_size}, we separate models in $\mathcal{S}$ into two groups $\mathcal{S}_0$ and $\mathcal{S}\setminus\mathcal{S}_0$ according to the excess data recovery error $\Delta_L$. The result states that only models in $\mathcal{S}_0$ with a small excess recovery error (i.e., those smaller than and close to the true model $\tau_0$) will likely be included in the candidate set $\mathcal{C}$ and farther away models in $\mathcal{S}\setminus\mathcal{S}_0$ will not.

We also consider stronger signal strength conditions of which all models $\tau \ne \tau_0$ with $|\tau| \le |\tau_0|$ have a positive excess error from $\tau_0$ of order $\frac{1}{\sqrt{n}}$. This $\frac{1}{\sqrt{n}}$ order is the same as the typical $\frac{1}{\sqrt{n}}$ rate required for model selection consistency under a well-specified GLM \citep[e.g.,][]{zhao2006model,zhang2010nearly,bunea2008honest}, and it is stronger than the $\frac{1}{n}$ order required in Assumption~\ref{ass_signal}. 
Theorem~\ref{thm_candidate_strong_signal} below states that, under a set of such stronger $\frac{1}{\sqrt{n}}$-rate signal strength conditions, the model candidate set $\mathcal{C}$ contains $\tau_0$ with high probability for any $1\le d \ll \infty$.
A proof is in the Supplement.
 
\begin{Theorem}\label{thm_candidate_strong_signal}

For any $n$ and $\tau$ with $\abs{\tau}\le |\tau_0|,\tau\ne\tau_0$, if the excess recovery error $\Delta_L(\tau)$ is large enough such that 
\begin{equation}\label{eq_signal_strong}
    \begin{aligned}
        \Delta_L(\tau)\gtrsim\sqrt{\frac{\abs{\tau_0}+ 1}{n}}+\sqrt{\abs{\tau_0\setminus\tau}\wedge(\abs{\tau}\vee1)}\sqrt{\frac{\log p}{n}},
    \end{aligned} 
\end{equation}
fix any finite $d\ge 1$, for some constant $C>0$, we have 
\[\Prob(\tau_0\not\in\mathcal{C})\lesssim e^{-Cn\tilde c_{\min}^*}\wedge e^{-Cnc_{\min}^*},\quad\E|\mathcal{C}|-1\lesssim e^{-Cn\tilde c_{\min}^*}\wedge e^{-Cnc_{\min}^*}.\]
Here, for some sufficiently large constant $c>0$, 
\[\tilde c_{\min}^*=\bigg(\inf_{|\tau|\le |\tau_0|,\tau\ne\tau_0}\dfrac{\Delta_L(\tau)-c\sqrt{\frac{\abs{\tau_0}+1}{n}}}{\sqrt{\abs{\tau_0\setminus\tau}}}\bigg)^2,\quad c_{\min}^*=\bigg(\inf_{|\tau|\le |\tau_0|,\tau\ne\tau_0}\dfrac{\Delta_L(\tau)-c\sqrt{\frac{\abs{\tau_0}+1}{n}}}{\sqrt{\abs{\tau}\vee 1}}\bigg)^2. \]
\end{Theorem}

    From Theorems \ref{thm_candidate}, \ref{thm_candidate_size}, and \ref{thm_candidate_strong_signal}, we know that the model candidate set $\mathcal{C}$ is adaptive to the signal strength in both coverage and size:
    %
      (a) \textbf{Coverage:} Under the weak signal strength Assumption \ref{ass_signal}, 
      our approach assures $\tau_0\in\mathcal{C}$ provided $d$ is sufficiently large. Furthermore, under the stronger signal strength condition \eqref{eq_signal_strong}, $\tau_0\in\mathcal{C}$ holds for any $d\ge 1$.
      (b) \textbf{Size:} If only the weak signal Assumption \ref{ass_signal} holds, $\mathcal{C}$ may include some wrong models from $\mathcal{S}_0$ that are close to $\tau_0$. However, if we have a stronger signal strength assumption \eqref{eq_signal_strong}, $\mathcal{C}$ contains only $\tau_0$ even for a finite $d$ when $n$ goes to infinity. 
      
      The above set of results is in line with but also goes beyond the model selection consistency results in the literature \citep{zhao2006model,bunea2008honest,zhang2010nearly,shen2012likelihood},  when the observed data are from a sparse GLM. Under the stronger $\frac{1}{\sqrt{n}}$ signal strength condition, our candidate set also recovers the single true model $\tau_0$ with probability tending to one as $n \to \infty$, aligning with model selection consistency results. Moreover, under the weaker $\frac{1}{n}$ signal strength condition, our method can still recover $\tau_0$, whereas existing model selection approaches do not have such recovery. This is possible since we use a set of models to recover $\tau_0$, which is less restrictive than requiring exact model selection of a single model.

 In practical implementations, the 0--1 loss and $\ell_0$ constraint in line 4 of Algorithm~\ref{alg_candidate} may be replaced by a computationally more tractable surrogate loss and penalty. As discussed in Remark~2 of \cite{wang2022finite}, analogous theoretical guarantees can often be established for such surrogate-based implementations, although doing so may require additional assumptions tailored to the particular surrogate loss and penalty. Following \cite{wang2022finite}, we also develop our general theoretical results using the $\ell_0$ penalty rather than committing to a specific surrogate. This choice keeps the theoretical framework general and separates the fundamental statistical guarantees from implementation-specific computational considerations, while allowing researchers to select surrogate methods that best suit their particular applications.  Section~\ref{sec_simulation_candidate} later has further discussions.

\subsection{Inference for Regression Coefficients}

This section focuses on constructing confidence sets for regression coefficients. 
Section~\ref{sec_Abeta} considers the general case of arbitrary linear combinations of the regression coefficients, \( A\bm{\beta}_0 \), for any given \( A \in \mathbb{R}^{q \times p} \), \( q \ge 1 \). 
Section~\ref{sec_betaj} focuses on special cases with \( A = e_j^\top \) for a single coefficient, \( A = I_p \) for the entire vector, and \( A = \bm{X}_{\rm new} \in \mathbb{R}^{n_{\rm new} \times p} \) for predicted (case) probabilities. 
These special cases correspond to the three most commonly studied inference problems in the sparse GLM literature.

\subsubsection{General: Inference for \texorpdfstring{$A\bm{\beta}_0$}{Abeta\_0}} \label{sec_Abeta}
Under the oracle working sparse model $A\bm{\beta}_0=A_{\cdot\tau_0}\bm{\beta}_{0,\tau_0}$, so it depends on both values of $\tau_0$ and $\bm{\beta}_{0,\tau_0}$, where $A_{\cdot\tau_0}$ is a submatrix of $A$ consisting of all the columns with column indexes in $\tau_0$. Since it involves the  discrete unknown parameter $\tau_0$, the inference task is unconventional. We dissect the task of uncertainty quantification for estimating $A\bm{\beta}_0$ in two components: the uncertainty~of unknown $\tau_0$ and the uncertainty~of unknown ${\bm{\beta}}_{0,\tau_0}$ given $\tau_0$. 
When $\tau_0$ is given, we may use a Wald test and the low-dimensional data $\{(X_{i,\tau_0}^{obs},y^{obs}_i):i\in[n]\}$ with covariates $\bm{X}^{obs}_{\cdot\tau_0}$ constrained on~$\tau_0$ to construct a confidence set for $A_{\cdot\tau_0}\bm{\beta}_{0,\tau_0}$.
To address the impact of unknown $\tau_0$, we use the development in Section~\ref{sec_candidate} to construct a candidate set $\mathcal{C}$ of $\tau_0$ and, for each $\tau\in\mathcal{C}$,  perform the Wald test described above to get a level-$(1-\alpha)$ set for $A_{\cdot\tau}{\bm{\beta}}_{0,\tau}$ treating $\tau$ as if it is $\tau_0$. If $\tau_0 \in \mathcal{C}$, then the union of these level-$(1-\alpha)$ sets immediately provides us a level-$(1-\alpha)$ confidence set for $A\bm{\beta}_0=A_{\cdot\tau_0}\bm{\beta}_{0,\tau_0}$. This confidence set accommodates both uncertainties arising from not knowing $\tau_0$ and ${\bm{\beta}}_{0,\tau_0}$.   

Particularly, let us first consider the case where $\tau_0$ is given and derive a confidence set for $A_{\cdot\tau_0}\bm\beta_{0,\tau_0}$ using a Wald test. We denote ${\rm rank}(A_{\cdot\tau_0})=r(\tau_0)\le q\wedge|\tau_0|$ and write the rank factorization of $A_{\cdot\tau_0}$ to be $A_{\cdot\tau_0}=C(\tau_0)D(\tau_0)$ with $C(\tau_0)\in\R^{q\times r(\tau_0)}$, $D(\tau_0)\in\R^{r(\tau_0)\times |\tau_0|}$ and $D(\tau_0)D(\tau_0)^\top=I_{r(\tau_0)}$. Then it suffices to construct a confidence set for $D(\tau_0)\bm\beta_{0,\tau_0}$. We denote  
\[\nabla l(\tau_0,\bm\beta_{\tau_0}|X,Y)=\frac{\partial}{\partial\bm\beta_{\tau_0}} l(\tau_0,\bm\beta_{\tau_0}|X,Y),\quad \nabla^2 l(\tau_0,\bm\beta_{\tau_0}|X,Y)=\frac{\partial^2}{\partial\bm\beta_{\tau_0}\partial\bm\beta_{\tau_0}^\top} l(\tau_0,\bm\beta_{\tau_0}|X,Y), \]
and maximize the (quasi) log-likelihood \eqref{eq_loglike} to get the corresponding (quasi) MLE  
of $\bm\beta_{0,\tau_0}$, 
\[\hat{\bm\beta}_{\tau_0}=\argmax_{\bm\beta_{\tau_0}\in\R^{|\tau_0|}}\sum_{i\in[n]} l(\tau_0,\bm\beta_{\tau_0}|X_i^{obs},Y_i^{obs}). \]
We estimate the asymptotic covariance matrix of $D(\tau_0)\hat{\bm\beta}_{0,\tau_0}$ by a sandwich formula 
\begin{equation}\label{eq_variance}
    \hat V(\tau_0)=D(\tau_0)\hat H(\tau_0)^{-1}\widehat{\Cov}(\nabla l(\tau_0,\hat{\bm\beta}_{\tau_0}|X,Y))\hat H(\tau_0)^{-1}D(\tau_0)^\top,
\end{equation}
where $\widehat{\Cov}$ is the sample covariance matrix and $\hat H(\tau_0)=\frac{1}{n}\sum_{i\in[n]}\nabla^2l(\tau_0,\hat{\bm\beta}_{\tau_0}|X_i^{obs},Y_i^{obs}) $.~A~(squared) 
Wald test statistic for the test $H_0:D(\tau_0)\bm\beta_{0,\tau_0}=t,\bm\beta_{0,\tau_0^c}=\bm0$ vs $H_1:D(\tau_0)\bm\beta_{0,\tau_0}\ne t, \bm\beta_{0,\tau_0^c}=\bm0$ is 
\[\tilde T(\bm X^{obs}, \bm y^{obs},(\tau_0,t))=n\|\hat V(\tau_0)^{-\frac{1}{2}}(D(\tau_0)\hat{\bm\beta}_{0,\tau_0}-t)\|_2^2. \]
Since the Wald test statistic can be approximated by a chi-squared distribution when $|\tau_0| \ll \sqrt{n}$,~i.e., 
\[\Prob\big(\tilde T(\bm X,\bm y,(\tau_0,D(\tau_0)\bm\beta_{0,\tau_0}))\le F^{-1}_{\chi^2_{r(\tau_0)}}(1-\alpha)\big)\rightarrow 1-\alpha, \quad \hbox{as $n \to \infty$}, \]
where $F^{-1}_{\chi^2_r}(1-\alpha)$ is the $(1-\alpha)$-quantile of $\chi^2_r$. Thus, when we know $\tau_0$, a  level-$(1-\alpha)$ confidence set for $D(\tau_0)\bm\beta_{0,\tau_0}$ is $\big\{t:\tilde T(\bm X^{obs},\bm y^{obs},(\tau_0,t))\le F^{-1}_{\chi^2_{r(\tau_0)}}(1-\alpha)\big\}$. Alternative test statistics, such as the score test or other pseudo-likelihood based methods, may also work.

To account for the impact of unknown $\tau_0$, we construct a model candidate set $\mathcal{C}$ using the development in Section~\ref{sec_candidate},  then we 
apply the above Wald method to each $\tau \in \mathcal{C}$ pretending it is the true $\tau_0$ and combine all the sets together to get a level-$(1-\alpha)$ confidence set of $A{\bm{\beta}}_{0}$: 
\begin{align*}
    \Gamma_\alpha^{A{\bm{\beta}}_{0}}(\bm{X}^{obs},\bm{y}^{obs})=\big\{\tilde t:\tilde t=C(\tau)t,\tilde T(\bm{X}^{obs},\bm{y}^{obs},(\tau,t))\le F^{-1}_{\chi^2_{r(\tau)}}(1-\alpha),\tau\in\mathcal{C}\big\}. 
\end{align*}

We summarize the above development in Algorithm \ref{alg_coefficient_confidence} below. 

\begin{algorithm}
\small
\caption{Confidence set for $A{\bm{\beta}}_{0}$}\label{alg_coefficient_confidence}
\begin{algorithmic}[1]
\State{\bf Input:} Observed data $(\bm{X}^{obs},\bm{y}^{obs})$, model candidate set $\mathcal{C}$. 
\State{\bf Output:} Confidence set $\Gamma_{\alpha}^{A\bm{\beta}_{0}}(\bm{X}^{obs},\bm{y}^{obs})$ for $A{\bm{\beta}}_{0}$.
\For{$\tau\in\mathcal{C}$}
\State Calculate the MLE
$\hat{\bm\beta}_{\tau}=\argmax_{\bm\beta_{\tau}\in\R^{|\tau|}}\sum_{i\in[n]} l(\tau,\bm\beta_{\tau}|X_i^{obs},Y_i^{obs})$,
and the matrix factorization $A_{\cdot\tau}=C(\tau)D(\tau)$ with $D(\tau)D(\tau)^\top=I_{r(\tau)}$. 
\State Estimate the asymptotic covariance matrix
$\hat V(\tau)=D(\tau)\hat H(\tau)^{-1}\widehat{\Cov}(\nabla l(\tau,\hat{\bm\beta}_{\tau}|X,Y))\hat H(\tau)^{-1}D(\tau)^\top$, $\hat H(\tau)=\frac{1}{n}\sum_{i\in[n]}\nabla^2l(\tau,\hat{\bm\beta}_{\tau}|X_i^{obs},Y_i^{obs})$.
\State Calculate
$\tilde T(\bm X^{obs}, \bm y^{obs},(\tau,t))=n\|\hat V(\tau)^{-\frac{1}{2}}(D(\tau)\hat{\bm\beta}_{\tau}-t)\|_2^2$.
\EndFor
\State Construct
$\Gamma_{\alpha}^{A{\bm{\beta}}_0}(\bm{X}^{obs},\bm{y}^{obs})=\{\tilde t:\tilde t=C(\tau)t,\tilde T(\bm{X}^{obs},\bm{y}^{obs},(\tau,t))\le F_{\chi_{r(\tau)}^2}^{-1}(1-\alpha),\tau\in\mathcal{C}\}$.
\end{algorithmic} 
\end{algorithm}

Below we provide a formal theoretical support for using Algorithm~\ref{alg_coefficient_confidence} under mild conditions.   
 
\begin{Assumption}\label{ass_subgaussian}
    Suppose $\|X_{\tau_0}\|_{\psi_2}\lesssim 1$. Denote  $\eta(z)\overset{\triangle}{=}g^{-1}(z),$ $ h_1(z)\overset{\triangle}{=}\frac{\eta''(z)}{\eta(z)}-\big(\frac{\eta'(z)}{\eta(z)}\big)^2,$ $ h_0(z)\overset{\triangle}{=}\frac{\eta''(z)}{1-\eta(z)}+\big(\frac{\eta'(z)}{1-\eta(z)}\big)^2,$ %
    we assume $h_1,h_0$ are Lipschitz, 
    \begin{equation}\label{eq_loglike_condition}
        \bigg\|\frac{\eta'}{\eta}\bigg\|_\infty+\bigg\|\frac{\eta'}{1-\eta}\bigg\|_\infty+\|h_1\|_\infty+\|h_0\|_\infty\lesssim 1,\quad h_1<0<h_0.
    \end{equation}
\end{Assumption}
 
\begin{Assumption}\label{ass_hessian}
    Denote $H=\E\nabla^2 l(\tau_0,\bm\beta_{0,\tau_0}|X,Y)$ to be the expected Hessian of the log-likelihood; Let $\lambda_{\min}(-H)\asymp\lambda_{\max}(-H)\asymp\lambda_{\min}(\Cov(\nabla l(\tau_0,\bm\beta_{0,\tau_0}|X,Y)))\asymp\lambda_{\max}(\Cov(\nabla l(\tau_0,\bm\beta_{0,\tau_0}|X,Y)))\asymp 1$. 
\end{Assumption}

\noindent
Assumption \ref{ass_subgaussian} implies that the gradient of the log-likelihood of the oracle model is sub-Gaussian and the Hessian of the log-likelihood is sub-exponential. The $\ell_\infty$ control can be relaxed to other tail probability assumptions, such as sub-Gaussian conditions. Here we take $\ell_\infty$ for simplicity, and it is satisfied by the logistic regression model. 
Assumption \ref{ass_hessian} is on the Hessian and covariance matrices under $\tau_0$, rather than those with respect to the full coefficient vector $\bm{\beta}_0$. It is weaker than other commonly imposed conditions on the Hessian or covariance matrices in the literature \citep[cf.,][]{cai2021statistical,van2014asymptotically,fei2021estimation}.

Theorem~\ref{thm_Abeta} below states that the confidence set $\Gamma_\alpha^{A{\bm{\beta}}_0}(\bm{X}^{obs},\bm{y}^{obs})$ is asymptotically valid for $A{\bm{\beta}}_0$. A proof can be found in the Supplement.
 
\begin{Theorem}\label{thm_Abeta}
If Assumptions \ref{ass_subgaussian}, \ref{ass_hessian} hold and $n\gg s^2$, when one of the following conditions holds:
\begin{enumerate}[1)] 
    \item \, $d\rightarrow \infty$ at first, then $n\rightarrow\infty$, and $n, p, s$ satisfy Assumption \ref{ass_signal}, or 
    \item \, fix any $d$, $n\rightarrow \infty$, and $n, p, s$ satisfy condition \eqref{eq_signal_strong},
\end{enumerate} 
the confidence set $\Gamma_\alpha^{A{\bm{\beta}}_0}(\bm{X}^{obs},\bm{y}^{obs})$ is asymptotically valid, 
\[\Prob(A{\bm{\beta}}_0\in\Gamma_{\alpha}^{A{\bm{\beta}}_0}(\bm{X},\bm{y}))\ge1-\alpha-o(1). \]
\end{Theorem}
    Note that our target parameter $\bm\beta_{0,\tau_0}$ is defined to be the optimal GLM based on a subset of covariates $X_{\tau_0}$ and we do not need to have the typical sparsity assumption that the outside covariates $X_j$, $j \not \in \tau_0$, have no impact on $Y$. Thus, the standard inference methods for high-dimensional problems \citep{shi2019linear,van2014asymptotically,cai2021statistical} are not applicable. See Section \ref{sec_simulation}, the settings of (M1) and (M2), for such concrete numerical examples.  
%

\subsubsection{Applications: Inference for \texorpdfstring{$\beta_{0,j}$}{beta\_{0,j}}, \texorpdfstring{$\bm{\beta}_{0}$}{beta\_0}, and Case Probabilities} \label{sec_betaj}

Most existing literature on sparse regression inference focuses on 
marginal inference for single coefficients $\beta_{0,j}$ or simultaneous inference for the whole vector ${\bm{\beta}}_0$. They correspond to the special cases of Section \ref{sec_Abeta} with $A=e_j^\top$ and $A=I_p$, respectively.
Additionally, there are some recent publications that conduct simultaneous inference for the case probabilities of a set of new observations in sparse GLMs. For instance, in an analysis of electronic health record data, \cite{guo2021inference} studied the case probabilities of $\{g^{-1}(X_{{\rm new},i}^\top\bm{\beta}_0):i\in[n_{\rm new}]\}$ for a group of new patients $i\in[n_{\rm new}]$ of given covariates $X_{{\rm new},i}$, where $g^{-1}(X_{{\rm new},i}^\top\bm{\beta}_0) = \Prob(Y_{\text{new},i}=1|X_{\text{new},i})$ is the expectation of the binary outcome of the new patient $i$. 
Let  
$\bm{X}_{\rm new}=(X_{{\rm new},1},\ldots,X_{{\rm new},n_{\rm new}})^\top\in\R^{n_{\rm new}\times p}$ and $g^{-1}(\bm{X}_{\rm new}\bm{\beta}_0)=(g^{-1}(X_{{\rm new},1}^\top\bm{\beta}_0),\ldots,g^{-1}(X_{{\rm new},n_{\rm new}}^\top\bm{\beta}_0))^\top\in\R^{n_{\rm new}}$. We can first use the results of Section \ref{sec_Abeta} with
$A=\bm X_{\rm new}$ to make a simultaneous  inference for $\bm{X}_{\rm new}\bm{\beta}_0$
and then obtain the simultaneous inference for $g^{-1}(\bm X_{{\rm new}}\bm{\beta}_0)$. Again, this inference problem on case probabilities 
corresponds to a special case of Section \ref{sec_Abeta}. In this subsection we focus on the testing problems of these special cases under our broader working sparse GLM setup.


Since they are special cases, by plugging $A=e_j^\top$, $A=I_p$, and $A=\bm X_{\rm new}$ into Algorithm \ref{alg_coefficient_confidence}, we can construct confidence sets $\Gamma_\alpha^{\beta_{0,j}}(\bm X^{obs},\bm y^{obs})$, $\Gamma_\alpha^{\bm\beta_0}(\bm X^{obs},\bm y^{obs})$, and $\Gamma_\alpha^{g^{-1}(\bm X_{\rm new}\bm\beta_0)}(\bm X^{obs},\bm y^{obs})$. 
It follows immediately from Theorem \ref{thm_Abeta} that they are asymptotically valid level-$(1-\alpha)$ confidence~sets.

\begin{Corollary}\label{cor_Abeta_app}
    If Assumptions \ref{ass_subgaussian}, \ref{ass_hessian} hold and $n\gg s^2$, when one of the following conditions holds
\begin{enumerate}[1)] 
    \item \, $d\rightarrow \infty$ at first, then $n\rightarrow\infty$, and $n, p, s$ satisfy Assumption \ref{ass_signal}, or 
    \item \, fix any $d$, $n\rightarrow \infty$, and $n, p, s$ satisfy condition \eqref{eq_signal_strong},
\end{enumerate} 
the confidence sets $\Gamma_\alpha^{\beta_{0,j}}(\bm X^{obs},\bm y^{obs})$, $\Gamma_\alpha^{\bm\beta_0}(\bm X^{obs},\bm y^{obs})$, and $\Gamma_\alpha^{g^{-1}(\bm X_{\rm new}\bm\beta_0)}(\bm X^{obs},\bm y^{obs})$ are asymptotically valid; i.e., \(\Prob(\beta_{0,j}\in\Gamma_{\alpha}^{\beta_{0,j}}(\bm{X},\bm{y}))\ge1-\alpha-o(1)\), \(\Prob(\bm\beta_0\in\Gamma_{\alpha}^{\bm\beta_0}(\bm{X},\bm{y}))\ge1-\alpha-o(1)\) and \(\Prob(g^{-1}(\bm X_{\rm new}\bm\beta_0)\in\Gamma_{\alpha}^{g^{-1}(\bm X_{\rm new}\bm\beta_0)}(\bm{X},\bm{y}))\ge1-\alpha-o(1)\), respectively.
\end{Corollary}

Note that Theorem \ref{thm_Abeta} is developed for the oracle working sparse model. When the underlying true model $\mu(X)$ is a sparse GLM, the oracle model is the true sparse GLM. Thus, the results in Corollary~\ref{cor_Abeta_app} also cover the aforementioned three special inference problems of interest  in the literature on sparse GLMs.  Unlike the existing methods in the literature, here we explicitly account for the uncertainty in $\tau_0$ through the model candidate set $\mathcal{C}$. As a result, our theory requires only the weaker sparsity condition $n\gg s^2$, rather than the commonly imposed $n\gg s^2{\rm polylog}(p)$ in the literature when the true underlying GLM is sparse. Due to space limitations, we have placed the detailed implementation of Algorithm~\ref{alg_coefficient_confidence} for these special cases and further theoretical comparisons with existing methods under sparse GLMs in the Supplement.

\subsection{Inference for \texorpdfstring{$\tau_0$}{tau\_0} When \texorpdfstring{$\mu(X)$}{mu(X)} Is a Sparse GLM}\label{sec_tau}

A particularly interesting case is 
when the sparse GLM is correctly specified, i.e., the mean function in \eqref{eq_eta_generating}  $\mu(X)\equiv g^{-1}(X_{\tau_0}^\top\bm\beta_{0,\tau_0})$.
This is the starting assumption in most existing work on sparse GLMs. For such sparse models, our framework can additionally quantify uncertainty in model selection by constructing a level-$(1-\alpha)$ confidence set for $\tau_0$. Relatively little work has addressed this problem, in part because conventional CLT-based methods have difficulty handling inference problems involving discrete parameters such as $\tau_0$.

In this sparse GLM case, the data-generating model \eqref{eq_generating_model_linear} becomes 
\begin{equation}\label{eq_generation_correct}
Y=\1(X_{\tau_0}^\top\bm\beta_{0,\tau_0}+\epsilon>0),\quad \epsilon=-g(U),\quad U\sim \text{Unif}[0,1].
\end{equation}
Here, we have proved that the oracle model support defined in \eqref{eq_tau0} is the true support $\tau_0$ of $\mu(X)$, and that the coefficient defined in \eqref{eq_beta0} is the true coefficient $\bm{\beta}_{0,\tau_0}$ of $\mu(X)$. Due to space limitations, the proof is given in Lemma~C.2 of the Supplement.

In this section, our target is the true model $\tau_0$, and  $\bm{\beta}_{0,\tau_0}$  is 
a nuisance parameter. To deal with the nuisance parameter, we consider the following form of confidence set in which we profile out $\bm\beta$, 
\begin{equation}\label{eq_model_conf}
    \Gamma^{\tau_0}_\alpha(\bm{X}^{obs},\bm{y}^{obs})=\big\{\tau:\min_{\bm{\beta}_\tau\in\R^{|\tau|}}\tilde T(\bm{X}^{obs},\bm{y}^{obs},(\tau,\bm{\beta}_\tau))<1-\alpha\big\}.
\end{equation}

 \noindent
In \eqref{eq_model_conf}, $1-\tilde T(\bm{X},\bm{Y}^*,(\tau,\bm{\beta}_\tau))$ serves as a $p$-value for $\bm Y^*$ generated according to \eqref{eq_generation_correct} with parameter

\noindent$(\tau,\bm\beta_\tau)$. Here, we refer to
$\min_{\bm{\beta}_\tau\in\R^{|\tau|}}\tilde T(\bm{X}^{obs},\bm{Y}^*,(\tau,\bm{\beta}_\tau))$ as a {\it profile nuclear mapping function}.

Specifically, we construct the nuclear mapping statistic $\tilde T$ and the model confidence sets as follows. For any given $\bm{\theta}=(\tau,\bm{\beta}_\tau)$ and $\bm{Y}^*\in\{0,1\}^n$ generated by $Y_i^*=\mathbbm{1}\{X_{i,\tau}^{obs\top}\bm{\beta}_{\tau}+\epsilon_i^*> 0\}$ with $\epsilon_i^*=-g(u_i^*)$, $u_i^*\overset{\rm i.i.d.}{\sim}{\rm Unif}[0,1]$, we compute the Lasso estimator based on the synthetic response~$\bm Y^*$, 
\begin{equation}\label{eq_synthetic_lasso}
    \tilde{\bm{\beta}}(\lambda)=\argmin_{\bm{\beta}\in\R^p}-\frac{1}{n}\sum_{i=1}^n\bigg\{Y_i^*\log\frac{g^{-1}(X_i^{obs\top}\bm\beta)}{1-g^{-1}(X_i^{obs\top}\bm\beta)}+\log\big(1-g^{-1}(X_i^{obs\top}\bm\beta)\big)\bigg\}+\lambda\norm{\bm{\beta}}_1. 
\end{equation}
Then we choose the smallest value of $\lambda$ such that the model sparsity doesn't exceed $|\tau|$, 
\[\tilde\lambda(\tau,\bm{\beta}_\tau)=\argmax_{\lambda\ge0}\|\tilde{\bm{\beta}}(\lambda)\|_0,\quad{\rm s.t. }\norm{\tilde{\bm{\beta}}(\lambda)}_0\le\abs{\tau}. \]
The selected model support is then given by 
\[\tilde\tau(\bm{X}^{obs},\bm{Y}^*,\bm{\theta})=\supp(\tilde{\bm{\beta}}(\tilde\lambda(\bm{\theta}))).\]
The model selector $\tilde\tau(\bm{X}^{obs},\bm{Y}^*,\bm{\theta})$ is the largest model with cardinality at most $|\tau|$ in the solution path of Problem \eqref{eq_synthetic_lasso} using the synthetic data $(\bm{X}^{obs},\bm{Y}^*)$. 
We denote 
\[P_{\bm{\theta}}(\tau^*)=\Prob_{\bm{\epsilon}^*|\bm{\theta}}(\tilde \tau(\bm{X}^{obs},\bm{Y}^*,\bm{\theta})=\tau^*), \]
where $\Prob_{\bm{\epsilon}^*|\bm{\theta}}$ accounts for the randomness of $\bm{Y}^*$ given $\bm{X}^{obs}$. Then we consider the nuclear statistic 
\[\tilde T(\bm{X}^{obs},\bm{y},\bm{\theta})=\Prob_{\bm{\epsilon^*}|\bm{\theta}}\big(P_{\bm{\theta}}(\tilde\tau(\bm{X}^{obs},\bm{Y^*},\bm{\theta}))>P_{\bm{\theta}}(\tilde\tau(\bm{X}^{obs},\bm{y},\bm{\theta}))\big),\]
which is the probability that $\tilde\tau(\bm{X}^{obs},\bm{y},\bm{\theta})$ appears less often than the synthetic model selector $\tilde\tau(\bm{X}^{obs},\bm{Y}^*,\bm{\theta})$ in $P_{\bm{\theta}}(\cdot)$. Since $\tilde T(\bm{X}^{obs},\bm{y},\bm{\theta})$ is also the survival function of random variable $P_{\bm{\theta}}(\tilde \tau(\bm{X}^{obs},\bm{Y}^*,\bm{\theta}))$ evaluated at $P_{\bm{\theta}}(\tilde\tau(\bm{X}^{obs},\bm{y},\bm{\theta}))$, when $\bm{\theta}=\bm{\theta}_0,\bm{y}=\1(\bm{X}^{obs}\bm{\beta}_0+\bm{\epsilon}>0)$, we know that $1-\tilde T(\bm{X}^{obs},\bm{y},\bm{\theta}_0)$ is a p-value with 
\[\Prob_{\bm{\epsilon}}(\tilde T(\bm{X}^{obs},\bm{y},\bm{\theta}_0)<1-\alpha)\ge1-\alpha. \]
Here $\Prob_{\bm{\epsilon}}$ counts the randomness of $\bm{y}$ given $\bm{X}^{obs}$. Since $\tau_0$ belongs to $\mathcal{C}$ with high probability as guaranteed by Theorems \ref{thm_candidate} and \ref{thm_candidate_strong_signal}, we constrain the model confidence set to be a subset of $\mathcal{C}$. Then according to Equation \eqref{eq_model_conf}, we define the confidence set for $\tau_0$ as 
\begin{align*}
    \Gamma_\alpha^{\tau_0}(\bm{X}^{obs},\bm{y}^{obs})=&\big\{\tau:\min_{\bm{\beta}_\tau\in\R^{|\tau|}}\tilde T(\bm{X}^{obs},\bm{y}^{obs},(\tau,\bm{\beta}_\tau))<1-\alpha,\tau\in\mathcal{C}\big\}.
\end{align*}

Since we don't have an explicit expression for $P_{\bm{\theta}}(\tau)$, we apply the Monte Carlo method to approximate it. More specifically, we generate $\{\bm{\epsilon}^{*(j)}:j\in[m]\}$ with $\epsilon_i^{*(j)}=-g(u_i^{*(j)})$, $u_i^{*(j)}\overset{i.i.d.}{\sim}\text{Unif}[0,1]$ for $i\in[n], j\in[m]$, and then generate $\{\bm{Y}^{*(j)}:j\in[m]\}$ by $Y_i^{*(j)}=\mathbbm{1}\{X_{i,\tau}^{obs\top}{\bm{\beta}}_\tau+\epsilon_i^{*(j)}> 0\}$. For each $\bm{Y}^{*(j)}$, we calculate the corresponding $\tilde\tau^{(j)}\overset{\triangle}{=}\tilde\tau(\bm{X}^{obs},\bm{Y}^{*(j)},\bm{\theta})$ and estimate $P_{\bm{\theta}}(\tau^*)$ by 
$\hat P_{\bm{\theta}}(\tau^*)=\frac{1}{m}\sum_{j=1}^m\mathbbm{1}\{\tilde \tau^{(j)}=\tau^*\}$. Denote the estimated profile nuclear statistic as 
\[\hat T(\bm{X}^{obs},\bm{y},\tau)=\min_{{\bm{\beta}}_\tau\in\R^{\abs{\tau}}}\frac{\abs{\{j\in[m]:\hat P_{\tau,{\bm{\beta}}_\tau}(\tilde\tau^{(j)})>\hat P_{\tau,{\bm{\beta}}_\tau}(\tilde\tau(\bm{X}^{obs},\bm{y},\bm{\theta}))\}}}{m}, \]
then the final confidence set for $\tau_0$ becomes 
\[\hat\Gamma_\alpha^{\tau_0}(\bm{X}^{obs},\bm{y}^{obs})=\{\tau:\hat T(\bm{X}^{obs},\bm{y}^{obs},\tau)<1-\alpha,\tau\in\mathcal{C}\}. \]
We summarize the procedure in Algorithm \ref{alg_model_confidence}. 

\begin{algorithm}
\small
\caption{Model Confidence Set under Well-Specified GLMs}\label{alg_model_confidence}
\begin{algorithmic}[1]
\State{\bf Input:} Observed data $(\bm{X}^{obs},\bm{y}^{obs})$, model candidate set $\mathcal{C}$, the number of Monte Carlo samples $m$. 
\State{\bf Output:} Model confidence set $\hat\Gamma_\alpha^{\tau_0}(\bm{X}^{obs},\bm{y}^{obs})$.
\For{$\tau\in\mathcal{C}$}
\State Generate $m$ copies of random noises $\{\bm{\epsilon}^{*(j)}:\epsilon_i^{*(j)}=-g(u_i^{*(j)}),u_i^{*(j)}\overset{{\rm i.i.d.}}{\sim}{\rm Unif}[0,1],i\in[n],j\in[m]\}$.
\State For some ${\bm{\beta}}_\tau$ to be optimized later, compute $\{{\bm Y}^{*(j)}:j\in[m]\}$ with $Y_i^{*(j)}=\1\{X_{i,\tau}^{obs\top}{\bm{\beta}}_\tau+\epsilon_i^{*(j)}>0\}$.
\State For each ${\bm Y}^{*(j)},j\in[m]$, calculate
\[\tilde{\bm{\beta}}^{(j)}(\lambda)=\argmin_{{\bm{\beta}}\in\R^p}-\frac{1}{n}\sum_{i=1}^n\bigg\{Y_i^{*(j)}\log\frac{g^{-1}(X_i^{obs\top}\bm\beta)}{1-g^{-1}(X_i^{obs\top}\bm\beta)}+\log\big(1-g^{-1}(X_i^{obs\top}\bm\beta)\big)\bigg\}+\lambda\norm{{\bm{\beta}}}_1,\]
\[\tilde\tau^{(j)}=\supp(\tilde{\bm{\beta}}^{(j)}(\tilde\lambda^{(j)}(\tau,\bm{\beta}_\tau))),\qquad\tilde\lambda^{(j)}(\tau,\bm{\beta}_\tau)=\argmax_{\lambda\ge0}\norm{\tilde{\bm{\beta}}^{(j)}(\lambda)}_0\quad{\rm s.t. }\norm{\tilde{\bm{\beta}}^{(j)}(\lambda)}_0\le\abs{\tau},\]
\[\tilde{\bm{\beta}}(\lambda)=\argmin_{{\bm{\beta}}\in\R^p}-\frac{1}{n}\sum_{i=1}^n\bigg\{y_i^{obs}\log\frac{g^{-1}(X_i^{obs\top}\bm\beta)}{1-g^{-1}(X_i^{obs\top}\bm\beta)}+\log\big(1-g^{-1}(X_i^{obs\top}\bm\beta)\big)\bigg\}+\lambda\norm{{\bm{\beta}}}_1,\]
\[\tilde\tau(\bm{X}^{obs},\bm{y}^{obs},(\tau,\bm{\beta}_\tau))=\supp(\tilde{\bm{\beta}}(\tilde\lambda(\tau,\bm{\beta}_\tau))),\qquad\tilde\lambda(\tau,\bm{\beta}_\tau)=\argmax_{\lambda\ge0}\norm{\tilde{\bm{\beta}}(\lambda)}_0\quad{\rm s.t. }\norm{\tilde{\bm{\beta}}(\lambda)}_0\le\abs{\tau}.\]

\State Calculate 
\[\hat T(\bm{X}^{obs},\bm{y}^{obs},\tau)=\min_{{\bm{\beta}}_\tau\in\R^{\abs{\tau}}}\frac{\abs{\{j\in[m]:\hat P_{\tau,{\bm{\beta}}_\tau}(\tilde\tau^{(j)})>\hat P_{\tau,{\bm{\beta}}_\tau}(\tilde\tau(\bm{X}^{obs},\bm{y}^{obs},(\tau,\bm{\beta}_\tau)))\}}}{m},\]
with $\hat P_{\tau,{\bm{\beta}}_\tau}(\tau^*)=\frac{1}{m}\sum_{j=1}^m\mathbbm{1}\{\tilde \tau^{(j)}=\tau^*\}$.
\EndFor
\State Construct the model confidence set as
\[\hat\Gamma_\alpha^{\tau_0}(\bm{X}^{obs},\bm{y}^{obs})=\{\tau:\hat T(\bm{X}^{obs},\bm{y}^{obs},\tau)<1-\alpha,\tau\in\mathcal{C}\}.\]
\end{algorithmic}
\end{algorithm}

Let $\tilde\tau^{(j)}=\tilde\tau(\bm X^{obs},\bm Y^{*(j)},\bm\theta_0),j=1,2$ be the selected models obtained from two independent synthetic datasets generated under $\bm\theta_0$, conditional on $\bm X^{obs}$. We use $\delta_n=\Prob_{\bm\theta_0}(\tilde\tau^{(1)}\ne\tilde\tau^{(2)})$ to measure the variability of model selection under the true distribution, where the probability is over both $\bm X^{obs}$ and $\bm Y^{*(1)},\bm Y^{*(2)}$. 
Formally, we have 
the following theorem that guarantees the validity of the 

\noindent
level $(1 - \alpha)$ confidence set $\hat\Gamma_\alpha^{\tau_0}(\bm{y}^{obs})$. A proof is given in the Supplement.

\begin{Theorem}\label{thm_tau}

\begin{enumerate}[1)]
    \item \, If Assumption \ref{ass_signal} holds, $d$ is large enough as required in Theorem \ref{thm_candidate} and $n$ is any fixed number, for $c_{\min},\tilde c_{\min}$ defined in Theorem \ref{thm_candidate}, we have 
    \[\Prob(\tau_0\in\hat\Gamma_\alpha^{\tau_0}(\bm{X},\bm{y}))\ge1-\min\bigg\{\alpha+\sqrt{\frac{(\frac{ep}{s})^s}{4m}}+\sqrt{\frac{\pi}{8m}},\frac{\delta_n}{1-\alpha}\bigg\}-ce^{-cnc_{\min}}\wedge ce^{-cn\tilde c_{\min}}.\]
    \item \, If condition \eqref{eq_signal_strong} holds, $n,d$ are any fixed numbers, for $c_{\min}^*,\tilde c_{\min}^*$ defined in Theorem \ref{thm_candidate_strong_signal}, we have 
    \[\Prob(\tau_0\in\hat\Gamma_\alpha^{\tau_0}(\bm{X},\bm{y}))\ge1-\min\bigg\{\alpha+\sqrt{\frac{(\frac{ep}{s})^s}{4m}}+\sqrt{\frac{\pi}{8m}},\frac{\delta_n}{1-\alpha}\bigg\}-ce^{-cnc_{\min}^*}\wedge ce^{-cn\tilde c_{\min}^*}.\]
\end{enumerate}
\end{Theorem}

\begin{Remark}[Practical implementation of Algorithm \ref{alg_model_confidence}]
    Line 7 in Algorithm \ref{alg_model_confidence} involves the optimization for indicator functions, which could be computationally challenging. This optimization on ${\bm{\beta}}_\tau$ ensures that under the true model $\tau_0$, the statistic $\hat T(\bm{X}^{obs},\bm{y}^{obs},\tau_0)$ is more conservative than $\frac{\abs{\{j\in[m]:\hat P_{\tau_0,{\bm{\beta}}_{0,\tau_0}}(\tilde\tau^{(j)})>\hat P_{\tau_0,{\bm{\beta}}_{0,\tau_0}}(\tilde\tau(\bm{X}^{obs},\bm{y}^{obs},(\tau_0,\bm{\beta}_{0,\tau_0})))\}}}{m}$ which is the oracle statistic when using ${\bm{\beta}}_{0,\tau_0}$ to generate $\bm{Y}^*$. In practice, for any $\tau\in\mathcal{C}$, the MLE of ${\bm{\beta}}_{\tau}$ can also be employed to generate $\bm{Y}^{*(j)}$ since it is a consistent estimator in the low-dimensional setting given $\tau_0$. And our numerical results confirm that the MLE indeed yields confidence sets with reasonable coverages and sizes.
\end{Remark}

\section{Numerical Results}\label{sec_numerical}

In this section, we illustrate the performance of the proposed methods using both simulated and real data examples. Throughout the section, we set the confidence level $1-\alpha$ as 0.95. 

\subsection{Simulation Studies with Synthetic Data}\label{sec_simulation}

\subsubsection{Simulation Study Setup}
In this subsection, we conduct simulation studies to demonstrate the performance of the proposed methods. Throughout this subsection, for $n,p$ to be specified later, we generate $n$ i.i.d. copies $\{X_i:i\in[n]\}$ of $X\in\R^p$ from a normal distribution $N(\bm{0},\Sigma)$ with mean vector $\bm{0}$ and covariance matrix $\Sigma\in\R^{p\times p}$ satisfying $\Sigma_{ij}=0.2^{\abs{i-j}}$. Denote $\bm \gamma=(5,4,3,2.5,0.1,-0.1,\ldots,0.1,-0.1)^\top\in\R^p$, $\bm\omega=(1,-1,\ldots,1,-1)^\top\in\R^p$, and $g(t)=\log\frac{t}{1-t}$. We consider the following four model settings with combinations of mean function $\mu(X)$, sample size $n$, dimension $p$ and the number $d$ of repro samples. After simulating data sets from these models, we use working sparse logistic regression models to fit the data. To evaluate frequentist performance, each simulation is repeated 300 times.

\begin{enumerate}[(M1)]
    \item \, $n=500$, $p=1000$, $d=5000$, 
    \[\mu(X)=\frac{1}{2}+0.95\bigg(g^{-1}(X^\top\bm\gamma)-\frac{1}{2}\bigg)+0.05\bigg(\Phi(X^\top\bm\omega)-\frac{1}{2}\bigg). \]

    \item \, $n=500$, $p=1000$, $d=5000$, 
    \[\mu(X)=\left\{\begin{matrix}
        \max\big\{0,\min\big\{1,g^{-1}(X^\top\bm\gamma)+0.2\big|g^{-1}(X^\top\bm\gamma)-\frac{1}{2}\big|\sin(X^\top\bm{\omega})\big\}\big\},&g^{-1}(X^\top\bm\gamma)\ge\frac{1}{2}\\
        \max\big\{0,\min\big\{1,g^{-1}(X^\top\bm\gamma)+0.2\big|g^{-1}(X^\top\bm\gamma)-\frac{1}{2}\big|\sin(5X^\top\bm{\omega})\big\}\big\},&g^{-1}(X^\top\bm\gamma)<\frac{1}{2}
    \end{matrix}\right.. \]

    \item \, $n=500$, $p=1000$, $d=5000$, 
    \[\mu(X)=g^{-1}(X^\top\bm\beta), \quad\bm\beta=(5,4,3,2.5,0,\ldots,0)^\top\in\R^p. \]
    \item \, $n=900$, $p=1000$, $d=10000$, 
    \[\mu(X)=g^{-1}(X^\top\bm\beta), \quad\bm\beta=(5,4,3,1,0,\ldots,0)^\top\in\R^p. \]
\end{enumerate}
Both models (M1) and (M2) are dense (not sparse) models, and the working logistic regression models are misspecified. However, the first four covariates are more significantly  influential in the response than the other covariates. For (M3) and (M4), the mean functions $\mu(X)$ are indeed sparse logistic models, therefore, the working model is the actual data-generating model.

In Section \ref{sec_model}, we consider the working sparse GLMs at a user-specified sparsity level $s$ and require that the model $\tau_0$ has a stronger signal compared to other models. 
In practice, instead of directly aiming at the model with a user-specified sparsity level $s$, we suggest to set a maximal sparsity level $s_u$ and let the data help us detect the target model size $s$ to be the one that balances the approximation error and model complexity, among all models with size no greater than $s_u$. In particular, in all our implementations in the paper, 
given a dataset of $n$ samples, we adopt the extended BIC (EBIC) \citep{chen2008extended} to select the sparsity $s$, by minimizing 
\begin{equation}\label{eq_s_ebic_ran}-2\sum_{i=1}^nl(\tau,\bm\beta_\tau|X_i,Y_i)+|\tau|\log n+2\log{p\choose|\tau|}. 
\end{equation}
Note that $s$ considered above depends on the observed sample, and therefore is random. 

To evaluate the frequentist coverage of the proposed methods in our simulation studies, we further define a  population level EBIC and choose the sparsity level $s\le s_u$ to minimize 
\begin{equation}\label{eq_s_ebic}
-2n\E l(\tau,\bm\beta_\tau|X,Y)+|\tau|\log n+2\log{p\choose|\tau|}, 
\end{equation}
where $n$ is the observed sample size. The $s$ determined here is then used to define the (population-level) oracle parameters  $(\tau_0,\bm\beta_{0,\tau_0})$ in Section~\ref{sec_model}. 

Since we know the true data-generating models in the simulation studies, we can use a Monte Carlo method to evaluate these population level parameters.  
Specifically, we set the sparsity upper bound as $s_u=10$. To calculate the population level $s$, we generate 50000 samples from the data-generating models to approximate the expectation in \eqref{eq_s_ebic} and find $s=4$ for all models (M1)-(M4).
In the following Figure \ref{fig_fit-size}, we verify the selected sparsity level $s=4$ by generating 50000 samples and applying forward stepwise logistic regression to approximate the relationship between model size and model fitting.  The oracle parameters obtained from the four model settings are:  

\begin{enumerate}[(M1)]
    \item \, $\tau_0=[4]$, $\bm\beta_{0,\tau_0}=(2.03, 1.63, 1.24, 1.04)^\top$. 
    \item \, $\tau_0=[4]$, $\bm{\beta}_{0,\tau_0}=(1.93, 1.52, 1.15, 0.98)^\top$. 
    \item \, $\tau_0=[4]$, $\bm\beta_{0,\tau_0}=(5,4,3,2.5)^\top$. 
    \item \, $\tau_0=[4]$, $\bm\beta_{0,\tau_0}=(5,4,3,1)^\top$. 
\end{enumerate}

 \noindent
In both (M1) and (M2), we can see that the chosen $s=4$ is a reasonable target model size, achieving the optimal balance between model fitting and model size.

\begin{figure}

\small
\centering
\subfloat[M1\label{fig:fit-size_mix}]{
  \includegraphics[width=0.33\textwidth]{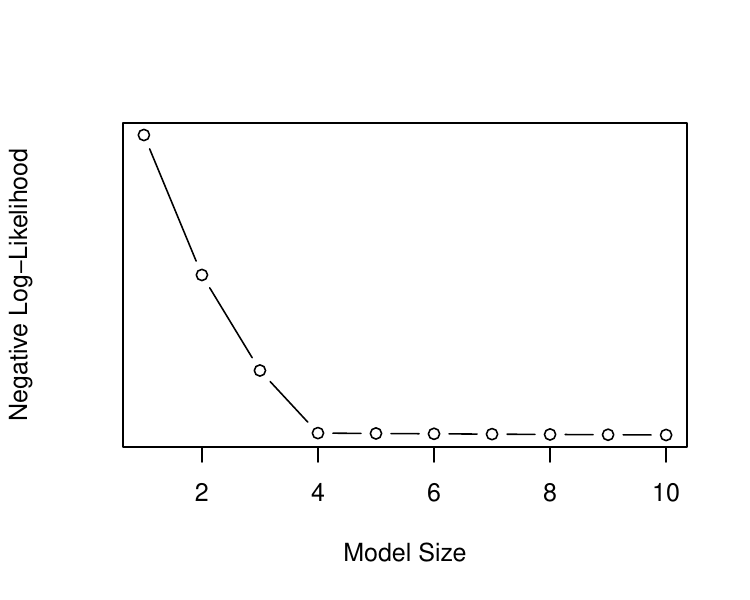}
}
\subfloat[M2\label{fig:fit-size_sin}]{
  \includegraphics[width=0.33\textwidth]{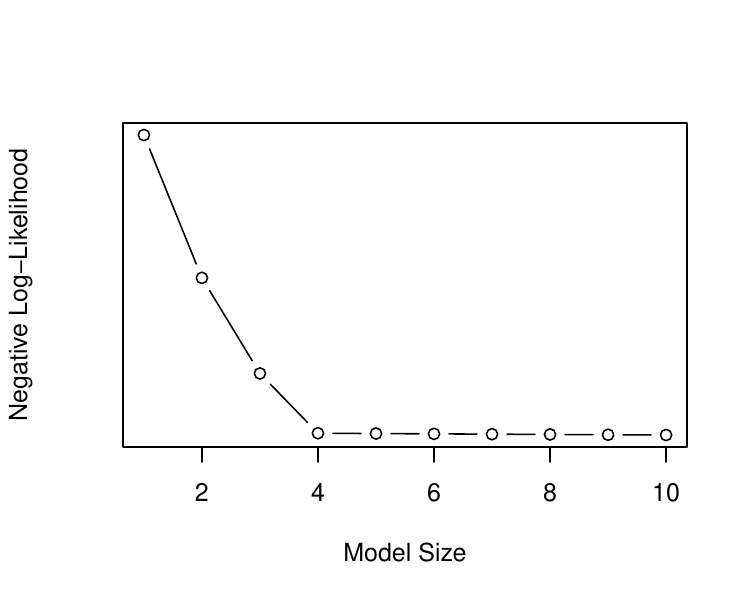}
}
\caption{Plot of negative log-likelihood of working logistic model and model size under (M1) and (M2), respectively. The curves are calculated 
under 50000 samples
using the standard R function~\texttt{step}.}
\label{fig_fit-size}

\end{figure}

When we perform data analysis with only one set of observed data, we do not know the underlying true data-generating model. So we cannot use a Monte Carlo method to evaluate  \eqref{eq_s_ebic} to get the population level $s$. The $s$ used in Algorithm~\ref{alg_candidate} 
is obtained from the empirical EBIC \eqref{eq_s_ebic_ran}. In Section \ref{sec_simulation_candidate}, we show that the candidate models selected based on the $s$ from the empirical EBIC also have a good frequency coverage rate for the oracle $\tau_0$ determined in the previous paragraph.

\subsubsection{Constructing Model Candidate Set}\label{sec_simulation_candidate}

In this section, we study the coverage of our proposed model candidate set for the oracle $\tau_0$ defined

\noindent 
in the previous subsection. 
We show that the proposed model candidate set obtained 
using observed data has a good coverage rate for $\tau_0$.

Line 4 in Algorithm \ref{alg_candidate} requires a 0-1 loss minimization with $\ell_0$ constraint, which is  computationally expensive to solve. In practice, it is common \citep{zhang2004statistical,bartlett2006convexity} to replace the 0-1 loss with a convex surrogate, such as hinge loss or logistic loss. Under these surrogate losses, if we assume $\bm X$ and $\bm \beta_0$ are bounded, then under Assumption \ref{ass_signal}, Theorem \ref{thm_candidate} remains valid. If we replace the 0-1 loss in $\Delta_L$ by the surrogate loss, Theorems \ref{thm_candidate_strong_signal} and \ref{thm_candidate_size} also hold.
    As for the $\ell_0$ penalty, researchers could implement those approaches from \cite{shen2012likelihood} or \cite{bertsimas2016best}, although both approaches are more computationally expensive than a convex penalty function, especially when the dimension is high. On the other hand, commonly used convex penalty functions such as adaptive lasso often provide comparable results in practice with much less computing cost. However, replacing an $\ell_0$ penalty with a convex surrogate would require additional theoretical conditions such as the restricted strong convexity \citep{negahban2012unified}, and larger $\beta$-min conditions \citep{bunea2008honest,zhang2016variable}. Similar considerations are discussed in Remarks 1 and 2 of \cite{wang2022finite}, who develop their repro samples theory using $\ell_0$ penalties or constraints but employ adaptive Lasso and EBIC-type tuning in their numerical implementation. They note that establishing guarantees for a particular surrogate penalty may require additional assumptions, reflecting a trade-off between computational cost and theoretical requirements. Which side of this trade-off to favor is a matter of practical preference. Therefore, in Section \ref{sec_candidate}, we develop our general theory under the $\ell_0$ constraint to offer a fundamental understanding of the procedure and leave the choice of the penalty functions and the surrogate loss functions to the practitioners. Finally, we would like to note that the computational complexity of the proposed repro samples procedure is $O(dp^{3})$, while that of the Debiased Lasso \citep{van2014asymptotically}  is $O(p^{4})$, which makes the proposed procedure more efficient computationally when $p>d$. 

In the numerical study, we replace the $\ell_0$ constrained empirical 0-1 risk minimization problem in Line 4 of Algorithm \ref{alg_candidate} by the following computationally efficient surrogate 
\[(\hat{\bm{\beta}}^{(j)}(\lambda_j),\hat\sigma^{(j)}(\lambda_j))=\argmin_{{\bm{\beta}}\in\R^p,\sigma\ge0}\sum_{i=1}^nL_S((2y_i^{obs}-1)(X_i^{obs\top}{\bm{\beta}}+\sigma\epsilon_i^{*(j)}))+\lambda_j\sum_{k\in[p]}\frac{|\beta_k|}{|\tilde\beta^{(j)}_k|}, \]
\[\hat\tau(\bm{\epsilon}^{*(j)},\lambda_j)=\supp\{\hat{\bm{\beta}}^{(j)}(\lambda_j)\}, 
\]
where we take $L_S$ to be either the logistic loss $L_l$ or hinge loss $L_h$ defined respectively as $L_l(t)=\log(1+e^{-t}),$ $ L_h(t)=\max\{0,1-t\},$  
and we choose $\tilde{\bm{\beta}}^{(j)}$ as the solution of 
\[(\tilde{\bm{\beta}}^{(j)}(\tilde\lambda_j),\tilde\sigma^{(j)}(\tilde\lambda_j))=\argmin_{{\bm{\beta}}\in\R^p,\sigma\ge0}\sum_{i=1}^nL_S((2y_i^{obs}-1)(X_i^{obs\top}{\bm{\beta}}+\sigma\epsilon_i^{*(j)}))+\tilde\lambda_j\norm{{\bm{\beta}}}_2^2, 
\]
for $\tilde\lambda_j$ chosen by 3-fold cross-validation. The tuning parameter $\lambda_j$ is selected using EBIC 
\begin{align*}
    {\rm EBIC}_{j,\xi}(\lambda)=&2\sum_{i=1}^nL_S((2y_i^{obs}-1)(X_i^{obs\top}\hat{\bm{\beta}}^{(j)}(\lambda)+\hat\sigma^{(j)}(\lambda)\epsilon_i^{*(j)}))\\
    &+\abs{\hat\tau(\bm{\epsilon}^{*(j)},\lambda)}\log n+2\xi\log{p \choose\abs{\hat\tau(\bm{\epsilon}^{*(j)},\lambda)}}. 
\end{align*}
Here we choose $\lambda_j(\xi)$ to minimize ${\rm EBIC}_{j,\xi}(\lambda)$ under the sparsity constraint $|\hat\tau(\bm\epsilon^{*(j)},\lambda_j(\xi))|\le s_u$ for each $\xi\in[0,1]$. Therefore for each $\bm{\epsilon}^{*(j)}$, we collect all models $\{\hat\tau(\bm{\epsilon}^{*(j)},\lambda_j(\xi)):\xi\in[0,1]\}$. Then the final model candidate set becomes 
\[\mathcal{C}=\{\hat\tau(\bm{\epsilon}^{*(j)},\lambda_j(\xi)):j\in[d],\xi\in[0,1]\}. 
\]

For the logistic loss $L_l$ and the hinge loss $L_h$, we calculate the model candidate sets with 300 replications and report the averaged coverage of $\tau_0$ and the averaged cardinality of the candidate sets with standard deviations in the parentheses in Table \ref{tab_candidate}. We can see from Table \ref{tab_candidate} that the proposed method performs well for both the misspecified and well-specified models. Based on $d = 5000$ repro samples, the model candidate sets for (M1), (M2), and (M3) achieve nearly 100\% coverage of the target model $\tau_0$ and contain only six candidate models. For the well-specified model (M4) with weak signals, the model candidate sets based on $d =10000$ repro samples attain the desired coverages and contain only four candidate models on average. The average sizes of the candidate sets are within the range of 2.3 - 6.5 models in these examples, indicating a good model selection performance.

\begin{table}
\small
    \centering
    \begin{tabular*}{\columnwidth}{@{\extracolsep\fill}lcccc@{\extracolsep\fill}}\hline
         & \multicolumn{4}{c}{Losses}\\\cline{2-5}
        & \multicolumn{2}{c}{Hinge} & \multicolumn{2}{c}{Logistic} \\\cline{2-3}\cline{4-5}
        Models  & Coverage & Cardinality & Coverage & Cardinality \\\hline

         M1 & 0.99(0.11) & 4.79(2.18) & 0.98(0.15) & 3.92(2.96)\\
         
         M2 & 0.99(0.10) & 4.94(2.33) & 0.98(0.15) & 3.59(2.52)\\
         
         M3 & 0.99(0.11) & 6.42(2.58) & 0.99(0.11) & 5.86(3.25)\\
         
         M4 & 0.98(0.15) & 4.38(2.20) & 0.99(0.08) & 2.38(1.43)\\
         \hline
    \end{tabular*}
    \caption{Comparison of performance of the model candidate sets. Here ``Coverage" means the probability that the model candidate set $\mathcal{C}$ contains $\tau_0$, and ``Cardinality" indicates the number of models in $\mathcal{C}$.}
    \label{tab_candidate}
    
\end{table}

\subsubsection{Inference for individual \texorpdfstring{$\beta_{0,j}$}{beta\_0,j} While Accommodating Model Uncertainty}\label{sec_sim_betaj}

In this subsection, we study the performance of the confidence sets for individual coefficients $\beta_{0,j}$ for $j\in[p]$. We compare our method with the oracle Wald test assuming $\tau_0$ were known. For the well-specified models (M3) and (M4), we also compare with the standard Debiased Lasso method in \cite{van2014asymptotically} implemented using the \texttt{lasso.proj} function in the \texttt{hdi} package. 

For models (M1) and (M2), the sparse logistic model is misspecified. As we demonstrated in Remark \ref{rem_zero_coefficient}, $\beta_{0,j}=0$ for $j\in[p]\setminus\tau_0$ in (M1) and (M2) do not imply the lack of association between $X_j$ and $Y$, but merely indicates that $X_j$ contributes less to $Y$ relative to those included in $X_{\tau_0}$. Consequently, $\beta_{0,j}=0$ for $j\in[p]\setminus\tau_0$ doesn't have a quantitative meaning. Therefore, for models (M1) and (M2), we only calculate the coverage and size of confidence sets for $\beta_{0,j},j\in\tau_0$, and then we average the performance over $j\in\tau_0$. For the well-specified models (M3) and (M4), we also report the confidence sets for $\beta_{0,j},j\in[p]\setminus\tau_0$. Here, to accommodate the model selection uncertainty, the confidence sets for $\beta_{0,j}$ are a union of intervals, so we report the Lebesgue measure of the confidence~sets.

The numerical results reported in Table \ref{tab_betai_conf} contain the averaged coverages and sizes of confidence sets over 300 replications with standard deviations in the parentheses.
As we discussed in Section \ref{sec_simulation_candidate}, we consider two losses, logistic loss and hinge loss, for Line 4 in Algorithm \ref{alg_candidate}. Hereafter, we use the abbreviations ``Repro-Logistic" and ``Repro-Hinge" to denote the repro samples method with logistic loss and hinge loss, respectively. We also use ``Debias" to denote the Debiased Lasso method and use ``Oracle" to denote the oracle Wald test with the knowledge of $\tau_0$. From Table \ref{tab_betai_conf}, we see that for $j\in\tau_0$, the proposed methods Repro-Hinge and Repro-Logistic and the Oracle method have the desired coverage of 0.95 for all the models. However, the Debiased method seriously undercovers the nonzero coefficients in (M3) and (M4),  which is a phenomenon that is also reported in high-dimensional linear models \citep[see, e.g.,][]{wang2022finite}. Note that, since the inference targets of the Debiased Lasso method are not well defined in the two non-GLM high-dimensional settings (M1) and (M2), Table~\ref{tab_betai_conf} does not include any comparison with the Debiased Lasso method for (M1) and (M2). 
In terms of size, the confidence sets produced by Repro-Hinge and Repro-Logistic are comparable to those of the Oracle method, but the sizes of the intervals calculated by the Debiased Lasso method are even shorter than those of the Oracle method, so are likely to undercover. For the zero coefficients with $j\in[p]\setminus\tau_0$ in (M3) and (M4), Repro-Hinge, Repro-Logistic, and Debiased Lasso all have coverage rates of 1, but the sizes corresponding to Repro-Hinge and Repro-Logistic are shorter than the sizes corresponding to Debiased Lasso. The reason is that Repro-Hinge and Repro-Logistic also account for the uncertainty of the selected models. When no models in the candidate set contain $j$, we estimate $\beta_{0,j}$ by 0 with confidence 1. 

\begin{table}
\small
    \centering
    \begin{tabular*}{\columnwidth}{@{\extracolsep\fill}lccccc@{\extracolsep\fill}}\hline
        &&\multicolumn{2}{c}{$\beta_{0,j},j\in\tau_0$}&\multicolumn{2}{c}{$\beta_{0,j},j\in[p]\setminus\tau_0$}\\\cline{3-4}\cline{5-6}
        Model & Method & Coverage & Length & Coverage & Length \\\hline
        M1 & Repro-Hinge & 0.96(0.12) & 0.95(0.12) &  & \\
        & Repro-Logistic & 0.96(0.12) & 0.94(0.13) &  & \\
        & Oracle & 0.95(0.13) & 0.84(0.09) & &\\[6pt]
        M2 & Repro-Hinge & 0.96(0.11) & 0.93(0.12) &  & \\
        & Repro-Logistic & 0.96(0.12) & 0.91(0.13) &  & \\
        & Oracle & 0.95(0.13) & 0.83(0.09) & &\\[6pt]
        M3 & Repro-Hinge & 0.97(0.11) & 2.59(0.72) & 1.00(0.00) & 0.00(0.00)\\
        & Repro-Logistic & 0.97(0.11) & 2.72(0.93) & 1.00(0.00) & 0.00(0.00)\\
        & Debias & 0.09(0.23) & 0.87(0.17) & 1.00(0.00) & 0.72(0.14) \\
        & Oracle & 0.93(0.18) & 1.98(0.40) & &\\[6pt]
        M4 & Repro-Hinge & 0.96(0.15) & 1.52(0.25) & 1.00(0.00) & 0.00(0.00)\\
        & Repro-Logistic & 0.95(0.15) & 1.42(0.23) & 1.00(0.00) & 0.00(0.00)\\
        & Debias & 0.14(0.25) & 0.64(0.06) & 0.99(0.00) & 0.51(0.05)\\
        & Oracle & 0.94(0.17) & 1.30(0.17) & &\\\hline
    \end{tabular*}
    \caption{Comparison of performance of the confidence sets of $\beta_{0,j}$. Here ``Coverage" means the probability for $\Gamma_\alpha^{\beta_{0,j}}(\bm{X}^{obs},\bm{y}^{obs})$ to contain $\beta_{0,j}$, and ``Length" means the Lebesgue measure of $\Gamma_\alpha^{\beta_{0,j}}(\bm{X}^{obs},\bm{y}^{obs})$. The third and fourth columns correspond to $j\in\tau_0$, and the last two columns correspond to $j\in[p]\setminus\tau_0$.}
    \label{tab_betai_conf}
    
\end{table}

\subsubsection{Inference for vector  \texorpdfstring{$\bm{\beta}_0$ While Accommodating Model Uncertainty
}{beta\_0}}

We also study the performance of the proposed method for simultaneous inference of the entire $p \times 1$ vector parameter $\bm\beta_0$. Since it is difficult to calculate the size of the confidence sets in $\R^p$, we report only the coverage rates for the Repro-Hinge, Repro-Logistic, and the Oracle methods in Table \ref{tab_beta_conf}. Repro-Hinge (or Repro-Logistic) refers to Algorithm \ref{alg_coefficient_confidence} with model candidate sets constructed using hinge (or logistic) loss for line 4 in Algorithm \ref{alg_candidate}, while Oracle method conducts Wald tests with the additional knowledge of $\tau_0$.  From Table \ref{tab_beta_conf}, we can see that Repro-Hinge and Repro-Logistic have similar performance to that of the Oracle method, with coverage rates between 0.90 and 0.94. 
The Oracle method exhibits slight undercoverage in (M1)-(M4), perhaps because $\bm\beta_{0,\tau_0}$ is a four-dimensional vector 
and the sample sizes are relatively limited. Likewise, the proposed methods for the entire vector 
$\bm\beta_{0}$ also exhibit slight undercoverage. The proposed methods and Oracle method appear to perform similarly. 

\begin{table}
\small
    \centering
    \begin{tabular*}{\columnwidth}{@{\extracolsep\fill}lccc@{\extracolsep\fill}}\hline
        & \multicolumn{3}{c}{Coverage}\\\cline{2-4}
         Models & Repro-Hinge & Repro-Logistic & Oracle \\\hline
         M1 & 0.93(0.25) & 0.92(0.27) & 0.94(0.24) \\
         M2 & 0.91(0.28) & 0.90(0.30) & 0.92(0.27)\\
         M3 & 0.91(0.29) & 0.91(0.28) & 0.92(0.27)\\
         M4 & 0.92(0.27) & 0.94(0.24) & 0.94(0.23)\\\hline
    \end{tabular*}
    \caption{Comparison of performance of the confidence sets of $\bm{\beta}_0$. Here ``Coverage" means the probability for $\Gamma_\alpha^{\bm{\beta}_0}(\bm{X}^{obs},\bm{y}^{obs})$ to contain $\bm{\beta}_0$.}
    \label{tab_beta_conf}
    
\end{table}

\subsubsection{Inference for Case Probabilities While Accommodating Model Uncertainty}

To evaluate the performance of our proposed method for simultaneous inference of case probabilities, we construct $\bm{X}_{\rm new}$ as follows. For (M1)-(M4), the number of new observations is set to $n_{\rm new}=2~{\rm or}~2000$. Then for each of the models, we generate i.i.d. random vectors  $X_{{\rm new},i}\in\R^p$ from a normal distribution $N(\bm{0},\Sigma)$ with the covariance matrix $\Sigma$ satisfying $\Sigma_{ij}=0.2^{|i-j|}$. Since it is hard to measure the volume of the confidence sets, we again only report the coverage rates of the Oracle (assuming known $\tau_0$), Repro-Hinge and Repro-Logistic methods in Table~\ref{tab_case_conf}. The~results in Table~\ref{tab_case_conf} reveal that both Repro-Hinge and Repro-Logistic have the similar performance to the Oracle method, with coverage rates close to the nominal value 0.95, even though the Oracle method uses the knowledge of $\tau_0$. Notably, when $n_{\rm new}=2$, we have ${\rm rank}(\bm X_{\rm new})\le 2$. In this case, the effective parameter has dimension at most 2, which is lower than 4 as in Table \ref{tab_beta_conf}. Consequently, both the proposed methods and the Oracle method achieve better coverage. In contrast, when $n_{\rm new}=2000>p$, it is typical that ${\rm rank}(\bm X_{\rm new})=p$. Hence, testing $\bm X_{\rm new}\bm\beta_0$ is equivalent to testing $\bm\beta_0$. Accordingly, the coverages for $g^{-1}(\bm X_{\rm new}\bm\beta_0)$, listed in Table \ref{tab_case_conf}, are identical to those~for~$\bm\beta_0$~in~Table~\ref{tab_beta_conf}.

\begin{table}
\small
    \centering
    \begin{tabular*}{\columnwidth}{@{\extracolsep\fill}lcccc@{\extracolsep\fill}}\hline
        & & \multicolumn{3}{c}{Coverage}\\\cline{3-5}
         $n_{\rm new}$ & Models & Repro-Hinge & Repro-Logistic & Oracle \\\hline
         2 & M1 & 0.98(0.14) & 0.96(0.19) & 0.97(0.18)\\
         & M2 & 0.98(0.15) & 0.97(0.17) & 0.95(0.23)\\
         & M3 & 0.98(0.15) & 0.97(0.17) & 0.94(0.23)\\
         & M4 & 0.96(0.20) & 0.95(0.23) & 0.94(0.24)\\[6pt]
         
         2000 & M1 & 0.93(0.25) & 0.92(0.27) & 0.94(0.24) \\
         & M2 & 0.91(0.28) & 0.90(0.30) & 0.92(0.27)\\
         & M3 & 0.91(0.29) & 0.91(0.28) & 0.92(0.27)\\
         & M4 & 0.92(0.27) & 0.94(0.24) & 0.94(0.23)\\\hline
    \end{tabular*}
    \caption{Comparison of performance of the confidence sets of $g^{-1}(\bm{X}_{\rm new}\bm{\beta}_0)$. Here ``Coverage" means the probability for $\Gamma_\alpha^{g^{-1}(\bm{X}_{\rm new}\bm{\beta}_0)}(\bm{X}^{obs},\bm{y}^{obs})$ to contain $g^{-1}(\bm{X}_{\rm new}\bm{\beta}_0)$.}
    \label{tab_case_conf}
    
\end{table}

\subsubsection{Inference for $\tau_0$}

In this subsection, we focus on settings (M3) and (M4) where the data-generating mean function $\mu(X)$ is indeed a sparse logistic regression model. We study the performance of the model confidence set proposed in Section \ref{sec_tau}. The optimization problem in Algorithm \ref{alg_model_confidence}, in Line 7, is solved by the \texttt{optim} function in \texttt{R} using the method in \cite{nelder1965simplex}.
To speed up the computing time, we have also avoided this optimization by using the MLE of ${\bm{\beta}}_\tau$ to generate $\bm{Y}^{*(j)}$ in Algorithm~\ref{alg_model_confidence}. Their results, indicated as Profile and MLE, respectively, are reported in Table~\ref{tab_tau_conf}, which was computed 
using $m=500$ Monte Carlo samples. The results include the coverages and cardinalities of the model confidence sets for the true $\tau_0$.
From Table \ref{tab_tau_conf}, we find that the model confidence sets are smaller than the model candidate sets while the coverages of the model confidence sets are the same as those of the model candidate sets reported in Table~\ref{tab_candidate} for both (M3) and (M4). Due to the discreteness of model space, the model confidence sets are conservative.  However, the approach is still able to exclude some candidate models to produce smaller model confidence sets in all cases.

\begin{table}
\small
    \centering
    \begin{tabular*}{\columnwidth}{@{\extracolsep\fill}lccccc@{\extracolsep\fill}}\hline
        & & \multicolumn{4}{c}{Nuisance}\\\cline{3-6}
        & & \multicolumn{2}{c}{Profile} & \multicolumn{2}{c}{MLE}\\\cline{3-4}\cline{5-6}
        Models & Losses   & Coverage & Cardinality  & Coverage & Cardinality \\\hline
        M3 & Hinge & 0.99(0.11) & 5.46(2.37) & 0.99(0.11) & 4.62(2.18)\\
         & Logistic & 0.99(0.11) & 5.08(2.78) & 0.99(0.11) & 4.38(2.34)\\[6pt]
         M4 & Hinge & 0.98(0.15) & 3.59(1.75) & 0.98(0.15) & 3.12(1.59)\\
         & Logistic & 0.99(0.08) & 2.23(1.29) & 0.99(0.08) & 2.06(1.16)\\\hline
    \end{tabular*}
   \caption{Comparison of performance of the model confidence sets. Here ``Coverage" means the probability for the model confidence set $\Gamma_\alpha^{\tau_0}(\bm{X}^{obs},\bm{y}^{obs})$ to contain $\tau_0$, and ``Cardinality" indicates the number of models in $\Gamma_\alpha^{\tau_0}(\bm{X}^{obs},\bm{y}^{obs})$.}
    \label{tab_tau_conf}
    
\end{table}

\subsection{Real Data Analysis}\label{sec_real_data}

In this section, we consider a high-dimensional real data analysis on single-cell RNA-seq data from \cite{shalek2014single}. Most existing methods focus on statistical inference for single coefficients, while our method can also quantify the uncertainty of model selection. As will be demonstrated, the Debiased Lasso method identifies only one variable as significant. In contrast, our model confidence sets find several variables that have been shown to be important by many existing biological studies. 

The single-cell RNA-seq data set from \cite{shalek2014single} comprises gene expression profiles for 27723 genes across 1861 primary mouse bone-marrow-derived dendritic cells spanning several experimental conditions. Specifically, we focus on a subset of the data consisting of 96 cells stimulated by the pathogenic component PIC (viral-like double-stranded RNA) and 96 control cells without stimulation, with gene expressions measured six hours after stimulation. In our study, we label each cell with 0 and 1 to indicate ``unstimulated" and ``stimulated" states, respectively. Our goal is to investigate the association between gene expression and stimulation status. Similar to \cite{cai2021statistical}, we filter out genes that are not expressed in more than 80\% of the cells and discard the bottom 90\% genes with the lowest variances. Subsequently, we log-transform and normalize the gene expressions to have mean 0 and unit variance. The resulting dataset consists of 192 samples with 697 covariates. In the high-dimensional binary classification problem, we do not assume model sparsity; rather, we focus on identifying a set of the most influential genes.

Using the same parameter tuning strategy as detailed in Section \ref{sec_simulation} with 1000 repro samples, Repro-Hinge and Repro-Logistic identify 7 and 12 candidate models, respectively. We list them in Table \ref{tab_real}. Most of the identified genes have been previously associated with the immune system. RSAD2 is involved in antiviral innate immune responses, and is also a powerful stimulator of adaptive immune response mediated via mDCs \citep{jang2018rsad2}. IFIT1 inhibits viral replication by binding viral RNA that carries PPP-RNA \citep{pichlmair2011ifit1}. IFT80 is known to be an essential component for the development and maintenance of motile and sensory cilia \citep{wang2018ift80}, while ciliary machinery is repurposed by T cells to focus the signaling protein LCK at immune synapse \citep{stephen2018ciliary}. BC044745 has been identified as significant in MRL/MpJ mouse, which exhibits distinct gene expression patterns involved in immune response \citep{podolak2015transcriptome}. ACTB has shown associations with immune cell infiltration, immune checkpoints, and other immune modulators in most cancers \citep{gu2021pan}. HMGN2 has been validated to play an important role in the innate immune system during pregnancy and development in mice \citep{deng2012chromosomal}. Finally, IFI47, also known as IRG47, has been shown to be vital for immune defense against protozoan and bacterial infections \citep{collazo2001inactivation}.

Regarding confidence sets for individual genes, we compare the proposed Repro-Hinge and Repro-Logistic methods with the debiased approach. 
Repro-Hinge identifies RSAD2 and AK217941 as significant, while both Repro-Logistic and Debiased Lasso only identify RSAD2 as significant. 
While RSAD2 plays an important role in antiviral innate immune responses, AK217941, though not studied in the literature, deserves further attention as it has been identified in both model candidate sets and single coefficient confidence sets.

\begin{table}
\small
    \centering
    \begin{tabular*}{\columnwidth}{@{\extracolsep\fill}lccccccccccccccccccc@{\extracolsep\fill}}\hline
        & \multicolumn{17}{c}{Methods}\\\cline{2-20}
        Genes & \multicolumn{7}{c}{Repro-Hinge} & \multicolumn{12}{c}{Repro-Logistic} \\\cline{1-1}\cline{2-8}\cline{9-20}
         RSAD2 & \solidcirc & \solidcirc & \solidcirc & \solidcirc & \solidcirc & \solidcirc & \solidcirc & \hollowcirc & \hollowcirc & \hollowcirc & \hollowcirc & \hollowcirc & \hollowcirc & \hollowcirc & \hollowcirc & \hollowcirc & \hollowcirc & \hollowcirc & \hollowcirc  \\
         AK217941 & \solidcirc & \solidcirc & \solidcirc & \solidcirc & \solidcirc & \solidcirc & \solidcirc & \hollowcirc & \hollowcirc & \hollowcirc & \hollowcirc & & \hollowcirc & \hollowcirc & \hollowcirc & \hollowcirc & \hollowcirc & \hollowcirc  & \hollowcirc \\
         IFIT1 & &&&& \solidcirc& \solidcirc & &\hollowcirc &\hollowcirc &  &\hollowcirc &  &\hollowcirc &  &\hollowcirc  &\hollowcirc &\hollowcirc & \hollowcirc  & \hollowcirc \\
         IFT80 &&&&&&&&& \hollowcirc & & \hollowcirc & & \hollowcirc& \hollowcirc& \hollowcirc& \hollowcirc& \hollowcirc & \hollowcirc  & \hollowcirc \\
         BC044745 & & & \solidcirc & \solidcirc & \solidcirc & \solidcirc & \solidcirc & &&& \hollowcirc & & && \hollowcirc & & \hollowcirc & & \\
         ACTB & & \solidcirc & \solidcirc & & & \solidcirc & \solidcirc & & & & & & & & & \hollowcirc &\hollowcirc & & \\
        HMGN2 & & & & & & & \solidcirc & & & & & & & & \hollowcirc & & & \hollowcirc & \\
         IFI47 & & & & & & & & & & & & & \hollowcirc & & & & & & \\
         ZFP488 & & & & & & & & & & & & & & & & & & & \hollowcirc \\\hline
    \end{tabular*}
    \caption{All the models in the model candidate sets. Each row represents a gene while each column corresponds to a model. The circle in the $i$-th row and $j$-th column indicates that the $i$-th gene appears in the $j$-th model.}
    \label{tab_real}
    
\end{table}

\section{Discussions}\label{sec:disc}
This paper develops a model-free yet interpretable approach to high-dimensional binary~classification. Rather than assuming  the data-generating mechanism follows a particular regression~structure, we use a  working sparse GLM as an interpretable representation of the unknown~relationship~between the outcome and covariates. The oracle model and coefficients are defined through this working model, allowing inference to remain meaningful even when the working GLM is misspecified. This perspective distinguishes us from methods that rely on a correctly specified parametric model or use black-box methods without providing uncertainty quantification for an interpretable model.

A key feature of the proposed framework is that it quantifies uncertainty in both model selection and coefficient estimation. The repro samples framework yields a candidate set for the most influential covariates with guaranteed coverage under a weak signal condition, together with confidence sets for arbitrary linear transformations of the oracle coefficients. When the working GLM coincides with the true model, these results provide simultaneous inference for the true support and coefficients. More generally, they provide uncertainty quantification for an interpretable sparse representation without requiring sparsity of the underlying data-generating mechanism.

The signal-strength condition is an important aspect of the proposed model-free variable identification. Some separation between influential and negligible covariates is necessary for reliable identification in a general high-dimensional setting, without which uniformly valid identification~is generally impossible. 
Future work 
could extend the framework to more flexible working-model classes. 
Our results here suggest that model-free inference need not come at~the expense of interpretability: a working model can serve as an inferential representation~of important covariate– response relationships without being regarded as a literal model for the data-generating mechanism.

\section{Data Availability}
The single-cell RNA-seq data analyzed in Section \ref{sec_real_data} are publicly available from the NCBI Gene Expression Omnibus (GEO) under accession GSE48968 \citep{shalek2014single}. The computing code of the paper 
is available at \url{https://github.com/XiaotianHou/Repro-Samples_HDClassification}.

\section{Funding} 

This work was supported in part by the National Science Foundation (NSF) under Grants DMS2015373, DMS2027855, DMS2311064, DMS2319260, and DMS2515766, and by the National Institutes of Health (NIH) under Grant 1R01GM157610-1.
%

\bibliographystyle{abbrvnat}
\bibliography{reference}

@article{zhang2016variable,
  title={Variable selection for support vector machines in moderately high dimensions},
  author={Zhang, Xiang and Wu, Yichao and Wang, Lan and Li, Runze},
  journal={J. R. Stat. Soc. Series. B Stat. Methodol.},
  volume={78},
  number={1},
  pages={53--76},
  year={2016},
  publisher={Oxford University Press}
}

@article{negahban2012unified,
  title={A Unified Framework for High-Dimensional Analysis of M-Estimators with Decomposable Regularizers},
  author={Negahban, Sahand  and Ravikumar, Pradeep and Wainwright, Martin J and Yu, Bin},
  journal={Stat. Sci.},
  volume={27},
  pages={538--557},
  year={2012}
}

@article{tomaselli2025robust,
  title={Robust Simulation Based Inference},
  author={Tomaselli, Lorenzo and Ventura, Val{\'e}rie and Wasserman, Larry},
  journal={arXiv:2508.02404},
  year={2025}
}

@article{rudin2019stop,
  title={Stop explaining black box machine learning models for high stakes decisions and use interpretable models instead},
  author={Rudin, Cynthia},
  journal={Nat. Mach. Intell.},
  volume={1},
  number={5},
  pages={206--215},
  year={2019},
  publisher={Nature Publishing Group UK London}
}

@article{shah2023double,
  title={Double-estimation-friendly inference for high-dimensional misspecified models},
  author={Shah, Rajen D and B{\"u}hlmann, Peter},
  journal={Stat. Sci.},
  volume={38},
  number={1},
  pages={68--91},
  year={2023},
  publisher={Institute of Mathematical Statistics}
}

@article{buhlmann2015high,
  title={High-dimensional inference in misspecified linear models},
  author={B{\"u}hlmann, Peter and van de Geer, Sara},
  journal={Electron. J. Stat.},
  volume={9},
  pages={1449--1473},
  year={2015}
}

@article{chen2023testing,
  title={Testing generalized linear models with high-dimensional nuisance parameters},
  author={Chen, Jinsong and Li, Quefeng and Chen, Hua Yun},
  journal={Biometrika},
  volume={110},
  number={1},
  pages={83--99},
  year={2023},
  publisher={Oxford University Press}
}

@article{hong2024inference,
  title={Inference for possibly misspecified generalized linear models with nonpolynomial-dimensional nuisance parameters},
  author={Hong, Shaoxin and Jiang, Jiancheng and Jiang, Xuejun and Wang, Haofeng},
  journal={Biometrika},
  volume={111},
  pages={1387--1404},
  year={2024},
  publisher={Oxford University Press}
}

@article{schmitt1992perturbation,
  title={Perturbation bounds for matrix square roots and Pythagorean sums},
  author={Schmitt, Bernhard A},
  journal={Linear algebra and its applications},
  volume={174},
  pages={215--227},
  year={1992},
  publisher={Elsevier}
}

@inproceedings{ravi2019detection,
  title={Detection of infectious disease using non-invasive logistic regression technique},
  author={Ravi, Anirudhh and Gopal, Varun and Roselyn, J Preetha and Devaraj, D and Chandran, Pranav and Madhura, R Sai},
  booktitle={2019 IEEE International Conference on Intelligent Techniques in Control, Optimization and Signal Processing (INCOS)},
  year={2019},
  organization={IEEE}
}

@article{chadwick2006distinguishing,
  title={Distinguishing dengue fever from other infections on the basis of simple clinical and laboratory features: application of logistic regression analysis},
  author={Chadwick, David and Arch, Barbara and Wilder-Smith, Annelies and Paton, Nicholas},
  journal={J. Clin. Virol.},
  year={2006},
  publisher={Elsevier}
}

@article{kuchibhotla2022moving,
  title={Moving beyond sub-Gaussianity in high-dimensional statistics: Applications in covariance estimation and linear regression},
  author={Kuchibhotla, Arun and Chakrabortty, Abhishek},
  journal={Inf. Inference},
  year={2022},
  publisher={Oxford University Press}
}

@article{collazo2001inactivation,
  title={Inactivation of LRG-47 and IRG-47 reveals a family of interferon $\gamma$--inducible genes with essential, pathogen-specific roles in resistance to infection},
  author={Collazo, Carmen and Yap, George and Sempowski, Gregory and Lusby, Kimberly and Tessarollo, Lino and Woude, George and Sher, Alan and Taylor, Gregory},
  journal={J. Exp. Med},
  volume={194},
  number={2},
  pages={181--188},
  year={2001},
  publisher={The Rockefeller University Press}
}

@article{deng2012chromosomal,
  title={The chromosomal protein HMGN2 mediates the LPS-induced expression of $\beta$-defensins in mice},
  author={Deng, Lu-Xia and Wu, Gui-Xia and Cao, Yue and Fan, Bo and Gao, Xiang and Tang, Xiao-Hai and Huang, Ning},
  journal={Inflamm.},
  volume={35},
  number={2},
  pages={456--473},
  year={2012},
  publisher={Springer}
}

@article{gu2021pan,
  title={A pan-cancer analysis of the prognostic and immunological role of $\beta$-actin (ACTB) in human cancers},
  author={Gu, Yuxi and Tang, Shouyi and Wang, Zhen and Cai, Luyao and Lian, Haosen and Shen, Yingqiang and Zhou, Yu},
  journal={Bioengineered},
  year={2021},
  volume={12},
  number={1},
  pages={6166--6185},
  publisher={Taylor \& Francis}
}

@article{podolak2015transcriptome,
  title={Transcriptome profiling reveals distinctive traits of retinol metabolism and neonatal parallels in the MRL/MpJ mouse},
  author={Podolak-Popinigis, Justyna and G{\'o}rnikiewicz, Bartosz and Ronowicz, Anna and Sachadyn, Pawe{\l}},
  journal={BMC Genomics},
  year={2015},
  volume={16},
  number={1},
  pages={926},
  publisher={BioMed Central}
}

@article{stephen2018ciliary,
  title={The ciliary machinery is repurposed for T cell immune synapse trafficking of LCK},
  author={Stephen, Louise and ElMaghloob, Yasmin and McIlwraith, Michael and Yelland, Tamas and Sanchez, Patricia and Roda-Navarro, Pedro and Ismail, Shehab},
  journal={Dev. Cell},
  year={2018},
  volume={47},
  number={1},
  pages={122--132},
  publisher={Elsevier}
}

@article{wang2018ift80,
  title={IFT80 improves invasion ability in gastric cancer cell line via ift80/p75NGFR/MMP9 signaling},
  author={Wang, Rui and Deng, Xiaoyan and Yuan, Chengfu and Xin, Hongmei and Liu, Geli and Zhu, Yong and Jiang, Xue and Wang, Changdong},
  journal={Int. J. Mol. Sci.},
  year={2018},
  volume={19},
  number={11},
  pages={3616},
  publisher={MDPI}
}

@article{pichlmair2011ifit1,
  title={IFIT1 is an antiviral protein that recognizes 5'-triphosphate RNA},
  author={Pichlmair, Andreas and Lassnig, Caroline and Eberle, Carol-Ann and G{\'o}rna, Maria and Baumann, Christoph and Burkard, Thomas and B{\"u}rckst{\"u}mmer, Tilmann and Stefanovic, Adrijana and Krieger, Sigurd and Bennett, Keiryn and others},
  journal={Nat. Immunol.},
  year={2011},
  volume={12},
  number={7},
  pages={624--630},
  publisher={Nature Publishing Group US New York}
}

@article{jang2018rsad2,
  title={Rsad2 is necessary for mouse dendritic cell maturation via the IRF7-mediated signaling pathway},
  author={Jang, Ji-Su and Lee, Jun-Ho and Jung, Nam-Chul and Choi, So-Yeon and Park, Soo-Yeoun and Yoo, Ji-Young and Song, Jie-Young and Seo, Han Geuk and Lee, Hyun Soo and Lim, Dae-Seog},
  journal={Cell Death Dis.},
  year={2018},
  volume={9},
  number={8},
  pages={823},
  publisher={Nature Publishing Group UK London}
}

@article{shalek2014single,
  title={Single-cell RNA-seq reveals dynamic paracrine control of cellular variation},
  author={Shalek, Alex and Satija, Rahul and Shuga, Joe and Trombetta, John and Gennert, Dave and Lu, Diana and Chen, Peilin and Gertner, Rona and Gaublomme, Jellert and Yosef, Nir and others},
  journal={Nature},
  volume={510},
  number={7505},
  pages={363--369},
  year={2014},
  publisher={Nature Publishing Group UK London}
}

@article{nelder1965simplex,
  title={A simplex method for function minimization},
  author={Nelder, John and Mead, Roger},
  journal={Comput. J.},
  volume={7},
  pages={308--313},
  year={1965},
  publisher={The British Computer Society}
}

@article{berend2012convergence,
  title={On the convergence of the empirical distribution},
  author={Berend, Daniel and Kontorovich, Aryeh},
  journal={arXiv:1205.6711},
  year={2012}
}

@article{zhang2017simultaneous,
  title={Simultaneous inference for high-dimensional linear models},
  author={Zhang, Xianyang and Cheng, Guang},
  journal={J. Am. Stat. Assoc.},
  year={2017},
  volume={112},
  number={518},
  pages={757--768},
  publisher={Taylor \& Francis}
}

@article{wang2017operator,
  title={Operator norm inequalities between tensor unfoldings on the partition lattice},
  author={Wang, Miaoyan and Duc, Khanh Dao and Fischer, Jonathan and Song, Yun S},
  journal={Linear Algebra Appl.},
  year={2017},
  publisher={Elsevier}
}

@article{zhilova2022new,
  title={New Edgeworth-type expansions with finite sample guarantees},
  author={Zhilova, Mayya},
  journal={Ann. Stat.},
  year={2022},
  publisher={Institute of Mathematical Statistics}
}

@article{chen2008extended,
  title={Extended Bayesian information criteria for model selection with large model spaces},
  author={Chen, Jiahua and Chen, Zehua},
  journal={Biometrika},
  year={2008},
  volume={95},
  number={3},
  pages={759--771},
  publisher={Oxford University Press}
}

@article{fan2021shrinkage,
  title={A shrinkage principle for heavy-tailed data: High-dimensional robust low-rank matrix recovery},
  author={Fan, Jianqing and Wang, Weichen and Zhu, Ziwei},
  journal={Ann. Stat.},
  year={2021},
  publisher={NIH Public Access}
}

@article{jin2019short,
  title={A short note on concentration inequalities for random vectors with subgaussian norm},
  author={Jin, Chi and Netrapalli, Praneeth and Ge, Rong and Kakade, Sham and Jordan, Michael},
  journal={arXiv:1902.03736},
  year={2019}
}

@article{hansen2011model,
  title={The model confidence set},
  author={Hansen, Peter and Lunde, Asger and Nason, James},
  journal={Econometrica},
  year={2011},
  volume={79},
  number={2},
  pages={453--497},
  publisher={Wiley Online Library}
}

@article{sur2019modern,
  title={A modern maximum-likelihood theory for high-dimensional logistic regression},
  author={Sur, Pragya and Cand{\`e}s, Emmanuel},
  journal={Proc. Natl. Acad. Sci. U.S.A.},
  year={2019},
  volume={116},
  number={29},
  pages={14516--14525},
  publisher={National Acad Sciences}
}

@article{shi2019linear,
  title={Linear hypothesis testing for high dimensional generalized linear models},
  author={Shi, Chengchun and Song, Rui and Chen, Zhao and Li, Runze},
  journal={Ann. Stat.},
  year={2019},
  volume={47},
  number={5},
  pages={2671},
  publisher={NIH Public Access}
}

@article{ning2017general,
  title={A general theory of hypothesis tests and confidence regions for sparse high dimensional models},
  author={Ning, Yang and Liu, Han},
  journal={Ann. Stat.},
  volume={45},
  number={1},
  pages={158--195},
  year={2017},
  publisher={Institute of Mathematical Statistics}
}

@article{chernozhukov2018double,
  title={Double/debiased machine learning for treatment and structural parameters: Double/debiased machine learning},
  author={Chernozhukov, Victor and Chetverikov, Denis and Demirer, Mert and Duflo, Esther and Hansen, Christian and Newey, Whitney and Robins, James},
  journal={Econom. J.},
  year={2018},
  volume={21},
  pages={C1--C68},
  publisher={Oxford University Press}
}

@article{belloni2016post,
  title={Post-selection inference for generalized linear models with many controls},
  author={Belloni, Alexandre and Chernozhukov, Victor and Wei, Ying},
  journal={J. Bus. Econ. Stat.},
  year={2016}, 
  volume={34},
  number={4},
  pages={606--619},
  publisher={Taylor \& Francis}
}

@article{fei2021estimation,
  title={Estimation and inference for high dimensional generalized linear models: A splitting and smoothing approach.},
  author={Fei, Zhe and Li, Yi},
  journal={J. Mach. Learn. Res.},
  volume={22},
  number={58},
  pages={1--32},
  year={2021}
}

@article{shi2021statistical,
  title={Statistical inference for high-dimensional models via recursive online-score estimation},
  author={Shi, Chengchun and Song, Rui and Lu, Wenbin and Li, Runze},
  journal={J. Am. Stat. Assoc.},
  year={2021},
  volume={116},
  number={535},
  pages={1307--1318},
  publisher={Taylor \& Francis}
}

@article{cai2021statistical,
  title={Statistical inference for high-dimensional generalized linear models with binary outcomes},
  author={Cai, T Tony and Guo, Zijian and Ma, Rong},
  journal={J. Am. Stat. Assoc.},
  volume={118},
  number={542},
  pages={1319--1332},
  year={2023},
  publisher={Taylor \& Francis}
}

@article{guo2021inference,
  title={Inference for the case probability in high-dimensional logistic regression},
  author={Guo, Zijian and Rakshit, Prabrisha and Herman, Daniel S and Chen, Jinbo},
  volume={22},
  number={254},
  pages={1--54},
  journal={J. Mach. Learn. Res.},
  year={2021}
}

@book{wainwright2019high,
  title={High-dimensional statistics: A non-asymptotic viewpoint},
  author={Wainwright, Martin J},
  year={2019},
  publisher={Cambridge University Press}
}

@article{bunea2008honest,
  title={Honest variable selection in linear and logistic regression models via $\ell$1 and $\ell$1+ $\ell$2 penalization},
  author={Bunea, Florentina},
  journal={Electron. J. Stat.},
  volume={2},
  pages={1153--1194},
  year={2008}
}

@article{wang2022finite,
  title={Finite-and large-sample inference for model and coefficients in high-dimensional linear regression with repro samples},
  author={Wang, Peng and Xie, Minge and Zhang, Linjun},
  journal={Ann. Stat.},
  year={2026},
  volumn = {54},
page = {834-860}
}

@article{li2019model,
  title={Model confidence bounds for variable selection},
  author={Li, Yang and Luo, Yuetian and Ferrari, Davide and Hu, Xiaonan and Qin, Yichen},
  journal={Biometrics},
  year={2019},
  volume={75},
  number={2},
  pages={392--403},
  publisher={Wiley Online Library}
}

@article{zheng2019model,
  title={Model selection confidence sets by likelihood ratio testing},
  author={Zheng, Chao and Ferrari, Davide and Yang, Yuhong},
  journal={Stat. Sin.},
  volume={29},
  number={2},
  pages={827--851},
  year={2019}
}

@article{ferrari2015confidence,
  title={Confidence sets for model selection by F-testing},
  author={Ferrari, Davide and Yang, Yuhong},
  journal={Stat. Sin.},
  volume={25},
  pages={1637--1658},
  year={2015}
}

@article{shen2012likelihood,
  title={Likelihood-based selection and sharp parameter estimation},
  author={Shen, Xiaotong and Pan, Wei and Zhu, Yunzhang},
  journal={J. Am. Stat. Assoc.},
  year={2012},
  volume={107},
  number={497},
  pages={223--232},
  publisher={Taylor \& Francis}
}

@article{dezeure2015high,
  title={High-dimensional inference: confidence intervals, p-values and R-software hdi},
  author={Dezeure, Ruben and B{\"u}hlmann, Peter and Meier, Lukas and Meinshausen, Nicolai},
  journal={Stat. Sci.},
  volume = {30},
  number = {4},
  year={2015},
  pages={533--558},
  publisher={JSTOR}
}

@article{van2014asymptotically,
  title={On asymptotically optimal confidence regions and tests for high-dimensional models},
  author={Van de Geer, Sara and B{\"u}hlmann, Peter and Ritov, Ya’acov and Dezeure, Ruben},
  journal={Ann. Stat.},
  year={2014},
  volume={42},
  number={3},
  pages={1166--1202},
  publisher={Institute of Mathematical Statistics}
}

@article{zhao2006model,
  title={On model selection consistency of Lasso},
  author={Zhao, Peng and Yu, Bin},
  journal={J. Mach. Learn. Res.},
  year={2006},
  volume={7},
  pages={2541--2563},
  publisher={JMLR. org}
}

@article{zhang2010nearly,
  title={Nearly unbiased variable selection under minimax concave penalty},
  author={Zhang, Cun-Hui},
  journal={Ann. Stat.},
  year={2010},
  volume={38},
  number={2},
  pages={894--942},
  publisher={Institute of Mathematical Statistics}
}

@incollection{hoeffding1994probability,
  title={Probability inequalities for sums of bounded random variables},
  author={Hoeffding, Wassily},
  booktitle={The collected works of Wassily Hoeffding},
  year={1994},
  publisher={Springer}
}

@book{shalev2014understanding,
  title={Understanding machine learning: From theory to algorithms},
  author={Shalev-Shwartz, Shai and Ben-David, Shai},
  year={2014},
  publisher={Cambridge university press}
}

@article{xie2022repro,
  title = {Repro Samples Method for a Performance Guaranteed Inference in General and Irregular Inference Problems},
  author = {Xie, Minge and Wang, Peng},
  year = {2024},
  journal = {arXiv:2402.15004},
  eprinttype = {arxiv},
  primaryclass = {math, stat},
  volume = {(Invited revision for {\it J.  Am. Stat. Assoc.})},
  archiveprefix = {arXiv}
}

@article{bartlett2006convexity,
  title={Convexity, classification, and risk bounds},
  author={Bartlett, Peter and Jordan, Michael and McAuliffe, Jon},
  journal={J. Am. Stat. Assoc.},
  volume={101},
  number={473},
  pages={138--156},
  year={2006},
  publisher={Taylor \& Francis}
}

@article{zhang2004statistical,
  title={Statistical analysis of some multi-category large margin classification methods},
  author={Zhang, Tong},
  journal={J. Mach. Learn. Res.},
  volume={5},
  number={Oct},
  pages={1225--1251},
  year={2004}
}

@article{wang2019variational,
  title={Variational Bayes under model misspecification},
  author={Wang, Yixin and Blei, David},
  journal={Adv. Neural Inf. Process. Syst.},
  volume={32},
  year={2019}
}

@article{bertsimas2016best,
  title   = {Best subset selection via a modern optimization lens},
  author  = {Bertsimas, Dimitris and King, Angela and Mazumder, Rahul},
  journal = {Ann. Stat.},
  volume  = {44},
  number  = {2},
  pages   = {813--852},
  year    = {2016},
  publisher = {Institute of Mathematical Statistics}
}

\appendix

\section{Comparison to \texorpdfstring{$C_{\min}$ and $\beta_{\min}$}{C_min and beta_min} Conditions}

\subsection{Comparison to \texorpdfstring{$C_{\min}$}{C_min} Condition}

If the sparse GLM is correctly specified, i.e., $\Delta(X)=0$ $P_X$-almost surely, then Assumption \ref{ass_signal} can be shown to be weaker than the $C_{\min}$ condition in \cite{shen2012likelihood}. Note that the $C_{\min}$ condition requires
    \[\inf_{\abs{\tau}\le |\tau_0|,\tau\ne\tau_0,\bm{\beta}_\tau\in\R^{\abs{\tau}}}\frac{[H(\Prob_{\bm{\theta}_0},\Prob_{(\tau,\bm{\beta}_\tau)})]^2}{\abs{\tau_0\setminus\tau}}\gtrsim\frac{\log p}{n},\]
    where $\Prob_{(\tau,\bm{\beta}_\tau)}$ is the joint distribution of $(X,Y)$ with $X\sim\Prob_X$, $\Prob(Y=1|X)=g^{-1}(X_\tau^\top\bm{\beta}_\tau)$, $H(\Prob_1,\Prob_2)$ is the Hellinger distance between $\Prob_1,\Prob_2$. However as we will show in Lemma \ref{lem_cmin} of Section \ref{sec_proof}, when $\sigma>0$, 
    \[L^R_{\bm{\theta}_0}(\tau,\bm{\beta}_\tau,\sigma)={\rm TV}(\Prob_{\bm{\theta}_0},\Prob_{(\tau,\frac{\bm{\beta}_\tau}{\sigma})}),\]
    where ${\rm TV}(\Prob_1,\Prob_2)=\sup_A\abs{\Prob_1(A)-\Prob_2(A)}$ is the total variation distance between $\Prob_1,\Prob_2$. If for any $\tau\subset[p]$ with $\abs{\tau}\le |\tau_0|,\tau\ne\tau_0$, the minimizer $(\bm{\beta}_\tau,\sigma)$ of Equation \eqref{eq_signal} satisfies $\sigma>0$, and if we further assume $s\log\frac{n}{s}\lesssim \log p$, then a sufficient condition for Assumption \ref{ass_signal} is
    \[\inf_{\abs{\tau}\le |\tau_0|,\tau\ne\tau_0,\bm{\beta}_\tau\in\R^{\abs{\tau}}}\frac{{\rm TV}(\Prob_{\bm{\theta}_0},\Prob_{(\tau,\bm{\beta}_\tau)})}{\abs{\tau_0\setminus\tau}}\gtrsim\frac{\log p}{n}.\]
    Since $\{H(\Prob_1,\Prob_2)\}^2\lesssim{\rm TV}(\Prob_1,\Prob_2)\lesssim H(\Prob_1,\Prob_2)$, Assumption \ref{ass_signal} is weaker than the $C_{\min}$ condition in \cite{shen2012likelihood}.

\subsection{Comparison to \texorpdfstring{$\beta_{\min}$}{beta_min} Condition}

  If the sparse GLM is correctly specified, i.e., $\Delta(X)=0$ $P_X$-almost surely, then Assumption \ref{ass_signal} is also weaker than the commonly used $\beta$-min condition \citep{bunea2008honest,zhang2010nearly,zhao2006model}. Denote $\beta_{\min}=\min_{j\in\tau_0}\abs{\beta_{0,j}}$, then the $\beta$-min condition assumes
    \[\beta_{\min}\gtrsim\sqrt{\frac{\log p}{n}}.\]
    As we will show in Lemma \ref{lem_betamin} of Section \ref{sec_proof}, if the samples come from logistic regression model, suppose $\norm{\bm{\beta}_0}_2\lesssim 1$, $X$ is sub-Gaussian and not too concentrated,
    then 
    \[\inf_{\abs{\tau}\le |\tau_0|,\tau\ne\tau_0,\bm{\beta}_\tau\in\R^{\abs{\tau}}}\frac{{\rm TV}(\Prob_{\bm{\theta}_0},\Prob_{(\tau,\bm{\beta}_\tau)})}{\sqrt{\abs{\tau_0\setminus\tau}}}\gtrsim\beta_{\min}.\]
    Therefore, another sufficient condition for Assumption \ref{ass_signal} is
    $\beta_{\min}\gtrsim\frac{\sqrt{s}\log p}{n}+\frac{s\log \frac{n}{s}}{n}.$ 
    When $\frac{s\log p}{n}+\frac{s^2\log^2\frac{n}{s}}{n\log p}\lesssim 1$, we have Assumption \ref{ass_signal} is weaker than the $\beta$-min condition.

\section{Inference for \texorpdfstring{$\beta_{0,j}$}{beta\_{0,j}}, \texorpdfstring{$\bm{\beta}_{0,\tau_0}$}{beta\_0}, and Case Probabilities}

\subsection{Inference for single coefficient \texorpdfstring{$\beta_{0,j}$}{beta\_0,j}} \label{sec_betaj}

Following the general framework described in Section \ref{sec_Abeta} with $A=e_j^\top$, to construct a confidence set for $\beta_{0,j}$, we apply the Wald test to $\beta_{0,j}$ under each candidate model. Concretely, given any candidate model $\tau\in\mathcal{C}$, we test the working hypothesis $H_0:\beta_{0,j}=\beta_j,{\bm{\beta}}_{0,\tau^c}=\bm{0}$ versus $H_1:\beta_{0,j}\ne \beta_j,{\bm{\beta}}_{0,\tau^c}=\bm{0}$.
Without loss of generality, we assume $j\in\tau$, otherwise, if $j\not\in\tau$ and $\beta_j=0$, we accept $H_0$ and if $j\not\in\tau$, $\beta_j\ne 0$, we reject $H_0$. With the quasi MLE $\hat{\bm\beta}_{\tau}$, we calculate the asymptotic variance \eqref{eq_variance} $\hat V=D(\tau)\hat H(\tau)^{-1}\widehat{\Cov}(\nabla l(\tau,\hat{\bm\beta}_{\tau}|X,Y))\hat H(\tau)^{-1}D(\tau)^\top$, 
then the Wald test statistic is $\tilde T(\bm{X}^{obs},\bm{y}^{obs},(\tau,\beta_j))=\frac{n(\hat{\beta}_{j}-\beta_j)^2}{\hat V}$. 
Finally, we combine the Wald test statistics corresponding to each candidate model and define the level-$(1-\alpha)$ confidence set for $\beta_{0,j}$ as 
\begin{align*} 
\Gamma_\alpha^{\beta_{0,j}}(\bm{X}^{obs},\bm{y}^{obs}) = \{\beta_j: \tilde T(\bm{X}^{obs},\bm{y}^{obs},(\tau,\beta_j)) \le F^{-1}_{\chi^2_{\1(j\in\tau)}}(1-\alpha), \tau \in \mathcal{C}\}. 
\end{align*}
Following Theorem \ref{thm_Abeta}, $\Gamma_{\alpha}^{\beta_{0,j}}(\bm{X}^{obs},\bm{y}^{obs})$ is a valid asymptotic level-$(1-\alpha)$ confidence set for $\beta_{0,j}$. 
\begin{Corollary}
If Assumptions \ref{ass_subgaussian}, \ref{ass_hessian} hold and $n\gg s^2$, for any $j\in[p]$, when one of the following conditions holds 
\begin{enumerate}[1)]
    \item $d\rightarrow \infty$ at first, then $n\rightarrow\infty$, and $n, p, s$ satisfy Assumption \ref{ass_signal},  
    \item  fix any $d$, $n\rightarrow \infty$, and $n, p, s$ satisfy condition \eqref{eq_signal_strong}, 
\end{enumerate}
then 
\[\Prob(\beta_{0,j}\in\Gamma_\alpha^{\beta_{0,j}}(\bm{X},\bm{y}))\ge 1-\alpha-o(1). \]
\end{Corollary}

The debiasing methods for high-dimensional logistic regression models \citep{cai2021statistical,van2014asymptotically} have been proposed for inferring single coefficients when the optimal GLM using all the covariates $X$ is sparse. 
These methods require a constant lower bound for the smallest eigenvalue, of either the Hessian matrix with respect to $\bm{\beta}_0$ or the covariance matrix $\E XX^\top$. Such assumptions can be violated if, for instance, two non-informative covariates are identical. 
However, since Assumption~\ref{ass_hessian} only involves $X_{\tau_0}$, our results remain valid in such cases. Moreover, the debiasing methods typically require the sample size to be large enough such that $n\gg s^2\log^2 p$, but we only suppose $n\gg s^2$. More importantly, our method doesn't require a well-specified sparse GLM and remains valid under a misspecified dense model.

The confidence sets generated by debiasing methods are intervals for any $\beta_{0,j}$, regardless of whether $\beta_{0,j}$ is zero. In contrast, the confidence sets produced by our method are unions of intervals. Specifically, if a candidate model contains the index $j$, the confidence set for $\beta_{0,j}$ will encompass the interval derived under that candidate model. If no candidate model includes $j$, then we are confident that $\beta_{0,j}=0$ and the confidence set for $\beta_{0,j}$ reduces to a singleton $\{0\}$. Therefore our method is more flexible and can adapt to the uncertainties of model selection.

\subsection{Inference for \texorpdfstring{$\bm{\beta}_{0}$}{beta\_0}}

Following the general framework in Section \ref{sec_Abeta} with $A=I_p$, to construct a confidence set for $\bm{\beta}_0$, we apply the Wald test to $\bm{\beta}_0$ under each candidate model. Particularly, for each candidate model $\tau\in\mathcal{C}$, we consider the working hypothesis $H_0:{\bm{\beta}}_{0,\tau}={\bm{\beta}}_{\tau},{\bm{\beta}}_{0,\tau^c}=\bm{0}$ versus $H_1:{\bm{\beta}}_{0,\tau}\ne{\bm{\beta}}_{\tau},{\bm{\beta}}_{0,\tau^c}=\bm{0}$.
Based on the quasi MLE $\hat{\bm{\beta}}_{\tau}$, we estimate the asymptotic covariance matrix
\[\hat V(\tau)=\hat H(\tau)^{-1}\widehat{\Cov}(\nabla l(\tau,\hat{\bm\beta}_{\tau}|X,Y))\hat H(\tau)^{-1},\]
then the Wald test statistic is 
\[\tilde T(\bm{X}^{obs},\bm{y}^{obs},(\tau,\bm{\beta}_{\tau}))=n\|\hat V(\tau)^{-\frac{1}{2}}(\hat{\bm\beta}_{\tau}-\bm\beta_{\tau})\|_2^2.\]
Given the Wald test statistics corresponding to each candidate model, the final level-$(1-\alpha)$ confidence set for ${\bm{\beta}}_0$ is
\[\Gamma_\alpha^{{\bm{\beta}}_0}(\bm{X}^{obs},\bm{y}^{obs})=\{{\bm{\beta}}:\tilde T(\bm{X}^{obs},\bm{y}^{obs},(\tau,{\bm{\beta}}_\tau))\le F^{-1}_{\chi^2_{|\tau|}}(1-\alpha),{\bm{\beta}}_{\tau^c}=\bm{0},\tau\in\mathcal{C}\}.\]
Similarly, we have the following corollary stating that $\Gamma_\alpha^{{\bm{\beta}}_0}(\bm{X}^{obs},\bm{y}^{obs})$ has asymptotic coverage $1-\alpha$.
\begin{Corollary}\label{cor_joint}
If Assumptions \ref{ass_subgaussian}, \ref{ass_hessian} hold and $n\gg s^2$, when one of the following conditions holds
\begin{itemize}
    \item[(1)] $d\rightarrow \infty$ at first, then $n\rightarrow\infty$, and $n, p, s$ satisfy Assumption \ref{ass_signal}, 
    \item[(2)]  fix any $d$, $n\rightarrow \infty$, and $n, p, s$ satisfy condition \eqref{eq_signal_strong},
\end{itemize}
then
\[\Prob({\bm{\beta}}_{0}\in\Gamma_\alpha^{{\bm{\beta}}_0}(\bm{X},\bm{y}))\ge 1-\alpha-o(1).\]
\end{Corollary}

When the optimal GLM using all the covariates $X$ is sparse, \cite{zhang2017simultaneous} also studied the simultaneous inference for $\bm{\beta}_0$ based on the debiasing method \citep{van2014asymptotically}. 
Their approach produces an asymptotically valid test for $\bm{\beta}_0$, provided the smallest eigenvalue of the Hessian matrix of the log-likelihood with respect to $\bm{\beta}_0$ exceeds a positive constant. 
However, this assumption fails to hold if there is collinearity among the non-informative covariates. 
In contrast, our method remains valid in such cases. Moreover, instead of being a full-dimensional ellipsoid, our constructed confidence set is a union of low-dimensional ellipsoids with many coefficients to be exactly zero. Therefore, our method can adapt to the uncertainty of model selection. In addition, we only assume $n\gg s^2$ which is weaker than $n\gg s^2\text{poly}\log(np)$ required in \cite{zhang2017simultaneous}. More importantly, our method remains valid even with model misspecification.

\subsection{Simultaneous Inference for Case Probabilities}

GLMs such as logistic regression have been widely applied to detect infectious diseases based on information of patients \citep{ravi2019detection, chadwick2006distinguishing}. 
Statistical inference for patients' case probabilities is critical for identifying those at risk, enabling early intervention.
However, individual-level inference lacks the capacity for group-wise error control and, therefore fails to control disease transmission due to interconnected infection dynamics. 
Consequently, there is an imperative need for simultaneous inference methods for case probabilities of a group of patients.

Given the fixed covariates $\{X_{{\rm new},i}\in\R^p: i\in[n_{\rm new}]\}$ of an arbitrary group of new patients, we use the working GLM $g^{-1}(X_{{\rm new},i}\bm\beta_{0})$
to model the conditional distribution $\Prob(Y_{\text{new},i}=1|X_{\text{new},i})$ of the unknown infection statuses $\{Y_{{\rm new},i}\in\{0,1\}:i\in[n_{\rm new}]\}$.
Then the case probabilities $\{g^{-1}(X_{{\rm new},i}\bm{\beta}_0):i\in[n_{\rm new}]\}$ measure the confidence for labeling each new patient as infected. 
Denote $\bm{X}_{\rm new}=(X_{{\rm new},1},\ldots,X_{{\rm new},n_{\rm new}})^\top\in\R^{n_{\rm new}\times p}$, $g^{-1}(\bm{X}_{\rm new}\bm{\beta}_0)=(g^{-1}(X_{{\rm new},1}\bm{\beta}_0),\ldots,g^{-1}(X_{{\rm new},n_{\rm new}}\bm{\beta}_0))^\top\in\R^{n_{\rm new}}$.
To quantify the uncertainty of predicting $Y_{{\rm new}, i}$'s, we aim to conduct statistical inference for all the case probabilities $g^{-1}(\bm{X}_{\rm new}\bm{\beta}_0)$ of these $n_{\rm new}$ new patients simultaneously. To this end, we construct a confidence set for the vector $g^{-1}(\bm{X}_{\rm new}\bm{\beta}_0)$ and the matrix $A$ in Section~\ref{sec_Abeta} equals $\bm{X}_{\rm new}$.
Then it suffices to form a confidence set for $\bm{X}_{\rm new}\bm{\beta}_0$. 

Following the strategy described in Section~\ref{sec_Abeta} with $A=\bm{X}_{\rm new}$, to construct a confidence set for $\bm{X}_{\rm new}\bm{\beta}_0$, we apply the Wald test to $\bm{X}_{\rm new}\bm{\beta}_0$ under each candidate model. Specifically, for any candidate model $\tau\in\mathcal{C}$, we consider the working hypotheses $H_0:D(\tau)\bm{\beta}_{0,\tau}=t, \bm{\beta}_{0,\tau^c}=\bm{0}$ versus $H_1:D(\tau)\bm{\beta}_{0,\tau}\ne t, \bm{\beta}_{0,\tau^c}=\bm{0}$. Where $D(\tau)$ satisfies $\bm X_{\text{new},\cdot\tau}=C(\tau)D(\tau)$, $D(\tau)D(\tau)^\top=I_{r(\tau)}$ and $\bm{X}_{{\rm new},\cdot\tau}$ is a submatrix consisting of the columns of $\bm{X}_{\rm new}$ with indexes in $\tau$, ${\rm rank}(\bm{X}_{{\rm new},\cdot\tau})=r(\tau)$. Without loss of generality, we assume the existence of $\bm{\beta}$ such that $D(\tau)\bm{\beta}_{\tau}=t$, otherwise we reject $H_0$. Based on the quasi MLE $\hat{\bm\beta}_{\tau}$, we estimate the asymptotic covariance matrix of $D(\tau)\hat{\bm\beta}_{\tau}$ as
\[\hat V(\tau)=D(\tau)\hat H(\tau)^{-1}\widehat{\Cov}(\nabla l(\tau,\hat{\bm\beta}_{\tau}|X,Y))\hat H(\tau)^{-1}D(\tau)^{\top}.\]
Then the Wald test statistic is
\[\tilde T(\bm{X}^{obs},\bm{y}^{obs},(\tau,t))=n\|\hat V(\tau)^{-\frac{1}{2}}(D(\tau)\hat{\bm\beta}_{\tau}-t)\|_2^2.\]
Given the Wald test statistics corresponding to each candidate model, we define the final confidence set for $g^{-1}(\bm{X}_{\rm new}\bm{\beta}_0)$ to be
\[\Gamma_\alpha^{g^{-1}(\bm{X}_{\rm new}\bm{\beta}_0)}(\bm{X}^{obs},\bm{y}^{obs})=\{g^{-1}(\tilde t):\tilde t=C(\tau)t,\tilde T(\bm{X}^{obs},\bm{y}^{obs},(\tau,t))\le F^{-1}_{\chi^2_{r(\tau)}}(1-\alpha),\tau\in\mathcal{C}\}.\]
According to Theorem~\ref{thm_Abeta}, we know $\Gamma^{g^{-1}(\bm{X}_{\rm new}\bm{\beta}_0)}_\alpha(\bm{X}^{obs},\bm{y}^{obs})$ is asymptotically valid.

\begin{Corollary}\label{cor_case_probability}
    If Assumptions \ref{ass_subgaussian}, \ref{ass_hessian} hold and $n\gg s^2$, when one of the following conditions holds
\begin{itemize}
    \item[(1)] $d\rightarrow \infty$ at first, then $n\rightarrow\infty$, and $n, p, s$ satisfy Assumption \ref{ass_signal}, 
    \item[(2)]  fix any $d$, $n\rightarrow \infty$, and $n, p, s$ satisfy condition \eqref{eq_signal_strong},
\end{itemize}
then the confidence set $\Gamma_\alpha^{g^{-1}(\bm{X}_{\rm new}\bm{\beta}_0)}(\bm{X}^{obs},\bm{y}^{obs})$ is asymptotically valid
\[\Prob(g^{-1}(\bm{X}_{\rm new}\bm{\beta}_0)\in\Gamma_\alpha^{g^{-1}(\bm{X}_{\rm new}\bm{\beta}_0)}(\bm{X},\bm{y}))\ge 1-\alpha-o(1).\]
\end{Corollary}

In comparison, \cite{guo2021inference} pioneered the study of statistical inference for case probabilities in high-dimensional logistic regression models. 
However, their method can only be applied to one observation and requires a well-specified model, in contrast, our method enables simultaneous inference for the case probabilities of an arbitrary set of new observations even with model misspecification.

\section{Proofs}\label{sec_proof}
This section includes all the proofs of the theoretical results in the previous sections.

\subsection{Prediction performance of sparsity-constrained GLM}
In addition to the interpretability, the following lemma shows that the defined sparsity-constrained GLM also has reasonable prediction performance.

\begin{Lemma}\label{lem_prediction}
    Denote $\phi(x)=\log(1+e^{-x})$ to be the logistic loss function. If the link function $g^{-1}(x)=\frac{e^x}{1+e^x}$ is the logit link, then the sparsity-constrained GLM defined in \eqref{eq_tau0} and \eqref{eq_beta0} has prediction error controlled as follows,
    \begin{align*}
        &\Prob\bigg(Y\ne\1\bigg\{g^{-1}(X_{\tau_0}^\top \bm\beta_{0,\tau_0})>\frac{1}{2}\bigg\}\bigg)-\inf_{f:\R^{|\tau_0|}\rightarrow \{0,1\}}\Prob\big(Y\ne f(X_{\tau_0})\big)\\
        \le&\sqrt{2}\bigg\{\inf_{\bm\beta\in\R^p}\E\phi\big((2Y-1)X_{\tau_0}^\top\bm\beta_{\tau_0}\big)-\inf_{f:\R^{|\tau_0|}\rightarrow\R}\E\phi\big((2Y-1)f(X_{\tau_0})\big)\bigg\}^{\frac{1}{2}}.
    \end{align*}
\end{Lemma}

\begin{proof}[Proof of Lemma \ref{lem_prediction}]
    Denote $\psi$ function as
    \[\psi(x)=\log 2+\frac{1+x}{2}\log\frac{1+x}{2}+\frac{1-x}{2}\log\frac{1-x}{2},\]
    it follows from Pinsker's inequality that 
    \[\psi(x)\ge \frac{x^2}{2}.\]
    Since the logistic loss $\phi(x)=\log(1+e^{-x})$ is convex, it follows from \cite{bartlett2006convexity} that
    \begin{align*}
        &\Prob\bigg(Y\ne\1\bigg\{g^{-1}(X_{\tau_0}^\top \bm\beta_{0,\tau_0})>\frac{1}{2}\bigg\}\bigg)-\inf_{f:\R^{|\tau_0|}\rightarrow \{0,1\}}\Prob\big(Y\ne f(X_{\tau_0})\big)\\
        \le&\sqrt{2}\bigg\{\inf_{\bm\beta\in\R^p}\E\phi\big((2Y-1)X_{\tau_0}^\top\bm\beta_{\tau_0}\big)-\inf_{f:\R^{|\tau_0|}\rightarrow\R}\E\phi\big((2Y-1)f(X_{\tau_0})\big)\bigg\}^{\frac{1}{2}}.
    \end{align*}
\end{proof}

\subsection{$\tau_0$ when $\mu(X)$ is close to $s$-sparse GLM}

The following lemma states that if $\mu(X)$ is close to an $s$-sparse GLM with support $\tilde\tau$, then $\tilde\tau$ will have a small error for recovering $Y$. Meanwhile, if all the other sparse models have a relatively large reconstruction error, then $\tau_0$ defined in \eqref{eq_tau0} will be $\tilde\tau$.
\begin{Lemma}\label{lem_tau0}
\begin{enumerate}[1)]
    \item Suppose $\mu(X)=g^{-1}(X_{\tilde\tau}^\top\tilde{\bm\beta}_{\tilde\tau})$ with $|\tilde\tau|= s$, then $\tau_0=\tilde\tau$.
    \item Suppose $\mu(X)$ is close to an $s$-sparse GLM $g^{-1}(X_{\tilde\tau}^\top\tilde{\bm\beta}_{\tilde\tau})$ with $|\tilde\tau|= s$, denote 
    \[\Delta(X)\overset{\triangle}{=}\mu(X)-g^{-1}(X_{\tilde\tau}^\top\tilde{\bm\beta}_{\tilde\tau}),\quad\delta\overset{\triangle}{=}\Prob\bigg(\mu(X)\in\bigg(\frac{1}{2},\frac{1}{2}+\Delta(X)\bigg]\cup\bigg(\frac{1}{2}+\Delta(X),\frac{1}{2}\bigg]\bigg),\]
    then
    \[\Prob\bigg(Y\ne \1\bigg(g^{-1}(X_{\tilde\tau}^\top\tilde{\bm\beta}_{\tilde\tau})>\frac{1}{2}\bigg)\bigg)-\Prob\bigg(Y\ne\1\bigg(\mu(X)>\frac{1}{2}\bigg)\bigg)\le\delta.\]
    If all models $\tau\ne\tilde\tau$ have a relatively large data reconstruction error such that
    \[\min_{\tau\ne\tilde\tau,|\tau|\le s}\inf_{\bm\beta_{\tau}\in\R^{|\tau|}}\Prob\bigg(Y\ne\1\bigg(g^{-1}(X_{\tau}^\top\bm\beta_\tau)>\frac{1}{2}\bigg)\bigg)>\Prob\bigg(Y\ne\1\bigg(\mu(X)>\frac{1}{2}\bigg)\bigg)+\delta,\]
    then $\tau_0=\tilde\tau$.
\end{enumerate}
    
\end{Lemma}

\begin{proof}[Proof of Lemma \ref{lem_tau0}]

    \begin{enumerate}[1)]
        \item The first part follows from the Fisher consistency of 0-1 loss.
        \item It is easy to verify that for any $f:\R^p\rightarrow \{0,1\}$, 
    \[\Prob(Y\ne f(X))=\E\mu(X)+\E(1-2\mu(X))f(X),\]
    \[\Prob(Y\ne f(X))-\Prob\bigg(Y\ne \1\bigg(\mu(X)>\frac{1}{2}\bigg)\bigg)=\E|2\mu(X)-1|\bigg|f(X)-\1\bigg(\mu(X)>\frac{1}{2}\bigg)\bigg|.\]
    Then, 
    \begin{align*}
        &\Prob\bigg(Y\ne\1\bigg(g^{-1}(X_{\tilde\tau}^\top\tilde{\bm\beta}_{\tilde\tau})>\frac{1}{2}\bigg)\bigg)-\Prob\bigg(Y\ne\1\bigg(\mu(X)>\frac{1}{2}\bigg)\bigg)\\
        \le&\E\bigg|\1\bigg(g^{-1}(X_{\tilde\tau}^\top\tilde{\bm\beta}_{\tilde\tau})>\frac{1}{2}\bigg)-\1\bigg(\mu(X)>\frac{1}{2}\bigg)\bigg|=\delta,
    \end{align*}
    and for any $\tau\ne\tilde\tau,|\tau|\le s,\bm\beta_{\tau}\in\R^{|\tau|}$, we have
    \begin{align*}
        &\Prob\bigg(Y\ne\1\bigg(g^{-1}(X_{\tau}^\top\bm\beta_\tau)>\frac{1}{2}\bigg)\bigg)-\Prob\bigg(Y\ne\1\bigg(g^{-1}(X_{\tilde\tau}^\top\tilde{\bm\beta}_{\tilde\tau})>\frac{1}{2}\bigg)\bigg)\\
        =&\Prob\bigg(Y\ne\1\bigg(g^{-1}(X_{\tau}^\top\bm\beta_\tau)>\frac{1}{2}\bigg)\bigg)-\Prob\bigg(Y\ne\1\bigg(\mu(X)>\frac{1}{2}\bigg)\bigg)\\
        &+\Prob\bigg(Y\ne\1\bigg(\mu(X)>\frac{1}{2}\bigg)\bigg)-\Prob\bigg(Y\ne\1\bigg(g^{-1}(X_{\tilde\tau}^\top\tilde{\bm\beta}_{\tilde\tau})>\frac{1}{2}\bigg)\bigg)\\
        >&\delta-\E|2\mu(X)-1|\1\bigg(\mu(X)\in\bigg(\frac{1}{2},\frac{1}{2}+\Delta(X)\bigg]\cup\bigg(\frac{1}{2}+\Delta(X),\frac{1}{2}\bigg]\bigg)\\
        >&0,
    \end{align*}
    which implies $\tau_0=\tilde\tau$.
    \end{enumerate}
    
\end{proof}

\subsection{Connection to $\beta_{\rm min}$}
\begin{Lemma}\label{lem_cmin}
For any $\tau_1,\tau_2\subset[p], {\bm{\beta}}_1\in\R^{\abs{\tau_1}},{\bm{\beta}}_2\in\R^{\abs{\tau_2}}$, we have
\[\Prob(\mathbbm{1}\{X_{\tau_1}^\top{\bm{\beta}}_1+\epsilon>0\}\ne\mathbbm{1}\{X_{\tau_2}^\top{\bm{\beta}}_2+\epsilon>0\})= {\rm TV}(\Prob_{X,Y|\tau_1,{\bm{\beta}}_{1}},\Prob_{X,Y|\tau_2,{\bm{\beta}}_{2}}).\]

\end{Lemma}
\begin{proof}[Proof of Lemma \ref{lem_cmin}]

\begin{align*}
    &\Prob(\mathbbm{1}\{X_{\tau_1}^\top{\bm{\beta}}_1+\epsilon>0\}\ne\mathbbm{1}\{X_{\tau_2}^\top{\bm{\beta}}_2+\epsilon>0\})\\
    =&\E\Prob(X_{\tau_1}^\top{\bm{\beta}}_{1}\le g(U)<X_{\tau_2}^\top{\bm{\beta}}_2|X)+\E\Prob(X_{\tau_2}^\top{\bm{\beta}}_{2}\le g(U)<X_{\tau_1}^\top{\bm{\beta}}_1|X)\\
    =&\E\left|g^{-1}(X_{\tau_1}^\top{\bm{\beta}}_1)-g^{-1}(X_{\tau_2}^\top{\bm{\beta}}_2)\right|\\
    =&\E\abs{\Prob_{Y|X,(\tau_1,{\bm{\beta}}_1)}(Y=1|X)-\Prob_{Y|X,(\tau_2,{\bm{\beta}}_2)}(Y=1|X))}\\
    =&{\rm TV}(\Prob_{(\tau_1,{\bm{\beta}}_1)},\Prob_{(\tau_2,{\bm{\beta}}_2)}).
\end{align*}

\end{proof}

\begin{Lemma}\label{lem_betamin}
    Denote $\beta_{\min}=\min_{j\in\tau_0}\abs{\beta_{0,j}}$. Assume $\norm{X}_{\psi_2}\le\xi$, $\norm{{\bm{\beta}}_0}_2\le B$ and the density of $X^\top\bm{\beta}$ is upper bounded by $C$ for any $\bm{\beta}$ satisfying $\|\bm{\beta}\|_0\le 2|\tau_0|, \|\bm{\beta}\|_2\ge 1$. Here, $\xi,B$ and $C$ are positive constants,
    then
    \[\inf_{\abs{\tau}\le |\tau_0|,\tau\ne\tau_0,{\bm{\beta}}_\tau\in\R^{\abs{\tau}}}\frac{{\rm TV}(\Prob_{\bm{\theta}_0},\Prob_{(\tau,{\bm{\beta}}_\tau)})}{\sqrt{\abs{\tau_0\setminus\tau}}}\gtrsim\beta_{\min}.\]
\end{Lemma}

\begin{proof}[Proof of Lemma \ref{lem_betamin}]
    
    By Lemma \ref{lem_cmin},
    \begin{align*}
        &{\rm TV}(\Prob_{\bm{\theta}_0},\Prob_{(\tau,{\bm{\beta}}_\tau)})\\
        =&\E\abs{\frac{1}{1+e^{-X_{\tau_0}^\top{\bm{\beta}}_{0,\tau_0}}}-\frac{1}{1+e^{-X_\tau^\top{\bm{\beta}}_\tau}}}\\
    =&\E\left|\dfrac{1}{1+e^{-X_{\tau_0}^\top{\bm{\beta}}_{0,\tau_0}}}-\dfrac{1}{1+e^{-X_\tau^\top{\bm{\beta}}_\tau}}\right|\{\mathbbm{1}\{|X_{\tau}^\top{\bm{\beta}}_\tau|\le2|X_{\tau_0}^\top{\bm{\beta}}_{0,\tau_0}|\}+\mathbbm{1}\{|X_{\tau}^\top{\bm{\beta}}_\tau|>2|X_{\tau_0}^\top{\bm{\beta}}_{0,\tau_0}|\}\}\\
    \ge &\E|X_{\tau_0}^\top{\bm{\beta}}_{0,\tau_0}-X_\tau^\top{\bm{\beta}}_\tau|\dfrac{e^{-2|X_{\tau_0}^\top{\bm{\beta}}_{0,\tau_0}|}}{(1+e^{-2|X_{\tau_0}^\top{\bm{\beta}}_{0,\tau_0}|})^2}\mathbbm{1}\{|X_{\tau}^\top{\bm{\beta}}_\tau|\le2|X_{\tau_0}^\top{\bm{\beta}}_{0,\tau_0}|\}\\
    &+\E|X_{\tau_0}^\top{\bm{\beta}}_{0,\tau_0}|\dfrac{e^{-2|X_{\tau_0}^\top{\bm{\beta}}_{0,\tau_0}|}}{(1+e^{-2|X_{\tau_0}^\top{\bm{\beta}}_{0,\tau_0}|})^2}\mathbbm{1}\{|X_{\tau}^\top{\bm{\beta}}_\tau|>2|X_{\tau_0}^\top{\bm{\beta}}_{0,\tau_0}|\}\\
    \ge&\dfrac{(\E\min\{|X_{\tau_0}^\top{\bm{\beta}}_{0,\tau_0}-X_\tau^\top{\bm{\beta}}_\tau|,|X_{\tau_0}^\top{\bm{\beta}}_{0,\tau_0}|\}^{1/2})^2}{\E(1+e^{-2|X_{\tau_0}^\top{\bm{\beta}}_{0,\tau_0}|})^2e^{2|X_{\tau_0}^\top{\bm{\beta}}_{0,\tau_0}|}}\\
    \ge&\beta_{\min}\inf_{\abs{\tau}\le |\tau_0|,\tau\ne\tau_0,{\bm{\beta}}_\tau\in\R^{\abs{\tau}}}\left(\E\min\left\{\abs{X_{\tau_0}^\top\frac{{\bm{\beta}}_{0,\tau_0}}{\beta_{\min}}-X_\tau^\top{\bm{\beta}}_\tau},\abs{X_{\tau_0}^\top\frac{{\bm{\beta}}_{0,\tau_0}}{\beta_{\min}}}\right\}^{1/2}\right)^2e^{-c\|{\bm{\beta}}_0\|_2^2\xi^2}/4.
    \end{align*}
    For $\tau\ne\tau_0,\abs{\tau}\le |\tau_0|$, there exists $b\in\R^p,\norm{b}_0\le 2|\tau_0|$ such that $X_{\tau_0}^\top\frac{{\bm{\beta}}_{0,\tau_0}}{\beta_{\min}}-X_\tau^\top{\bm{\beta}}_\tau=X^\top b$. For any $j\in\tau_0\setminus\tau$, we have $\abs{b_j}=\frac{\abs{\beta_{0,j}}}{\beta_{\min}}\ge 1$, therefore $\norm{b}_2\ge \sqrt{\abs{\tau_0\setminus\tau}}$. Similarly $\norm{\frac{{\bm{\beta}}_{0,\tau_0}}{\beta_{\min}}}_2\ge \sqrt{|\tau_0|}$. 
    \begin{align*}
        \sup_{\norm{{\bm{\beta}}}_0\le 2|\tau_0|,\norm{{\bm{\beta}}}_2\ge1}\Prob(\abs{X^\top{\bm{\beta}}}\le\frac{1}{8C})\le\frac{1}{4}.
    \end{align*}
    Then for any $\abs{\tau}\le |\tau_0|,\tau\ne\tau_0$,
    \begin{align*}
        &\inf_{{\bm{\beta}}_\tau\in\R^{\abs{\tau}}}\E\min\left\{\abs{X_{\tau_0}^\top\frac{{\bm{\beta}}_{0,\tau_0}}{\beta_{\min}}-X_\tau^\top{\bm{\beta}}_\tau},\abs{X_{\tau_0}^\top\frac{{\bm{\beta}}_{0,\tau_0}}{\beta_{\min}}}\right\}^{1/2}\\
        \ge&\inf_{{\bm{\beta}}_\tau\in\R^{\abs{\tau}}}\E\min\left\{\abs{X_{\tau_0}^\top\frac{{\bm{\beta}}_{0,\tau_0}}{\beta_{\min}}-X_\tau^\top{\bm{\beta}}_\tau},\abs{X_{\tau_0}^\top\frac{{\bm{\beta}}_{0,\tau_0}}{\beta_{\min}}}\right\}^{1/2}\\
        &\qquad\cdot\1\left\{\abs{X_{\tau_0}^\top\frac{{\bm{\beta}}_{0,\tau_0}}{\beta_{\min}}-X_\tau^\top{\bm{\beta}}_\tau}>\sqrt{\abs{\tau_0\setminus\tau}}\frac{1}{8C},\abs{X_{\tau_0}^\top\frac{{\bm{\beta}}_{0,\tau_0}}{\beta_{\min}}}>\sqrt{|\tau_0|}\frac{1}{8C}\right\}\\
        \ge&\frac{1}{2\sqrt{2C}}\abs{\tau_0\setminus\tau}^{1/4}\inf_{{\bm{\beta}}_\tau\in\R^{\abs{\tau}}}\bigg(1-\Prob\bigg(\abs{X_{\tau_0}^\top\frac{{\bm{\beta}}_{0,\tau_0}}{\beta_{\min}}-X_\tau^\top{\bm{\beta}}_\tau}\le \sqrt{\abs{\tau_0\setminus\tau}}\frac{1}{8C}\bigg)\\
        &-\Prob\bigg(\abs{X_{\tau_0}^\top\frac{{\bm{\beta}}_{0,\tau_0}}{\beta_{\min}}}\le \sqrt{|\tau_0|}\frac{1}{8C}\bigg)\bigg)\\
        \ge&\frac{1}{2\sqrt{2C}}\abs{\tau_0\setminus\tau}^{1/4}(1-2\sup_{\norm{{\bm{\beta}}}_0\le 2|\tau_0|,\norm{{\bm{\beta}}}_2\ge1}\Prob(|X^\top{\bm{\beta}}|\le \frac{1}{8C}))\\
        \ge& \frac{1}{4\sqrt{2C}}\abs{\tau_0\setminus\tau}^{1/4}.
    \end{align*}
    Combining terms completes the proof.
\end{proof}

\subsection{Proofs in Section \ref{sec_candidate}}

The following lemma follows from the Fundamental Theorem of Learning Theory \citep{shalev2014understanding}
\begin{Lemma}\label{lem_pac_upper}
For any $\tau\subset[p]$, we have
\begin{align*}
    &\Prob(\exists{\bm{\beta}}_\tau\in\R^{\abs{\tau}},\sigma\ge 0 ~{\rm s.t.}~ L^R_n(\tau,{\bm{\beta}}_\tau,\sigma|\bm{X},\bm{y},\bm{\epsilon})=0, L^R_{\bm{\theta}_0}(\tau,{\bm{\beta}}_\tau,\sigma)\ge\eta)\\
    \le&(1-e^{-\frac{n\eta}{8}})^{-1}\bigg\{2^{\abs{\tau}+1}\vee\bigg(\dfrac{2en}{\abs{\tau}+1}\bigg)^{\abs{\tau}+1}\bigg\}2^{-\frac{n\eta}{2}}.
\end{align*}
\end{Lemma}

\begin{proof}[Proof of Lemma \ref{lem_pac_upper}]
Suppose we have another sample $\tilde S=\{(\tilde X_i,\tilde \epsilon_i,\tilde Y_i):i\in[n]\}$ that is i.i.d. with $S=\{(X_i^{obs},\epsilon_i^{rel},y_i^{obs}):i\in[n]\}$. Denote
\[A=\{\exists{\bm{\beta}}_\tau\in\R^{\abs{\tau}},\sigma\ge 0 ~{\rm s.t.}~ L^R_n(\tau,{\bm{\beta}}_\tau,\sigma|\bm{X},\bm{y},\bm{\epsilon})=0, L^R_{\bm{\theta}_0}(\tau,{\bm{\beta}}_\tau,\sigma)\ge\eta\},\]
\[B=\{\exists{\bm{\beta}}_\tau\in\R^{\abs{\tau}},\sigma\ge 0 ~{\rm s.t.}~ L^R_n(\tau,{\bm{\beta}}_\tau,\sigma|\bm{X},\bm{y},\bm{\epsilon})=0, L^R_n(\tau,{\bm{\beta}}_\tau,\sigma|\tilde{\bm{X}},\tilde{\bm{y}},\tilde{\bm{\epsilon}})\ge\dfrac{\eta}{2}\}.\]
Conditioning on event $A$, we denote $\hat{\bm{\beta}}_\tau\in R^{\abs{\tau}},\hat\sigma\ge 0$ to be the coefficients satisfy $A$. Given $S$ and $A$, $\mathbbm{1}\{\tilde Y\ne\mathbbm{1}\{\tilde X_\tau^\top\hat{\bm{\beta}}_{\tau}+\hat\sigma\tilde \epsilon>0\}\}$ is a Bernoulli random variable with parameter $\rho=L^R_{\bm{\theta}_0}(\tau,\hat{\bm{\beta}}_\tau,\hat\sigma)\ge\eta$, using Chernoff bound in multiplicative form \citep{hoeffding1994probability}, we have
\[\Prob(B^c|A)\le\Prob(L^R_n(\tau,\hat{\bm{\beta}}_\tau,\hat\sigma|\tilde{\bm{X}},\tilde{\bm{Y}},\tilde{\bm{\epsilon}})\le\frac{1}{2}L^R_{\bm{\theta}_0}(\tau,\hat{\bm{\beta}}_\tau,\sigma)|A)\le\E e^{-\frac{n\rho}{8}}\le e^{-\frac{n\eta}{8}}.\]
Then
\[\Prob(B)\ge\Prob(B|A)\Prob(A)\ge (1-e^{-\frac{n\eta}{8}})\Prob(A).\]
Now conditioning on $S\cup\tilde S$, we construct $T$ and $\tilde T$ by randomly partitioning $S\cup\tilde S$ into two sets with equal sizes. We also denote
\[L^R_n(\tau,{\bm{\beta}}_\tau,\sigma|T)=\dfrac{1}{n}\sum_{(X,\epsilon,Y)\in T}\mathbbm{1}\{Y\ne\mathbbm{1}\{X_{\tau}^\top{\bm{\beta}}_{\tau}+\sigma\epsilon>0\}\},\]
\[L^R_n(\tau,{\bm{\beta}}_\tau,\sigma|\tilde T)=\dfrac{1}{n}\sum_{(X,\epsilon,Y)\in \tilde T}\mathbbm{1}\{Y\ne\mathbbm{1}\{X_{\tau}^\top{\bm{\beta}}_{\tau}+\sigma\epsilon>0\}\},\]
then 
\begin{align*}
    \Prob(B)=&\E_{S\cup\tilde S}\Prob(\exists{\bm{\beta}}_\tau\in\R^{\abs{\tau}},\sigma\ge 0 ~{\rm s.t.}~ L^R_n(\tau,{\bm{\beta}}_\tau,\sigma|\bm{X},\bm{y},\bm{\epsilon})=0, L^R_n(\tau,{\bm{\beta}}_\tau,\sigma|\tilde{\bm{X}},\tilde{\bm{Y}},\tilde{\bm{\epsilon}})\ge\dfrac{\eta}{2}|S\cup\tilde S)\\
    =&\E_{S\cup\tilde S}\Prob(\exists{\bm{\beta}}_\tau\in\R^{\abs{\tau}},\sigma\ge 0 ~{\rm s.t.}~ L^R_n(\tau,{\bm{\beta}}_\tau,\sigma|T)=0, L^R_n(\tau,{\bm{\beta}}_\tau,\sigma|\tilde T)\ge\dfrac{\eta}{2}|S\cup\tilde S).
\end{align*}
Conditioning on $S\cup\tilde S$, instead of considering ${\bm{\beta}}_\tau$ directly, we study the evaluation of the classifiers $\mathbbm{1}(X_\tau^\top{\bm{\beta}}_\tau+\sigma\epsilon>0)$ on samples in $S\cup\tilde S$, then by Sauer's Lemma \citep{shalev2014understanding}, the total number of labellings of $\mathbbm{1}\{X_\tau^\top{\bm{\beta}}_\tau+\sigma\epsilon>0\},\forall{\bm{\beta}}_\tau\in\R^{\abs{\tau}},\sigma\ge 0$ on $S\cup\tilde S$ is less than $2^{\abs{\tau}+1}\vee\big(\frac{2en}{\abs{\tau}+1}\big)^{\abs{\tau}+1}$
\begin{align*}
    &\Prob(\exists{\bm{\beta}}_\tau\in\R^{\abs{\tau}},\sigma\ge 0 ~{\rm s.t.}~ L^R_n(\tau,{\bm{\beta}}_\tau,\sigma|T)=0, L^R_n(\tau,{\bm{\beta}}_\tau,\sigma|\tilde T)\ge\dfrac{\eta}{2}|S\cup\tilde S)\\
    \le&\bigg\{2^{\abs{\tau}+1}\vee\bigg(\dfrac{2en}{\abs{\tau}+1}\bigg)^{\abs{\tau}+1}\bigg\}\sup_{{\bm{\beta}}_\tau\in\R^{\abs{\tau}},\sigma\ge 0}\Prob(L^R_n(\tau,{\bm{\beta}}_\tau,\sigma|T)=0, L^R_n(\tau,{\bm{\beta}}_\tau,\sigma|\tilde T)\ge\dfrac{\eta}{2}|S\cup\tilde S)\\
    \le&\bigg\{2^{\abs{\tau}+1}\vee\bigg(\dfrac{2en}{\abs{\tau}+1}\bigg)^{\abs{\tau}+1}\bigg\}2^{-\frac{n\eta}{2}},
\end{align*}
where to derive the last inequality, we assume the total number of errors of ${\bm{\beta}}_\tau$ on $S\cup\tilde S$ to be $m\in[\frac{n\eta}{2},n]$, then the probability that all the $m$ wrong samples are in $\tilde T$ is $\binom{n}{m}/\binom{2n}{m}\le 2^{-m}\le 2^{-\frac{n\eta}{2}}$. 

In conclusion, we have 
\[\Prob(A)\le(1-e^{-\frac{n\eta}{8}})^{-1}\bigg\{2^{\abs{\tau}+1}\vee\bigg(\dfrac{2en}{\abs{\tau}+1}\bigg)^{\abs{\tau}+1}\bigg\}2^{-\frac{n\eta}{2}}.\]
\end{proof}

\begin{proof}[Proof of Lemma \ref{lem_oracle}]
Since $\tau_0$ is one of the minimizers of problem \eqref{eq_oracle}, we know the minimum is 0. Denote
\[\tilde c_{\min}=\min_{|\tau|\le |\tau_0|,\tau\not\supset\tau_0,{\bm{\beta}}_\tau\in\R^{\abs{\tau}},\sigma\ge 0}\frac{L^R_{\bm{\theta}_0}(\tau,{\bm{\beta}}_\tau,\sigma)-\frac{2\abs{\tau}+2}{n}\log_2\frac{2en}{\abs{\tau}+1}}{\abs{\tau_0\setminus\tau}},\]
\[c_{\min}=\min_{|\tau|\le |\tau_0|,\tau\not\supset\tau_0,{\bm{\beta}}_\tau\in\R^{\abs{\tau}},\sigma\ge 0}\frac{L^R_{\bm{\theta}_0}(\tau,{\bm{\beta}}_\tau,\sigma)-\frac{2\abs{\tau}+2}{n}\log_2\frac{2en}{\abs{\tau}+1}}{\abs{\tau}\vee 1},\]
then
\begin{align*}
    &\Prob(\inf_{\tau\not\supset\tau_0,|\tau|\le |\tau_0|,{\bm{\beta}}\in\R^p,\sigma\ge 0} L^R_n(\tau,{\bm{\beta}}_\tau,\sigma|\bm{X},\bm{y},\bm{\epsilon})=0)\\
    = &\Prob(\exists \tau\not\supset\tau_0, |\tau|\le |\tau_0|, {\bm{\beta}}_\tau\in\R^{\abs{\tau}},\sigma\ge 0 ~{\rm s.t.}~ L^R_n(\tau,{\bm{\beta}}_\tau,\sigma|\bm{X},\bm{y},\bm{\epsilon})=0,\\
    &\qquad L^R_{\bm{\theta}_0}(\tau,{\bm{\beta}}_\tau,\sigma)\ge\inf_{{\bm{\beta}}_\tau\in\R^{\abs{\tau}},\sigma\ge 0}L^R_{\bm{\theta}_0}(\tau,{\bm{\beta}}_\tau,\sigma))\\
    \le&\sum_{\tau\not\supset\tau_0,|\tau|\le |\tau_0|}\Prob(\exists{\bm{\beta}}_\tau\in\R^{\abs{\tau}},\sigma\ge 0 ~{\rm s.t.}~ L^R_n(\tau,{\bm{\beta}}_\tau,\sigma|\bm{X},\bm{y},\bm{\epsilon})=0,\\
    &\qquad L^R_{\bm{\theta}_0}(\tau,{\bm{\beta}}_\tau,\sigma)\ge\inf_{{\bm{\beta}}_\tau\in\R^{\abs{\tau}},\sigma\ge 0}L^R_{\bm{\theta}_0}(\tau,{\bm{\beta}}_\tau,\sigma))\\
    \overset{\triangle}{=}&T.
\end{align*}
On the one hand, noting that $\sum_{l=0}^r\binom{p-|\tau_0|}{l}\le (\frac{e(p-|\tau_0|)}{r})^r,\binom{|\tau_0|}{r}\le |\tau_0|^r,|\tau_0|(p-|\tau_0|)\le\frac{p^2}{4}$, if we divide $\abs{\tau}$ into $j=\abs{\tau_0\cap\tau}$ and $l=\abs{\tau\setminus\tau_0}$, then applying Lemma \ref{lem_pac_upper} gives
\begin{align*}
    T\lesssim &\sum_{\tau\not\supset\tau_0,\abs{\tau}\le |\tau_0|}2^{-\frac{1}{2}n\inf_{{\bm{\beta}}_\tau\in\R^{\abs{\tau}},\sigma\ge 0}L^R_{\bm{\theta}_0}(\tau,{\bm{\beta}}_\tau,\sigma)+(\abs{\tau}+1)\log_2 \frac{2en}{\abs{\tau}+1}}\\
    \le&\sum_{j=0}^{|\tau_0|-1}\sum_{l=0}^{|\tau_0|-j}\binom{|\tau_0|}{j}\binom{p-|\tau_0|}{l}2^{-\frac{1}{2}n(|\tau_0|-j)\tilde c_{\min}}\\
    \overset{r=|\tau_0|-j}{\le} & \sum_{r=1}^{|\tau_0|} |\tau_0|^r2^{-\frac{1}{2}nr\tilde c_{\min}}\sum_{l=0}^{r}\binom{p-|\tau_0|}{l}\\
    \le&\sum_{r=1}^{|\tau_0|}2^{-r(\frac{1}{2}n\tilde c_{\min}-\log_2 (e|\tau_0|(p-|\tau_0|)))}\\
    \le & \dfrac{2^{-\frac{1}{2}n\tilde c_{\min}+\log_2 (e|\tau_0|(p-|\tau_0|))}}{1-2^{-\frac{1}{2}n\tilde c_{\min}+\log_2 (e|\tau_0|(p-|\tau_0|))}}\\
    \le & 2^{-\frac{1}{2}n\tilde c_{\min}+\log_2(e|\tau_0|(p-|\tau_0|))+1}\\
    \lesssim&2^{-\frac{1}{2}n\tilde c_{\min}+2\log_2 p}.
\end{align*}

On the other hand, similarly we denote $j=\abs{\tau}$, then
\begin{align*}
    T\lesssim &\sum_{\tau\not\supset\tau_0,\abs{\tau}\le |\tau_0|}2^{-\frac{1}{2}n\inf_{{\bm{\beta}}_\tau\in\R^{\abs{\tau}},\sigma\ge 0}L^R_{\bm{\theta}_0}(\tau,{\bm{\beta}}_\tau,\sigma)+(\abs{\tau}+1)\log_2 \frac{2en}{\abs{\tau}+1}}\\
    \le&\sum_{j=0}^{|\tau_0|}\binom{p}{j}2^{-\frac{1}{2}n(j\vee 1)c_{\min}}\\
    \le&\sum_{j=0}^{|\tau_0|}2^{-\frac{1}{2}n(j\vee 1)c_{\min}+j\log_2 p}\\
    \lesssim &2^{-\frac{1}{2}nc_{\min}+\log_2 p}.
\end{align*}
\end{proof}

\begin{proof}[Proof of Theorem \ref{thm_candidate}]

If we denote
\[A=\{\bm{\epsilon}^*:-X^\top_{i,\tau_0}\bm{\beta}_{0,\tau_0}<\epsilon^*_i\le\epsilon_i~{\rm or}~\epsilon_i\le\epsilon_i^*\le-X^\top_{i,\tau_0}\bm{\beta}_{0,\tau_0},\forall i\in[n]\},\]
then we have the following decomposition
\begin{align*}
    \Prob(\tau_0\not\in\mathcal{C})
    \le\Prob(\{\tau_0\not\in\mathcal{C})\}\cap(\cup_{j\in[d]}\{\bm{\epsilon}^{*(j)}\in A\}))+\Prob(\cap_{j\in[d]}\{\bm{\epsilon}^{*(j)}\not\in A\})
    =T_1+T_2.
\end{align*}
Note that for any $\bm{\epsilon}^*\in A$, we have 
\[y_i=\1(X^\top_{i,\tau_0}\bm{\beta}_{0,\tau_0}+\epsilon_i^*>0),\quad \epsilon^*_i-\epsilon_i\left\{\begin{array}{cc}
    \le 0 & {\rm if~} y_i=1, \\
    \ge 0 & {\rm if~} y_i=0.
\end{array}\right.\]
Then for all $\tau\subset[p],\bm{\beta}_\tau\in\R^{|\tau|},\sigma\ge 0$,
\begin{align*}
    &\1(y_i\ne\1(X^\top_{i,\tau}\bm{\beta}_\tau+\sigma\epsilon^*_i>0))\\
    =&\1(y_i=1,X^\top_{i,\tau}\bm{\beta}_\tau+\sigma\epsilon^*_i\le 0)+\1(y_i=0,X^\top_{i,\tau}\bm{\beta}_\tau+\sigma\epsilon^*_i>0)\\
    =&\1(y_i=1,X^\top_{i,\tau}\bm{\beta}_\tau+\sigma\epsilon_i+\sigma(\epsilon^*_i-\epsilon_i)\le 0)+\1(y_i=0,X^\top_{i,\tau}\bm{\beta}_\tau+\sigma\epsilon_i+\sigma(\epsilon_i^*-\epsilon_i)>0)\\
    \ge&\1(y_i=1,X^\top_{i,\tau}\bm{\beta}_\tau+\sigma\epsilon_i\le 0)+\1(y_i=0,X^\top_{i,\tau}\bm{\beta}_\tau+\sigma\epsilon_i>0)\\
    =&\1(y_i\ne\1(X^\top_{i,\tau}\bm{\beta}_\tau+\sigma\epsilon_i>0)),
\end{align*}
then we can control term $T_1$ as
\begin{align*}
    T_1\le&\Prob(\exists\bm{\epsilon}^*\in A~{\rm s.t.}~\tau_0\ne\argmin_{\abs{\tau}\le |\tau_0|}\min_{{\bm{\beta}}_\tau\in\R^{\abs{\tau}},\sigma\ge0}L^R_n(\tau,{\bm{\beta}}_\tau,\sigma|\bm{X},\bm{y},\bm{\epsilon}^*))\\
    \le&\Prob(\exists\bm{\epsilon}^*\in A~{\rm s.t.}~\inf_{\tau\ne\tau_0,\abs{\tau}\le |\tau_0|,{\bm{\beta}}_\tau\in\R^{\abs{\tau}},\sigma\ge0}L^R_n(\tau,{\bm{\beta}}_\tau,\sigma|\bm{X},\bm{y},\bm{\epsilon}^*)\le L^R_n(\tau_0,{\bm{\beta}}_{0,\tau_0},1|\bm{X},\bm{y},\bm{\epsilon}^*))\\
    =&\Prob(\exists\bm{\epsilon}^*\in A~{\rm s.t.}~\inf_{\tau\ne\tau_0,\abs{\tau}\le |\tau_0|,{\bm{\beta}}_\tau\in\R^{\abs{\tau}},\sigma\ge0}L^R_n(\tau,{\bm{\beta}}_\tau,\sigma|\bm{X},\bm{y},\bm{\epsilon}^*)=0)\\
    \le&\Prob(\inf_{\tau\ne\tau_0,\abs{\tau}\le |\tau_0|,{\bm{\beta}}_\tau\in\R^{\abs{\tau}},\sigma\ge0}L^R_n(\tau,{\bm{\beta}}_\tau,\sigma|\bm{X},\bm{y},\bm{\epsilon})=0)\\
    \lesssim&2^{-\frac{1}{2}n\tilde c_{\min}+2\log_2p}\wedge2^{-\frac{1}{2}nc_{\min}+\log_2p},
\end{align*}
where we have used Lemma \ref{lem_oracle} in the last inequality.

For term $T_2$, denote $F_{\rm log}(z)=(1+e^{-z})^{-1}$ to be the CDF of logistic distribution, then
\begin{align*}
    T_2=&\E\big(1-\Prob(\bm{\epsilon}^*\in A\mid\bm X,\bm \epsilon)\big)^d\\
    =&\E\bigg(1-\prod_{i\in[n]}\Prob(-X^\top_{i,\tau_0}\bm{\beta}_{0,\tau_0}<\epsilon^*_i\le\epsilon_i~{\rm or}~\epsilon_i\le\epsilon_i^*\le-X^\top_{i,\tau_0}\bm{\beta}_{0,\tau_0}\mid X_i,\epsilon_i)\bigg)^d\\
    =&\E\bigg(1-\prod_{i\in[n]}\big|F_{\rm log}(\epsilon_i)-F_{\rm log}(-X^\top_{i,\tau_0}\bm{\beta}_{0,\tau_0})\big|\bigg)^d,
\end{align*}
where in the last equation, we have used the fact that $\epsilon^*$ is independent of $X,\epsilon$. Combining terms completes the proof.
\end{proof}

\begin{proof}[Proof of Theorem \ref{thm_candidate_size}]
    For any $\tau\subset[p]$, $|\tau|\le|\tau_0|$, we have
    \begin{align*}
        &\Prob(\hat\tau(\bm\epsilon^*)=\tau)\\
        \le&\Prob\big(\inf_{\bm\beta_\tau\in\R^{|\tau|},\sigma\ge 0}L_n^R(\tau,\bm\beta_\tau,\sigma|\bm X,\bm y,\bm\epsilon^*)\le\inf_{\bm\beta_{\tau_0}\in\R^{|\tau_0|}}L_n^R(\tau_0,\bm\beta_{\tau_0},0|\bm X,\bm y,\bm\epsilon^*)\big)\\
        \le&\Prob\big(\inf_{\bm\beta_\tau\in\R^{|\tau|},\sigma\ge 0}L_{\bm\theta_0}^{R*}(\tau,\bm\beta_\tau,\sigma)-\inf_{\bm\beta_{\tau_0}\in\R^{|\tau_0|}}L_{\bm\theta_0}^{R*}(\tau_0,\bm\beta_{\tau_0},0)\\
        &\qquad\le\sup_{\bm\beta_\tau\in\R^{|\tau|},\sigma\ge0}|L_n^R-\E L_n^R|(\tau,\bm\beta_\tau,\sigma|\bm X,\bm y,\bm\epsilon^*)+\sup_{\bm\beta_{\tau_0}\in\R^{|\tau_0|}}|L_n^R-\E L_n^R|(\tau_0,\bm\beta_{\tau_0},0|\bm X,\bm y,\bm\epsilon^*)\big).
    \end{align*}
    If we split the set $\mathcal{S}=\{\tau:\tau\subset[p],|\tau|\le|\tau_0|\}$ into $\mathcal{S}_0$ and $\mathcal{S}\setminus\mathcal{S}_0$ where
    \[\mathcal{S}_0=\bigg\{\tau\in\mathcal{S}:\inf_{\bm\beta_\tau\in\R^{|\tau|},\sigma\ge 0}L_{\bm\theta_0}^{R*}(\tau,\bm\beta_\tau,\sigma)-\inf_{\bm\beta_{\tau_0}\in\R^{|\tau_0|}}L_{\bm\theta_0}^{R*}(\tau_0,\bm\beta_{\tau_0},0)\le c\sqrt{\frac{|\tau_0|+1}{n}}\bigg\},\]
    we have
    \begin{align*}
        \E|\mathcal{C}|=&\sum_{\tau\in\mathcal{S}}\Prob(\tau\in\mathcal{C})\\
        =&\sum_{\tau\in\mathcal{S}}\Prob\big(\exists j\in[d]\text{ s.t. }\hat\tau(\bm\epsilon^{*(j)})=\tau\big)\\
        \le&\sum_{\tau\in\mathcal{S}}\big\{d\Prob\big(\hat\tau(\bm\epsilon^{*})=\tau\big)\big\}\wedge 1\\
        \le&|\mathcal{S}_0|+\sum_{\tau\in\mathcal{S}\setminus\mathcal{S}_0}\big\{d\Prob\big(\hat\tau(\bm\epsilon^{*})=\tau\big)\big\}\wedge 1\\
        \overset{\eqref{eq_uniform_converge}}{\le}&|\mathcal{S}_0|+\sum_{\tau\in\mathcal{S}\setminus\mathcal{S}_0}e^{-\frac{n}{8}\big\{\inf_{{\bm{\beta}}_\tau\in\R^{\abs{\tau}},\sigma\ge 0}L^{R*}_{\bm{\theta}_0}(\tau,{\bm{\beta}}_\tau,\sigma)-\inf_{\bm\beta_{\tau_0}\in\R^{|\tau_0|}}L^{R*}_{\bm{\theta}_0}(\tau_0,{\bm{\beta}}_{\tau_0},0)-c\sqrt{\frac{\abs{\tau_0}+1}{n}}\big\}^2+\log d}\wedge 1.
    \end{align*}
    
\end{proof}

\begin{proof}[Proof of Theorem \ref{thm_candidate_strong_signal}]
    For any $\bm{\epsilon}^*$ independent of the observed data, by Theorem 4.10 and Example 5.24 in \cite{wainwright2019high}, given any $\tau\subset[p]$, $|\tau|\le|\tau_0|$, we have
    \begin{equation}\label{eq_uniform_converge}
        \begin{aligned}
        &\Prob(\sup_{{\bm{\beta}}_\tau\in\R^{|\tau|},\sigma\ge0}\abs{L^R_n-\E L^R_n}(\tau,{\bm{\beta}}_\tau,\sigma|\bm{X},\bm{y},\bm{\epsilon}^*)\vee\sup_{{\bm{\beta}}_{\tau_0}\in\R^{|\tau_0|}}\abs{L^R_n-\E L^R_n}(\tau_0,{\bm{\beta}}_{\tau_0},0|\bm{X},\bm{y},\bm{\epsilon}^*)\\
        &\qquad\ge c\sqrt{\frac{\abs{\tau_0}+1}{n}}+\delta)\\
        \le&e^{-\frac{n\delta^2}{2}}.
        \end{aligned}
    \end{equation}
    Then we can control the probability of false model selection as
    \begin{align*}
        &\Prob(\hat \tau(\bm{\epsilon}^*)\ne\tau_0)\\
        \le&\Prob(\inf_{\tau\ne\tau_0,\abs{\tau}\le |\tau_0|,{\bm{\beta}}_\tau\in\R^{\abs{\tau}},\sigma\ge 0}L^R_n(\tau,{\bm{\beta}}_\tau,\sigma|\bm X,\bm{y},\bm{\epsilon}^*)\le \inf_{\bm{\beta}_{\tau_0}\in\R^{|\tau_0|}}L^R_n(\tau_0,{\bm{\beta}}_{\tau_0},0|\bm X,\bm{y},\bm{\epsilon}^*))\\
        \le&\sum_{\tau\ne\tau_0,\abs{\tau}\le |\tau_0|}\Prob(\inf_{{\bm{\beta}}_\tau\in\R^{\abs{\tau}},\sigma\ge 0}L^{R*}_{\bm{\theta}_0}(\tau,{\bm{\beta}}_\tau,\sigma)-\inf_{\bm{\beta}_{\tau_0}\in\R^{|\tau_0|}}L^{R*}_{\bm{\theta}_0}(\tau_0,{\bm{\beta}}_{\tau_0},0)\\
        &\qquad\le\sup_{{\bm{\beta}}_{\tau}\in\R^{\abs{\tau}},\sigma\ge 0}\abs{L^R_n-\E L^R_n}(\tau,{\bm{\beta}}_{\tau},\sigma|\bm X,\bm{y},\bm{\epsilon}^*)+\sup_{\bm{\beta}_{\tau_0}\in\R^{|\tau_0|}}\abs{L^R_n-\E L^R_n}(\tau_0,{\bm{\beta}}_{\tau_0},0|\bm X,\bm{y},\bm{\epsilon}^*))\\
        \le&\sum_{\tau\ne\tau_0,|\tau|\le|\tau_0|}e^{-\frac{n}{8}\big\{\inf_{{\bm{\beta}}_\tau\in\R^{\abs{\tau}},\sigma\ge 0}L^{R*}_{\bm{\theta}_0}(\tau,{\bm{\beta}}_\tau,\sigma)-\inf_{\bm\beta_{\tau_0}\in\R^{|\tau_0|}}L^{R*}_{\bm{\theta}_0}(\tau_0,{\bm{\beta}}_{\tau_0},0)-c\sqrt{\frac{\abs{\tau_0}+1}{n}}\big\}^2}\\
        =&T.
    \end{align*}
    Similar with the proof of Lemma \ref{lem_oracle}, on the one hand, if we denote $j=\abs{\tau_0\cap\tau},l=\abs{\tau\setminus\tau_0}$, then
    \begin{align*}
        T\le\sum_{j=0}^{|\tau_0|-1}\sum_{l=0}^{|\tau_0|-j}\binom{|\tau_0|}{j}\binom{p-|\tau_0|}{l}e^{-\frac{n}{8}(|\tau_0|-j)\tilde c_{\min}^*}
        \lesssim e^{-\frac{1}{8}n\tilde c_{\min}^*+2\log p}.
    \end{align*}
    On the other hand, if we denote $j=\abs{\tau}$, then
    \begin{align*}
        T\le\sum_{j=0}^{|\tau_0|}\binom{p}{j}e^{-\frac{1}{8}n(j\vee 1)c_{\min}^*}
        \lesssim e^{-\frac{1}{8}nc_{\min}^*+\log p}.
    \end{align*}
    Suppose $\mathcal{C}=\{\hat\tau(\bm{\epsilon}^{*(j)}):\epsilon_i^{*(j)}\overset{{\rm i.i.d.}}{\sim}{\rm Logistic},i\in[n],j\in[d]\}$, then
    \[\Prob(\tau_0\not\in\mathcal{C})\le\Prob(\hat\tau(\bm{\epsilon}^{*})\ne\tau_0)\lesssim e^{-\frac{n}{8}\tilde c_{\min}^*+2\log p}\wedge e^{-\frac{n}{8}c_{\min}^*+\log p}.\]
    Following the proof of Theorem \ref{thm_candidate_size}, if for all $\tau$ with $\abs{\tau}\le |\tau_0|,\tau\ne\tau_0$, 
    \[\Delta_L(\tau)\gtrsim\sqrt{\frac{\abs{\tau_0}\vee 1}{n}}+\sqrt{\abs{\tau_0\setminus\tau}\wedge(\abs{\tau}\vee1)}\sqrt{\frac{\log p}{n}},\]
    then we have $|\mathcal{S}_0|=1$ and
    \[\Prob(\mathcal{C}\ne\{\tau_0\})\le\sum_{j=1}^d\Prob(\hat\tau(\bm{\epsilon}^{*(j)})\ne\tau_0)\lesssim de^{-\frac{n}{8}\tilde c_{\min}^*+2\log p}\wedge de^{-\frac{n}{8}c_{\min}^*+\log p}\lesssim de^{-cn\tilde c_{\min}^*}\wedge de^{-cnc_{\min}^*},\]
    \[\E|\mathcal{C}|-1\le\sum_{\tau\in\mathcal{S}\setminus\{\tau_0\}}de^{-\frac{n}{8}\big\{\inf_{{\bm{\beta}}_\tau\in\R^{\abs{\tau}},\sigma\ge 0}L^{R*}_{\bm{\theta}_0}(\tau,{\bm{\beta}}_\tau,\sigma)-\inf_{\bm\beta_{\tau_0}\in\R^{|\tau_0|}}L^{R*}_{\bm{\theta}_0}(\tau_0,{\bm{\beta}}_{\tau_0},0)-c\sqrt{\frac{\abs{\tau_0}+1}{n}}\big\}^2}\lesssim de^{-cn\tilde c_{\min}^*}\wedge de^{-cnc_{\min}^*}.\]
    If we fix $d\ge 1$ a finite constant, then
    \[\Prob(\mathcal{C}\ne\{\tau_0\})\lesssim e^{-cn\tilde c_{\min}^*}\wedge e^{-cnc_{\min}^*}, \qquad\E|\mathcal{C}|-1\lesssim e^{-cn\tilde c_{\min}^*}\wedge e^{-cnc_{\min}^*}.\]
\end{proof}

\subsection{Proofs in Section \ref{sec_Abeta}}

\begin{Lemma}\label{lem_gradient}
    Under conditions in Theorem \ref{thm_Abeta}, denote $s_0=|\tau_0|$, with probability at least $1-\delta$,
    \[\|\hat\E\nabla l(\tau_0,\bm\beta_{0,\tau_0}|X,Y)\|_2\lesssim\sqrt{\frac{s_0+\log\frac{1}{\delta}}{n}}.\]
\end{Lemma}

\begin{proof}[Proof of Lemma \ref{lem_gradient}]
    Take $\mathcal{N}$ to be the $\frac{1}{2}$-net of the unit ball $\mathcal{B}$ in $\R^{s_0}$, then we have $|\mathcal{N}|\le 4^{s_0}$,
    \begin{align*}
        \|\hat\E\nabla l(\tau_0,\bm\beta_{0,\tau_0}|X,Y)\|_2=&\sup_{a\in\mathcal{B}}a^\top\hat\E\nabla l(\tau_0,\bm\beta_{0,\tau_0}|X,Y)\\
        \le&\max_{a\in\mathcal{N}}a^\top\hat\E\nabla l(\tau_0,\bm\beta_{0,\tau_0}|X,Y)+\frac{1}{2}\sup_{a\in\mathcal{B}}a^\top\hat\E\nabla l(\tau_0,\bm\beta_{0,\tau_0}|X,Y),
    \end{align*}
    therefore $\|\hat\E\nabla l(\tau_0,\bm\beta_{0,\tau_0}|X,Y)\|_2\le 2\max_{a\in\mathcal{N}}a^\top\hat\E\nabla l(\tau_0,\bm\beta_{0,\tau_0}|X,Y)$.
    Since $\frac{\eta'}{\eta}$ and $\frac{\eta'}{1-\eta}$ are bounded, it follows from the Hoeffding's inequality and the union bound that with probability at least $1-\delta$,
    \[\|\hat\E\nabla l(\tau_0,\bm\beta_{0,\tau_0}|X,Y)\|_2\lesssim \sqrt{\frac{s_0+\log\frac{1}{\delta}}{n}}.\]
\end{proof}

\begin{Lemma}\label{lem_hessian}
Under conditions in Lemma \ref{lem_gradient}, with probability at least $1-\delta$,
\[\norm{(\hat\E-\E)\nabla^2l(\tau_0,\bm\beta_{0,\tau_0}|X,Y)}_{\rm sp}\lesssim\sqrt{\frac{s_0+\log\frac{1}{\delta}}{n}}+\frac{s_0+\log\frac{1}{\delta}}{n}.\]
\end{Lemma}
\begin{proof}[Proof of Lemma \ref{lem_hessian}]
    Take $\mathcal{N}$ to be the $\frac{1}{4}$-nets of the unit ball $\mathcal{B}$ in $\R^{s_0}$. As in \cite{fan2021shrinkage}, if we denote
    \[\Phi(A)=\max_{(u,v)\in\mathcal{N}\times\mathcal{N}}u^\top Av,\]
    we have 
    \[\norm{A}_{\rm sp}\le\frac{16}{7}\Phi(A).\]
    To see this, for any $(u,v)\in\mathcal{B}\times\mathcal{B}$, there exist $(u_1,v_1)\in\mathcal{N}\times\mathcal{N}$ such that $\norm{u-u_1}_2\le\frac{1}{4},\norm{v-v_1}_2\le\frac{1}{4}$,
    \begin{align*}
        u^\top Av=&u_1^\top Av_1+(u-u_1)^\top Av_1+u_1^\top A(v-v_1)+(u-u_1)^\top A(v-v_1)\\
        \le&\Phi(A)+(\frac{1}{4}+\frac{1}{4}+\frac{1}{16})\norm{A}_{\rm sp}.
    \end{align*}
    Taking supremum on both sides yields the result.
    
    Fix any $(u,v)\in\mathcal{N}\times\mathcal{N}$, we know $\nabla^2 l(\tau_0,\bm\beta_{0,\tau_0}|X,Y)$ is sub-exponential.
    By Bernstein's inequality, with probability at least $1-\delta$,
    \[(\hat\E-\E)u^\top\nabla^2l(\tau_0,\bm\beta_{0,\tau_0}|X,Y)v\lesssim\sqrt{\frac{\log\frac{1}{\delta}}{n}}+\frac{\log\frac{1}{\delta}}{n}.\]
    Applying union bound over $(u,v)\in\mathcal{N}\times\mathcal{N}$, we have with probability at least $1-\delta$,
    \[\|(\hat\E-\E)\nabla^2l(\tau_0,\bm\beta_{0,\tau_0}|X,Y)\|_{\rm sp}\lesssim\sqrt{\frac{s_0+\log\frac{1}{\delta}}{n}}+\frac{s_0+\log\frac{1}{\delta}}{n}.\]
\end{proof}

\begin{Lemma}\label{lem_tensor}

Under conditions in Lemma \ref{lem_gradient}, denote $\mathcal{B}=\{a\in\R^{s_0}:\norm{a}_2=1\}$ to be the unit sphere in $\R^{s_0}$, then with probability at least $1-\delta$,
\[\sup_{a,b,c\in\mathcal{B}}\frac{1}{n}\sum_{i=1}^n\abs{a^\top X_{i,\tau_0} b^\top X_{i,\tau_0} c^\top X_{i,\tau_0}}\lesssim 1+\sqrt{\frac{s_0+\log\frac{1}{\delta}}{n}}+\frac{(s_0\log n+\log\frac{n}{\delta})^{\frac{3}{2}}}{n}.\]
\end{Lemma}
\begin{proof}[Proof of Lemma \ref{lem_tensor}]
    Note that for any $a,b,c\in\mathcal{B}$, we have
    \[\norm{a^\top X_{\tau_0} b^\top X_{\tau_0} c^\top X_{\tau_0}}_{\psi_{2/3}}\lesssim 1,\quad\norm{\big(a^\top X_{\tau_0} b^\top X_{\tau_0} c^\top X_{\tau_0}\big)^2}_{\psi_{1/3}}\lesssim 1.\]
    Denote $\mathcal{N}$ to be the $\frac{1}{4}$-net of $\mathcal{B}$, then $|\mathcal{N}|\le 8^{s_0}$. For any $a,b,c\in\mathcal{B}$, there exist $\tilde a,\tilde b,\tilde c\in\mathcal{N}$ such that $\|a-\tilde a\|_2,\|b-\tilde b\|_2,\|c-\tilde c\|_2\le\frac{1}{4}$, and
    \begin{align*}
        &\frac{1}{n}\sum_{i=1}^n|a^\top X_{i,\tau_0}b^\top X_{i,\tau_0}c^\top X_{i,\tau_0}|\\
        \le&\frac{1}{n}\sum_{i=1}^n|\tilde a^\top X_{i,\tau_0}\tilde b^\top X_{i,\tau_0}\tilde c^\top X_{i,\tau_0}|+\frac{3}{4}\sup_{a,b,c\in\mathcal{B}}\frac{1}{n}\sum_{i=1}^n|a^\top X_{i,\tau_0}b^\top X_{i,\tau_0}c^\top X_{i,\tau_0}|.
    \end{align*}
    Taking supremum over $a,b,c\in\mathcal{B}$, we get
    \[\sup_{a,b,c\in\mathcal{B}}\frac{1}{n}\sum_{i=1}^n|a^\top X_{i,\tau_0}b^\top X_{i,\tau_0}c^\top X_{i,\tau_0}|\le4\max_{a,b,c\in\mathcal{N}}\frac{1}{n}\sum_{i=1}^n|a^\top X_{i,\tau_0}b^\top X_{i,\tau_0}c^\top X_{i,\tau_0}|.\]
    By Theorem 3.4 in \cite{kuchibhotla2022moving}, we have with probability at least $1-\delta$,
    \begin{align*}
        &\max_{a,b,c\in\mathcal{N}}\frac{1}{n}\sum_{i=1}^n\bigg\{|a^\top X_{i,\tau_0}b^\top X_{i,\tau_0}c^\top X_{i,\tau_0}|-\E|a^\top X_{i,\tau_0}b^\top X_{i,\tau_0}c^\top X_{i,\tau_0}|\bigg\}\lesssim\sqrt{\frac{s_0+\log\frac{1}{\delta}}{n}}+\frac{(s_0\log n+\log\frac{n}{\delta})^{\frac{3}{2}}}{n}.
    \end{align*}
    Then
    \begin{align*}
        &\sup_{a,b,c\in\mathcal{B}}\frac{1}{n}\sum_{i=1}^n\abs{a^\top X_{i,\tau_0} b^\top X_{i,\tau_0} c^\top X_{i,\tau_0}}\\
        \le&4\max_{a,b,c\in\mathcal{N}}\frac{1}{n}\sum_{i=1}^n|a^\top X_{i,\tau_0}b^\top X_{i,\tau_0}c^\top X_{i,\tau_0}|\\
        \le&4\max_{a,b,c\in\mathcal{N}}\E|a^\top X_{\tau_0} b^\top X_{\tau_0} c^\top X_{\tau_0}|+\max_{a,b,c\in\mathcal{N}}\frac{4}{n}\sum_{i=1}^n\bigg\{|a^\top X_{i,\tau_0}b^\top X_{i,\tau_0}c^\top X_{i,\tau_0}|-\E|a^\top X_{i,\tau_0}b^\top X_{i,\tau_0}c^\top X_{i,\tau_0}|\bigg\}\\
        \lesssim&1+\sqrt{\frac{s_0+\log\frac{1}{\delta}}{n}}+\frac{(s_0\log n+\log\frac{n}{\delta})^{\frac{3}{2}}}{n}.
    \end{align*}
\end{proof}

\begin{proof}[Proof of Theorem \ref{thm_Abeta}]
    Given $\tau_0$, we start by proving $\hat{\bm{\beta}}_{\tau_0}$ is consistent for $\bm{\beta}_{0,\tau_0}$, where
    \begin{equation}\label{eq_beta_hat}
        \hat{\bm{\beta}}_{\tau_0}=\argmax_{\bm{\beta}_{\tau_0}\in\R^{|\tau_0|}}\hat\E l(\tau_0,\bm{\beta}_{\tau_0}|X,Y),\quad \bm{\beta}_{0,\tau_0}=\argmax_{\bm{\beta}_{\tau_0}\in\R^{|\tau_0|}}\E l(\tau_0,\bm{\beta}_{\tau_0}|X,Y),
    \end{equation}
    \[l(\tau_0,\bm{\beta}_{\tau_0}|X,Y)=Y\log\frac{\eta(X_{\tau_0}^\top\bm{\beta}_{\tau_0})}{1-\eta(X_{\tau_0}^\top\bm{\beta}_{\tau_0})}+\log(1-\eta(X_{\tau_0}^\top\bm{\beta}_{\tau_0})),\quad \eta(\cdot)=g^{-1}(\cdot).\]
    Note that
    \[\nabla l(\tau_0,\bm{\beta}_{\tau_0}|X,Y)\overset{\triangle}{=}\frac{\partial}{\partial{\bm{\beta}_{\tau_0}}}l(\tau_0,\bm{\beta}_{\tau_0}|X,Y)=\frac{\eta'(X_{\tau_0}^\top\bm{\beta}_{\tau_0})}{\eta(X_{\tau_0}^\top\bm{\beta}_{\tau_0})}YX_{\tau_0}+\frac{\eta'(X_{\tau_0}^\top\bm{\beta}_{\tau_0})}{1-\eta(X_{\tau_0}^\top\bm{\beta}_{\tau_0})}(Y-1)X_{\tau_0},\]
    \begin{align*}
        \nabla^2 l(\tau_0,\bm{\beta}_{\tau_0}|X,Y)\overset{\triangle}{=}&\frac{\partial^2}{\partial\bm{\beta}_{\tau_0}\partial\bm{\beta}^\top_{\tau_0}}l(\tau_0,\bm{\beta}_{\tau_0}|X,Y)\\
        =&\bigg\{\frac{\eta''(X_{\tau_0}^\top\bm{\beta}_{\tau_0})}{\eta(X_{\tau_0}^\top\bm{\beta}_{\tau_0})}-\bigg(\frac{\eta'(X_{\tau_0}^\top\bm{\beta}_{\tau_0})}{\eta(X_{\tau_0}^\top\bm{\beta}_{\tau_0})}\bigg)^2\bigg\}YX_{\tau_0}X_{\tau_0}^\top\\
        &+\bigg\{\frac{\eta''(X_{\tau_0}^\top\bm{\beta}_{\tau_0})}{1-\eta(X_{\tau_0}^\top\bm{\beta}_{\tau_0})}+\bigg(\frac{\eta'(X_{\tau_0}^\top\bm{\beta}_{\tau_0})}{1-\eta(X_{\tau_0}^\top\bm{\beta}_{\tau_0})}\bigg)^2\bigg\}(Y-1)X_{\tau_0}X_{\tau_0}^\top.
    \end{align*}
    Denote 
    \[h_1(z)\overset{\triangle}{=}\frac{\eta''(z)}{\eta(z)}-\bigg(\frac{\eta'(z)}{\eta(z)}\bigg)^2,\quad h_0(z)\overset{\triangle}{=}\frac{\eta''(z)}{1-\eta(z)}+\bigg(\frac{\eta'(z)}{1-\eta(z)}\bigg)^2,\]
    we assume
    \begin{equation}\label{eq_loglike_condition}
        \bigg\|\frac{\eta'}{\eta}\bigg\|_\infty+\bigg\|\frac{\eta'}{1-\eta}\bigg\|_\infty+\|h_1\|_\infty+\|h_0\|_\infty\lesssim 1,\quad h_1<0<h_0,
    \end{equation}
    which implies $\|a^\top\nabla l(\bm{\beta}_{\tau_0};X_{\tau_0},Y)\|_{\psi_2}+\|a^\top\nabla^2l(\bm\beta_{\tau_0};X_{\tau_0},Y)b\|_{\psi_1}\lesssim 1$ and $l$ is concave in $\bm{\beta}_{\tau_0}$. We also suppose $h_1$ and $h_0$ to be Lipschitz. In the rest of the proof, we omit the arguments $(\tau_0,X_{\tau_0},Y)$ and abbreviate $l(\tau_0,\bm{\beta}_{\tau_0}|X,Y)$ to $l(\bm{\beta}_{\tau_0})$ when there is no confusion. For any $\bm{\beta}_{\tau_0}$ such that $\bm\Delta=\bm{\beta}_{\tau_0}-\bm{\beta}_{0,\tau_0}$ satisfies $\|\bm\Delta\|_2=c\sqrt{\frac{s}{n}}$, we have for some $\tilde{\bm\beta}_{\tau_0}$ between $\bm\beta_{\tau_0}$ and $\bm\beta_{0,\tau_0}$,
    \begin{equation}\label{eq_likelihood_comparison}
        \begin{aligned}
            &\hat \E l(\bm{\beta}_{\tau_0})-l(\bm{\beta}_{0,\tau_0})\\
        =&\hat\E\nabla^\top l(\bm{\beta}_{0,\tau_0})\bm\Delta+\frac{1}{2}\bm\Delta^\top\hat\E\nabla^2l(\tilde{\bm{\beta}}_{\tau_0})\bm\Delta\\
        \le&\underbrace{\|\hat \E\nabla^\top l(\bm{\beta}_{0,\tau_0})\|_2c\sqrt{\frac{s}{n}}}_{T_1}-\underbrace{\frac{1}{2}\lambda_{\min}\big(-\hat\E\nabla^2l(\bm{\beta}_{0,\tau_0})\big)c^2\frac{s}{n}}_{T_2}+\underbrace{\frac{1}{2}\bm\Delta^\top\hat\E\big(\nabla^2l(\tilde{\bm{\beta}}_{\tau_0})-\nabla^2l(\bm{\beta}_{0,\tau_0})\big)\bm\Delta}_{T_3}.
        \end{aligned}
    \end{equation}
    Since $\nabla l(\bm\beta_{0,\tau_0})$ is sub-Gaussian and centered, it follows from Lemma \ref{lem_gradient} that with high probability,
    \[T_1\lesssim c\frac{s}{n}.\]
    By Lemma \ref{lem_hessian}, $T_2$ can be controlled with high probability that
    \[T_2\gtrsim c^2\frac{s}{n}.\]
    For $T_3$, with high probability,
    \begin{align*}
        &\sup_{\|\bm\Delta\|_2=c\sqrt{\frac{s}{n}}}\bm\Delta^\top\hat\E\big(\nabla^2l(\tilde{\bm{\beta}}_{\tau_0})-\nabla^2l(\bm{\beta}_{0,\tau_0})\big)\bm\Delta\\
        =&\sup_{\|\bm\Delta\|_2=c\sqrt{\frac{s}{n}}}\hat\E\bigg\{Y\bigg(h_1(X_{\tau_0}^\top\tilde{\bm\beta}_{\tau_0})-h_1(X_{\tau_0}^\top\bm\beta_{0,\tau_0})\bigg)+(Y-1)\bigg(h_0(X_{\tau_0}^\top\tilde{\bm\beta}_{\tau_0})-h_0(X_{\tau_0}^\top\bm\beta_{0,\tau_0})\bigg)\bigg\}\big(X_{\tau_0}^\top\bm\Delta\big)^2\\
        \lesssim&\sup_{\|\bm\Delta\|_2=c\sqrt{\frac{s}{n}}}\hat\E|X_{\tau_0}^\top\bm\Delta|^3\\
        \overset{\text{Lemma }\ref{lem_tensor}}{\lesssim}&c^3\frac{s^{3/2}}{n^{3/2}}.
    \end{align*}
    Therefore, with high probability,
    \[\hat\E l(\bm\beta_{\tau_0})-l(\bm\beta_{0,\tau_0})\lesssim \bigg(c-c^2+c^3\sqrt{\frac{s}{n}}\bigg)\frac{s}{n},\quad\forall\bm\beta_{\tau_0}\text{ s.t. }\|\bm\beta_{\tau_0}-\bm\beta_{0,\tau_0}\|_2=c\sqrt{\frac{s}{n}}.\]
    Since $n\gg s$, choosing $c\asymp 1$ large enough ensures that with high probability,
    \[\hat\E l(\bm\beta_{\tau_0})<\hat\E l(\bm\beta_{0,\tau_0}),\quad\forall\bm\beta_{\tau_0}\text{ s.t. }\|\bm\beta_{\tau_0}-\bm\beta_{0,\tau_0}\|_2=c\sqrt{\frac{s}{n}}.\]
    Since $l$ is concave, it follows that
    \[\|\hat{\bm\beta}_{\tau_0}-\bm\beta_{0,\tau_0}\|_2=O_P\bigg(\sqrt{\frac{s}{n}}\bigg).\]

    In the remaining proof, we abbreviate $D(\tau_0),r(\tau_0)$ to $D,r$, respectively. Then we study the asymptotic distribution of $D\hat{\bm\beta}_{\tau_0}$. To this end, we utilize the first-order optimality condition of \eqref{eq_beta_hat},
    \begin{equation}\label{eq_first_order}
        \begin{aligned}
        0=&\hat\E\nabla l(\hat{\bm\beta}_{\tau_0})\\
        =&\hat\E\nabla l(\hat{\bm\beta}_{\tau_0})-\hat\E\nabla l(\bm\beta_{0,\tau_0})+(\hat\E-\E)\nabla l(\bm\beta_{0,\tau_0})\\
        =&\hat\E\nabla^2l(\bm\beta_{0,\tau_0})(\hat{\bm\beta}_{\tau_0}-\bm\beta_{0,\tau_0})+R_1+(\hat\E-\E)\nabla l(\bm\beta_{0,\tau_0})\\
        =&\{\E\nabla^2l(\bm\beta_{0,\tau_0})\}(\hat{\bm\beta}_{\tau_0}-\bm\beta_{0,\tau_0})+\{(\hat\E-\E)\nabla^2l(\bm\beta_{0,\tau_0})\}(\hat{\bm\beta}_{\tau_0}-\bm\beta_{0,\tau_0})+R_1+(\hat\E-\E)\nabla l(\bm\beta_{0,\tau_0}).
    \end{aligned}
    \end{equation} 
    Therefore
    \begin{align*}
        D\hat{\bm\beta}_{\tau_0}-D\bm\beta_{0,\tau_0}=-DH^{-1}(\hat\E-\E)\nabla l(\bm\beta_{0,\tau_0})-\underbrace{DH^{-1}\{(\hat\E-\E)\nabla^2l(\bm\beta_{0,\tau_0})\}(\hat{\bm\beta}_{\tau_0}-\bm\beta_{0,\tau_0})}_{T_4}-\underbrace{DH^{-1}R_1}_{T_5}.
    \end{align*}
    \[\|T_4\|_2\le\|DH^{-1}(\hat\E-\E)\nabla^2l(\bm\beta_{0,\tau_0})\|_{\rm sp}\|\hat{\bm\beta}_{\tau_0}-\bm\beta_{0,\tau_0}\|_2\overset{\text{Lemma }\ref{lem_hessian}}{=}O_P\bigg(\frac{s}{n}\bigg).\]
    For the Taylor expansion, for any $a\in\R^r$, there exists a vector $\tilde{\bm\beta}_{\tau_0}^a$ between $\hat{\bm\beta}_{\tau_0}$ and $\bm\beta_{0,\tau_0}$ such that
    \begin{align*}
        \|T_5\|_2=&\sup_{a\in\R^r,\|a\|_2\le 1}a^\top DH^{-1}\bigg\{\hat\E\nabla l(\hat{\bm\beta}_{\tau_0})-\hat\E\nabla l(\bm\beta_{0,\tau_0})-\hat\E\nabla^2 l(\bm\beta_{0,\tau_0})(\hat{\bm\beta}_{\tau_0}-\bm\beta_{0,\tau_0})\bigg\}\\
        =&\sup_{a\in\R^r,\|a\|_2\le 1}a^\top DH^{-1}\bigg\{\hat\E\nabla^2l(\tilde{\bm\beta}_{\tau_0}^a)-\hat\E\nabla^2l(\bm\beta_{0,\tau_0})\bigg\}(\hat{\bm\beta}_{\tau_0}-\bm\beta_{0,\tau_0})\\
        =&\sup_{a\in\R^r,\|a\|_2\le 1}a^\top DH^{-1}\hat\E\bigg\{Y\bigg(h_1(X_{\tau_0}^\top\tilde{\bm\beta}_{\tau_0}^a)-h_1(X_{\tau_0}^\top\bm\beta_{0,\tau_0})\bigg)\\
        &\qquad+(Y-1)\bigg(h_0(X_{\tau_0}^\top\tilde{\bm\beta}_{\tau_0}^a)-h_0(X_{\tau_0}^\top\bm\beta_{0,\tau_0})\bigg)\bigg\}X_{\tau_0}X_{\tau_0}^\top(\hat{\bm\beta}_{\tau_0}-\bm\beta_{0,\tau_0})\\
        \lesssim&\sup_{a\in\R^r,\|a\|_2\le 1}\hat\E \big|a^\top DH^{-1}X_{\tau_0}\big|\bigg(X_{\tau_0}^\top(\hat{\bm\beta}_{\tau_0}-\bm\beta_{0,\tau_0})\bigg)^2\\
        \overset{\text{Lemma }\ref{lem_tensor}}{=}&O_P\bigg(\frac{s}{n}\bigg).
    \end{align*}
    Therefore we have
    \[\sqrt{n}(D\hat{\bm\beta}_{\tau_0}-D\bm\beta_{0,\tau_0})=-\sqrt{n}DH^{-1}(\hat\E-\E)\nabla l(\bm\beta_{0,\tau_0})+R_2,\quad\|R_2\|_2=O_P\bigg(\frac{s}{\sqrt{n}}\bigg).\]

    Denote the variance estimator $\hat V$ to be
    \[\hat V=D\hat H^{-1}\widehat{\Cov}(\nabla l(\hat{\bm\beta}_{\tau_0}))\hat H^{-1}D^\top,\]
    it suffices to study the Gaussian approximation of $\sqrt{n}\hat V^{-\frac{1}{2}}D(\hat{\bm\beta}_{\tau_0}-\bm\beta_{0,\tau_0})$. To approach this, we study its error decomposition.
    \begin{align*}
        \sqrt{n}\hat V^{-\frac{1}{2}}D(\hat{\bm\beta}_{\tau_0}-\bm\beta_{0,\tau_0})=&-\sqrt{n}V^{-\frac{1}{2}}DH^{-1}(\hat\E-\E)\nabla l(\bm\beta_{0,\tau_0})\\
        &+\underbrace{\sqrt{n}(V^{-\frac{1}{2}}-\hat V^{-\frac{1}{2}})DH^{-1}(\hat\E-\E)\nabla l(\bm\beta_{0,\tau_0})}_{T_6}+\hat V^{-\frac{1}{2}}R_2.
    \end{align*}
    It suffices to study the spectral norm of $\hat V^{-\frac{1}{2}}-V^{-\frac{1}{2}}$. Since $\hat V^{-\frac{1}{2}}-V^{-\frac{1}{2}}=\hat V^{-\frac{1}{2}}(V^{\frac{1}{2}}-\hat V^{\frac{1}{2}})V^{-\frac{1}{2}}$, we start from $\|\hat V^{\frac{1}{2}}-V^{\frac{1}{2}}\|_{\rm sp}$. It follows from \cite{schmitt1992perturbation} that $\|\hat V^{\frac{1}{2}}-V^{\frac{1}{2}}\|_{\rm sp}\le\|\hat V-V\|_{\rm sp}/(\lambda_{\min}^{\frac{1}{2}}(\hat V)+\lambda_{\min}^{\frac{1}{2}}(V))$. Similar to \eqref{eq_likelihood_comparison}, we know
    \[\|\hat H-H\|_{\rm sp}=O_P\bigg(\sqrt{\frac{s}{n}}\bigg),\quad \|\widehat{\Cov}(\nabla l(\hat{\bm\beta}_{\tau_0}))-\Cov(\nabla l(\bm\beta_{0,\tau_0}))\|_{\rm sp}=O_P\bigg(\sqrt{\frac{s}{n}}\bigg),\]
    thus $\|\hat V-V\|_{\rm sp}=O_P(\sqrt{\frac{s}{n}})$, which implies
    \[\quad \|\hat V^{-\frac{1}{2}}-V^{-\frac{1}{2}}\|_{\rm sp}=O_P\bigg(\sqrt{\frac{s}{n}}\bigg).\]
    Then we have the decomposition
    \[\sqrt{n}\hat V^{-\frac{1}{2}}D(\hat{\bm\beta}_{\tau_0}-\bm\beta_{0,\tau_0})=\underbrace{-\sqrt{n}V^{-\frac{1}{2}}DH^{-1}\hat\E\nabla l(\bm\beta_{0,\tau_0})}_{\hat G}+R_3,\quad\|R_3\|_2=O_P\bigg(\frac{s}{\sqrt{n}}\bigg).\]
    Denote $Z_i= -V^{-\frac{1}{2}}DH^{-1}\nabla l(\tau_0,\bm\beta_{0,\tau_0}|X_{i},Y_i)$, we have $\Cov(Z_i)=I_r$ and $\|Z_i\|_{\psi_2}\lesssim 1$, therefore, for any $j_1,j_2,j_3,j_4\in[r]$, the fourth moment exists, $\E|Z_{i,j_1}Z_{i,j_2}Z_{i,j_3}Z_{i,j_4}|<\infty$. Denote $Z_i^{\otimes3}=Z_i\otimes Z_i\otimes Z_i$ to be tensor in $\R^{r^{\otimes 3}}$, by Corollary 4.10 in \cite{wang2017operator},
\begin{align*}
    \norm{\E Z_i^{\otimes 3}}_{\rm F}\le r\norm{\E Z_i^{\otimes 3}}_{\rm sp}=r\sup_{a,b,c\in\mathcal{B}_r}\E a^\top Z_i b^\top Z_i c^\top Z_i\lesssim r.
\end{align*}
By Lemma 1 in \cite{jin2019short}, we know $\norm{\norm{Z_i}_2}_{\psi_2}\lesssim \sqrt{r}$, then $\E\norm{Z_i}_2^4\lesssim r^2$.
    Then Corollary 1 in \cite{zhilova2022new} implies that for $G\sim N(0,I_r)$, 
    \[\sup_{t>0}|\Prob(\|\hat G\|_2\le t)-\Prob(\|G\|_2\le t)|\lesssim\frac{r}{\sqrt{n}}.\]
    Then
    \begin{align*}
        &\sup_{t>0}\Prob(\|\hat G+R_3\|_2\le t)-\Prob(\|G\|_2\le t)\\
        \le&\sup_{t>0}\Prob(\|\hat G\|_2\le t+\delta)+\Prob(\|R_3\|_2\ge \delta)-\Prob(\|G\|_2\le t+\delta)+\Prob(t<\|G\|_2\le t+\delta)\\
        \rightarrow& 0,
    \end{align*}
    where we let $n\rightarrow\infty$ at first and then $\delta\rightarrow 0$. Similarly,
    \[\sup_{t>0}\Prob(\|G\|_2\le t)-\Prob(\|\hat G+R_3\|_2\le t)\rightarrow 0.\]
    Combining pieces concludes the proof.

\end{proof}

\subsection{Proofs in Section \ref{sec_tau}}

To prove Theorem \ref{thm_tau}, we need the following lemma from \cite{berend2012convergence}.
\begin{Lemma}[Theorem 2 in \cite{berend2012convergence}]\label{lem_TV_tail}
For every $k\in\mathbb{N}$, distribution $\Prob$ with support on $k$ elements and its empirical version $\hat\Prob$ based on $m$ samples, we have
\[\Prob({\rm TV}(\Prob,\hat\Prob)>\delta)\le \exp\bigg(-\frac{m}{2}\bigg(2\delta-\sqrt{\frac{k}{m}}\bigg)^2\bigg),\quad2\delta\ge\sqrt{\frac{k}{m}}.\]
\end{Lemma}
\begin{proof}[Proof of Theorem \ref{thm_tau}]

    Denote $\bar T(\bm X,\bm y,\tau,\bm\beta)=\frac{1}{m}\big|\{j\in[m]:\hat P_{\tau,\bm \beta_\tau}(\tilde\tau^{(j)})>\hat P_{\tau,\bm\beta_\tau}(\tilde\tau(\bm X,\bm y,\tau,\bm\beta))\}\big|$, we have
    \begin{align*}
        &\Prob(\tau_0\not\in\Gamma_{\alpha}^{\tau_0}(\bm{X},\bm{y}))\\
        \le&\Prob(\tau_0\not\in\Gamma_{\alpha}^{\tau_0}(\bm{X},\bm{y}),\tau_0\in\mathcal{C})+\Prob(\tau_0\not\in\mathcal{C})\\
        \le&\Prob(\bar T(\bm{X},\bm{y},\tau_0,\bm \beta_0)\ge1-\alpha)+\Prob(\tau_0\not\in\mathcal{C})\\
        =&T_1+T_2.
    \end{align*}
    Denote $\Prob_\tau,\hat\Prob_\tau$ to be the probabilities with respect to $\tau$ with distribution $P_{\tau_0,{\bm{\beta}}_{0,\tau_0}},\hat P_{\tau_0,{\bm{\beta}}_{0,\tau_0}}$, respectively. And for simplicity, we also compress $P_{\tau_0,{\bm{\beta}}_{0,\tau_0}},\hat P_{\tau_0,{\bm{\beta}}_{0,\tau_0}}$ as $P,\hat P$, respectively. Then $\bar T(\bm{X}^{obs},\bm{y}^{obs},\tau_0,\bm\beta_0)$ can be expressed as
    \[\bar T(\bm{X}^{obs},\bm{y}^{obs},\tau_0,\bm\beta_0)=\hat\Prob_{\tau}(\hat P(\tau)>\hat P(\tilde\tau(\bm{X}^{obs},\bm{y}^{obs},\tau_0,\bm\beta_0))).\]
    Since $\tilde\tau(\bm{X}^{obs},\bm{y}^{obs},\tau_0,\bm\beta_0)$ is a realization from $P_{\tau_0,{\bm{\beta}}_{0,\tau_0}}$, then we can rewrite $T_1$ as
    \[T_1=\E\Prob_{\tilde\tau}(\hat\Prob_{\tau}(\hat P(\tau)>\hat P(\tilde\tau))\ge1-\alpha),\]
    where $\E$ includes the randomness of both the Monte Carlo samples $\{\bm{\epsilon}^{*(j)}:j\in[m]\}$ and the covariates $\bm{X}$. 
    
    On the one hand, by the property of CDF, we know
    \begin{align*}
        \hat\Prob_{\tilde\tau}(\hat\Prob_\tau(\hat P(\tau)>\hat P(\tilde\tau))\ge1-\alpha)
        =\hat\Prob_{\tilde\tau}(\hat\Prob_{\tau}(\hat P(\tau)\le\hat P(\tilde\tau))\le\alpha)
        \le\alpha\quad{\rm a.s.}.
    \end{align*}
    Then
    \begin{align*}
        T_1=&\E\Prob_{\tilde\tau}(\hat\Prob_{\tau}(\hat P(\tau)>\hat P(\tilde\tau))\ge1-\alpha)-\E\hat\Prob_{\tilde\tau}(\hat\Prob_\tau(\hat P(\tau)>\hat P(\tilde\tau))\ge1-\alpha)
        +\E\hat\Prob_{\tilde\tau}(\hat\Prob_\tau(\hat P(\tau)>\hat P(\tilde\tau))\ge1-\alpha)\\
        \le&\E{\rm TV(\Prob_\tau,\hat\Prob_\tau)}+\alpha.
    \end{align*}
    By construction, $P$ and $\hat P$ have supports on at most $\sum_{j=0}^s{p\choose j}\le (\frac{ep}{s})^s$ elements, then we can control the expected total variation term using Lemma \ref{lem_TV_tail} as
    \begin{align*}
        &\E{\rm TV}(\Prob,\hat\Prob)\\
        =&\E{\rm TV}(\Prob_\tau,\hat\Prob_\tau)\1\bigg\{{\rm TV}(\Prob_\tau,\hat\Prob_\tau)\le\sqrt{\frac{(\frac{ep}{s})^s}{4m}}\bigg\}+\E{\rm TV}(\Prob_\tau,\hat\Prob_\tau)\1\bigg\{{\rm TV}(\Prob_\tau,\hat\Prob_\tau)>\sqrt{\frac{(\frac{ep}{s})^s}{4m}}\bigg\}\\
        \le&\sqrt{\frac{(\frac{ep}{s})^s}{4m}}+\int_{\sqrt{\frac{(\frac{ep}{s})^s}{4m}}}^\infty\Prob({\rm TV}(\Prob_\tau,\hat\Prob_\tau)>t)dt\\
        \le&\sqrt{\frac{(\frac{ep}{s})^s}{4m}}+\int_{\sqrt{\frac{(\frac{ep}{s})^s}{4m}}}^\infty e^{-\frac{m}{2}(2t-\sqrt{\frac{(\frac{ep}{s})^s}{m}})^2}dt\\
        \le&\sqrt{\frac{(\frac{ep}{s})^s}{4m}}+\sqrt{\frac{\pi}{8m}}.
    \end{align*}

    On the other hand, 
    \[\bar T(\bm X,\bm y,\tau_0,\bm\beta_0)=\frac{1}{m}\sum_{j\in[m]}\1\big\{\hat P(\tilde\tau^{(j)})>\hat P(\tilde\tau^{(0)})\big\}\le\frac{1}{m}\sum_{j\in[m]}\1\big\{\tilde\tau^{(j)}\ne \tilde\tau^{(0)}\big\}.\]
    Then, by Markov's inequality,
    \[T_1\le\Prob\bigg(\frac{1}{m}\sum_{j\in[m]}\1\{\tilde\tau^{(j)}\ne \tilde\tau^{(0)}\}\ge1-\alpha\bigg)\le\frac{1}{m(1-\alpha)}\sum_{j\in[m]}\Prob(\tilde\tau^{(j)}\ne \tilde\tau^{(0)})=\frac{\delta_n}{1-\alpha}.\]

    Therefore
    \[\Prob(\tau_0\not\in\Gamma_{\alpha}^{\tau_0}(\bm{X},\bm{y}))\le\min\bigg\{\alpha+\sqrt{\frac{(\frac{ep}{s})^s}{4m}}+\sqrt{\frac{\pi}{8m}},\frac{\delta_n}{1-\alpha}\bigg\}+\Prob(\tau_0\not\in\mathcal{C}).\]
    Together with Theorem \ref{thm_candidate} and Theorem \ref{thm_candidate_strong_signal}, we conclude the proof.
\end{proof}

\end{document}